\documentclass{article}
\usepackage[a4paper,top=3cm,bottom=3cm,left=3cm,right=3cm,marginparwidth=2cm]{geometry}

\usepackage[english]{babel}
\usepackage[utf8]{inputenc}
\usepackage[T1]{fontenc}

\usepackage{mathtools,amsmath,amsthm,amsfonts,amssymb,bm}
\usepackage{graphicx}
\usepackage[colorlinks=true, allcolors=blue]{hyperref}
\usepackage{dsfont}
\usepackage{mathrsfs}
\usepackage{csquotes}
\usepackage[ruled,vlined]{algorithm2e}
\usepackage{caption}
\usepackage{multirow}
\usepackage{array, booktabs}
\usepackage{enumitem}
\usepackage{cleveref}
\usepackage{cases}
\usepackage{subcaption}
\usepackage{diagbox}
\usepackage{arydshln}

\newsavebox{\accentbox}

\theoremstyle{definition}
\newtheorem{definition}{Definition}[section]

\theoremstyle{plain}

\newtheorem{proposition}[definition]{Proposition}
\newtheorem{corollary}[definition]{Corollary}
\newtheorem{lemma}{Lemma}[section]
\newtheorem{remark}{Remark}[section]
\newtheorem{assumption}{Assumption}

\def\EE{{\mathbb E}}

\def\RR{{\mathbb R}}

\def\II{{\mathds 1}}

\def\dd{{\mathrm d}}

\newcommand{\Cov}{\operatorname{cov}}
\DeclareMathOperator{\diag}{diag}
\DeclareMathOperator{\Tr}{tr}
\newcommand{\iu}{\mathrm{i}}

\newcommand{\ie}{i.e.\xspace}

\usepackage{natbib}
\begin{document}
    \title{Spectral analysis for the inference of noisy Hawkes processes}

    \author{
      Anna Bonnet\thanks{Sorbonne Université and Université Paris Cité, CNRS, Laboratoire de Probabilités, Statistique et Modélisation, F-75005 Paris, France},
      Félix Cheysson\thanks{Gustave Eiffel University, CNRS, Laboratoire d'Analyse et de Mathématiques Appliquées, 77420 Champs-sur-Marne, France},
      Miguel Martinez Herrera\footnotemark[1],
      Maxime Sangnier\footnotemark[1]
    }
    
    \maketitle

    \begin{abstract}

        Classic estimation methods for Hawkes processes rely on the assumption that observed
        event times are indeed a realisation of a Hawkes process, without considering any potential
        perturbation of the model. However, in practice, observations are often altered by some noise,
        the form of which depends on the context. It is then required to model the alteration mechanism in order to infer accurately such a noisy Hawkes process. While several models exist, we
        consider, in this work, the observations to be the indistinguishable union of event times coming
        from a Hawkes process and from an independent Poisson process. Since standard inference
        methods (such as maximum likelihood or Expectation-Maximisation) are either unworkable
        or numerically prohibitive in this context, we propose an estimation procedure based on the
        spectral analysis of second order properties of the noisy Hawkes process. Novel results include sufficient conditions for identifiability of the ensuing statistical model with exponential
        interaction functions for both univariate and bivariate processes, along with consistency and asymptotic normality guarantees of our estimator in the univariate case. Although we mainly focus
        on the exponential scenario, other types of kernels are investigated and discussed. A new
        estimator based on maximising the spectral log-likelihood is then described, and its behaviour
        is numerically illustrated on both synthetic data and neuronal data. Besides being free from knowing the source of
        each observed time (Hawkes or Poisson process), the proposed estimator is shown to perform
        accurately in estimating both processes.

    \end{abstract}

\section{Introduction}
     
       Hawkes processes, introduced in \cite{Hawkes1971}, are a class of point processes that have been originally used to model self-exciting phenomena and more recently other types of past-dependent behaviours.
   Their fields of applications are wide and include for instance seismology \citep{Ogata1988, Ogata1998}, neuroscience \citep{Chornoboy1988, Lambert2018}, criminology \citep{Olinde2020}, finance \citep{Embrechts2011, Bacry2015} and biology \citep{Gupta2018}, to mention a few.
         Consequently, there has been a deep focus on estimation techniques for Hawkes processes. Among them, let us mention maximum likelihood approaches \citep{Ogata1978, Ozaki1979, Guo2018}, methods of moments \citep{DaFonseca2013}, least-squares contrast minimisation \citep{Reynaud2014,Bacry2020}, Expectation-Maximisation (EM) procedures \citep{Lewis2011}, and methods using approximations through autoregressive models \citep{Kirchner2017}.
        
        All of these methods are based on the assumption that the history of the process has been accurately observed, although this information is partial or noised in many contexts, in particular due to measurement or detection errors.
        Different models, described notably in \cite{Lund2000}, have been proposed to handle such errors for spatial point processes, with inference methods based on approximating the likelihood depending on each noise scenario.
            The first scenario, called \textit{displacement}, is when the event times are observed with a shift. If deconvolution methods to recover the unnoised process are standard approaches for simpler point processes such as Poisson processes (see for instance \cite{Antoniadis2006,Bonnet2022} or the review by \cite{Hohage2016}), the literature for Hawkes processes remains scarce and consists of the work of \cite{Trouleau2019} where the event times are observed with a delay and the work of \cite{Deutsch2020} where the shift follows a Gaussian distribution. The latter also explores the framework where some event times are undetected, which is similar to that studied in \cite{mei2019}. This setting can either be referred to as \textit{thinning} when the observations are randomly missing, or \textit{censoring} when complete regions are unobserved. The last scenario, called \textit{Superposition of ghost points} by \cite{Lund2000}, is the focus of this paper and describes situations where additional points are coming from an external point process, in our case a Poisson process.   A real-world application that motivated this work comes from spike trains analysis in neuroscience: the membrane potential, which is a continuous signal, is recorded and when it exceeds a certain threshold, one considers that an event time, called a spike, occurred. However, since the original signal is noised, it is possible to detect spikes that do not correspond to real events.
          Regarding inference of such a noisy process, let us highlight that exact likelihood approaches are intractable due to the unknown origin of each occurrence (Hawkes process or noise process) while methods based on inferring this missing information, for instance Expectation-Maximisation algorithms, are computationally too demanding.
               
    Inspired by the work of \cite{Cheysson2022} when the event times of a Hawkes process are aggregated, we turn to spectral analysis (Section~\ref{sec:spectral_analysis}) to propose a novel estimation procedure that allows to infer the noisy Hawkes process.
    On the way, we present a general characterisation of the Bartlett spectrum of the superposition of two independent processes (Section~\ref{sec:superposition})
    and identifiability results of the statistical model in the univariate (Section~\ref{sec:univariate_noisy_hawkes})
    -- which, combined with a recent result by \cite{Yang2024}, guarantees consistency and asymptotic normality of the proposed estimator --
    and the bivariate (Section~\ref{sec:dim2}) settings.
    Inference in these two settings is numerically illustrated in Section~\ref{sec:numerical_results}.
    
    Let us highlight that our analysis reveals situations of identifiability or non-identifiability of the models.
    In particular, identifiability conditions require to have prior information about the parameters, which may appear unrealistic in practice.
    For instance, for a univariate noisy process, one of the four parameters must be known for the model to be identifiable.
    This appears to be possible in specific applications, such as detecting people who are tested positive for a disease with a given detection method, for which the number of false positive occurrences could be known.
    In the bivariate setting, one must have some information on the support of the interaction matrix.
    Let us highlight that, except if one knows the support of the interaction matrix, there is no theoretical argument to ensure that an observed process lies in an identifiable model.
    However, we provide in Section~\ref{sec:numerical_results} numerical investigations regarding support identification.
    Although they are not theoretically grounded, we show that their relevance is supported by empirical evidence.

    Before starting, we present the mathematical setting (Section~\ref{sec:setting})
    and review the references related to spectral approaches (Section~\ref{sec:related_works}).

      \subsection{Mathematical setting}
      \label{sec:setting}
        For a given point process $N$ on $\mathbb R$ with natural filtration $(\mathcal F_t)_{t \in \mathbb R}$, we say that a nonnegative $\mathcal F_t$-progressively measurable process $(\lambda_t)_{t \in \mathbb R}$ is a $\mathcal F_t$-intensity of $N$ if, for all intervals $(a,b]$, 
        \begin{equation*}
            \mathbb E \bigl[ N\bigl((a,b]\bigr) \,\big|\, \mathcal F_a \bigr] = \mathbb E \left[ \int_a^b \lambda_t \,\dd t \:\middle|\: \mathcal F_a \right] \quad \text{a.s.}
        \end{equation*}
      
        Let $H = (H_1, \ldots, H_d)$ be a stationary multivariate Hawkes process on $\RR$,
        \((\mathcal H_t)_{t \in \RR}\) its natural filtration,
        such that \(H\) is defined by the following $\mathcal H_t$-
        intensity functions $\lambda_i^H$ ($i \in \{1, \dots, d\}$):
        for all $t \in \RR$,
        \begin{equation}\label{eq:hawkes_intensity}
          \lambda_i^H(t) = \mu_i + \sum_{j=1}^{d}\int_{-\infty}^{t}{h_{ij}(t-s) \, H_j(\dd s)} = \mu_i + \sum_{j=1}^{d}\sum_{T^{H_j}_k \leq t}{h_{ij}(t  - T_k^{H_j})}\,,
        \end{equation}
        where $\mu_i > 0$ is the baseline intensity of process $H_i$,
        $h_{ij} : \RR \to \RR_{\ge 0}$ is the interaction or kernel function describing the effect of process $H_j$ on process $H_i$,
        and $(T_k^{H_j})_{k\ge1}$ denotes the event times of $H_j$.
        
        By defining the matrix $S = (\|h_{ij}\|_1)_{1 \le i, j \le d}$,
        where
        \[\|h_{ij}\|_1 = \int_{-\infty}^{+\infty}{h_{ij}(t) \, \dd t}\,,\]
        the stationarity condition of $H$ reduces to controlling the spectral radius of $S$: $\rho(S)<1$ \citep{Bremaud1996}. 
        
        The goal of this paper is to study a noisy version of the Hawkes process where the sequences of event times $(T_k^{H_1})_{k\ge1}$, \dots, $(T_k^{H_d})_{k\ge1}$ of $H$ are contaminated by the event times from another process $P$, with associated filtration $(\mathcal P_t)_{t \in \mathbb R}$.
        Since the latter process is aimed at modeling an external noise mechanism due to errors in the detection of event times,
        it is naturally assumed that each subprocess of $H$ is perturbed uniformly with the same level of noise,
        such that $P$ is chosen to be a multivariate Poisson process with same intensity across subprocesses.
        Formally, let $P = (P_1, \ldots, P_d)$ be a collection of \(d\) independent homogeneous Poisson processes with same intensity \(\lambda_0 > 0\),
        independent from \(H\),
        the event times of which are noted
        $(T_k^{P_i})_{k\ge1}$ ($i \in \{1, \dots, d\}$).
        We then consider the point process $N = (N_1, \ldots, N_d)$, with associated filtration $(\mathcal F_t)_{t \in \RR}$, defined as the superposition of $H$ and $P$ (Definition~\ref{def:superposition}):
        the sequence of event times $(T_k^{N_i})_{k\ge1}$ of $N_i$ ($i \in \{1, \dots, d\}$) is the ordered union of $(T_k^{H_i})_{k\ge1}$ and $(T_k^{P_i})_{k\ge1}$.
        
        Throughout this paper we will refer to $N$ as the noisy Hawkes process,
        and it will be assumed that event times of $N$ are observed without knowing their origin (Hawkes or Poisson process).
        Our goal is to estimate both processes (\ie\ the baselines $\mu_i$, the kernels $h_{ij}$ and the shared Poisson intensity $\lambda_0$) from the sole observation of $(T_k^{N_i})_{k\geq1}$, $i \in \{1, \dots, d\}$.
        
        Inference procedures for point processes often leverage the intensity functions in order to devise maximum likelihood and method of moment estimators \citep{Ogata1988, Ozaki1979, DaFonseca2013}.
        Here, the process of interest $N$ being a superposition of two independent point processes,
        the $(\mathcal H_t \vee \mathcal P_t)$-intensity of each subprocess $N_i$, where $\mathcal H_t \vee \mathcal P_t$ stands for the $\sigma$-algebra generated by $\mathcal H_t$ and $\mathcal P_t$, reads (for any integer $i \in \{1, \dots, d\}$):
        for all $t \in \RR$,
        \begin{equation}
          \lambda_i^N(t) = \lambda_0 + \mu_i + \sum_{j=1}^{d}\int_{-\infty}^{t}{h_{ij}(t-s) \, H_j(\dd s)}\,.
          \label{eq:intensity_noisy}
        \end{equation}
        However, it appears from Equation~\eqref{eq:intensity_noisy} that usual estimators cannot be designed from intensity functions $\lambda_1^N, \dots, \lambda_d^N$ since they are based on $H$,
        which is indistinguishable from $P$ in our setting (in other words, we can only consider estimation with respect to the filtration $\mathcal F_t$ rather than the larger filtration $\mathcal H_t \vee \mathcal P_t$ needed to compute the conditional intensity function).
        
        To dive into details, one may observe that, for each \(i \in \{1, \dots, d\}\), \(t \mapsto \mathbb E [\lambda_i^N(t) \,\big|\, \mathcal F_t]\) is an $\mathcal F_t$-intensity of $N$.
        Indeed, since $\mathcal F_t \subset (\mathcal H_t \vee \mathcal P_t)$ for all $t$, then for all intervals $(a,b]$:
        \begin{equation*}
            \mathbb E \bigl[ N_i\bigl((a,b]\bigr) \,\big|\, \mathcal F_a \bigr] = \mathbb E \Bigl[  \mathbb E \bigl[ N_i\bigl((a,b]\bigr) \,\big|\, \mathcal H_a \vee \mathcal P_a \bigr] \,\Big|\, \mathcal F_a \Bigr] = \mathbb E \left[ \mathbb E \left[ \int_a^b \lambda_i^N(t) \,\dd t \:\middle|\: \mathcal H_a \vee \mathcal P_a \right] \:\middle|\: \mathcal F_a \right],
        \end{equation*}
        where the last equality comes from $\lambda_i^N$ being a $(\mathcal H_t \vee \mathcal P_t)$-intensity of $N_i$. 
        Then,
        \begin{equation*}
            \mathbb E \left[ \mathbb E \left[ \int_a^b \lambda_i^N(t) \,\dd t \:\middle|\: \mathcal H_a \vee \mathcal P_a \right] \:\middle|\: \mathcal F_a \right] = \mathbb E \left[ \int_a^b \lambda_i^N(t) \,\dd t \:\middle|\: \mathcal F_a \right] = \mathbb E \left[ \int_a^b \mathbb E \bigl[\lambda_i^N(t) \,\big|\, \mathcal F_t \bigr] \,\dd t \:\middle|\: \mathcal F_a \right],
        \end{equation*}
        since $\mathcal F_a \subset \mathcal F_t$.
        Given that \(t \mapsto \mathbb E \bigl[\lambda_i^N(t) \,\big|\, \mathcal F_t \bigr]\) is $\mathcal F_t$-progressively measurable (or admits a $\mathcal F_t$-progressively measurable modification), it is an $\mathcal F_t$-intensity of $N_i$.
        Then, by 
        Equation \eqref{eq:intensity_noisy}, 
        for each integer $i \in \{1, \dots, d\}$) and all $t \in \RR$:
        \begin{align*}
            \mathbb E \bigl[ \lambda_i^N(t) \,\big|\, \mathcal F_t \bigr] 
            &= \lambda_0 + \mu_i + \sum_{j=1}^d \mathbb E \left[ \int_{-\infty}^t h_{ij}(t-s) H_j(\dd s) \:\middle|\: \mathcal F_t \right]\\
            &= \lambda_0 + \mu_i + \sum_{j=1}^d \mathbb \int_{-\infty}^t h_{ij}(t-s) \mathbb E \bigl[ \II_{H_j(\{s\}) = 1} \,\big|\, \mathcal F_t \bigr] N_j(\dd s) \,.
        \end{align*}
        However, $\mathbb E \bigl[ \II_{H_j(\{s\}) = 1} \,\big|\, \mathcal F_t \bigr]$ is not tractable in practice, so that standard inference approaches based on likelihood are not workable.
        
        This motivates the need for alternative approaches, such as the spectral method proposed in this paper.
        Let us remark that EM procedures could be considered, while being challenging to devise since the distribution of $H_j(\{s\})$ given $\mathcal F_t$ is unknown.
        Regarding this point, \citep{staerman2024}
        addressed inference of
        noisy Hawkes processes in which each event is associated with a continuous mark dependent from its origin (from the Hawkes or from the Poisson process).
        This additional information, which does not exist in our setting, enables them to develop an EM algorithm.

        In order to estimate the Hawkes and the Poisson components of $N$,
        we propose to leverage the spectral analysis of point processes,
        recently advocated by \cite{Cheysson2022} for inference of aggregated Hawkes processes.        
        It consists in considering, for a multivariate point process, its matrix-valued spectral density function, denoted $\mathbf f : \RR \to \mathbb C^{d \times d}$,
        which is related to second-order measures \citep{Bartlett1963}.
        Given some observed
        times $(T_k^{N_1})_{k\ge1}$, \dots, $(T_k^{N_d})_{k\ge1}$,
        (in a prescribed time window $[0, T]$),
        the spectral density is linked to cross-periodograms,
        defined for all pairs $(i, j) \in \{1, \dots, d\}^2$ and all $\nu \in \RR$ by:
        \begin{equation}\label{eq:multivariate_periodogram}
          I_{ij}^T(\nu) = \frac{1}{T}\sum_{k=1}^{N_i(T)}\sum_{l=1}^{N_j(T)}\mathrm{e}^{-2 \pi \iu \nu (T_k^{N_i} - T_l^{N_j})}\,,
        \end{equation}
        where $N_i(t) = N_i([0, t))$.
        Indeed, considering the matrix-valued function $\mathbf I^T : \nu \in \RR \mapsto (I_{ij}^T(\nu))_{1 \le i, j \le d}$,
        the aforementioned link is to be understood as $\mathbf I^T(\nu)$ (for all $\nu \in\RR$) being asymptotically distributed according to a complex Wishart distribution with one degree of freedom and scale matrix $\mathbf f(\nu)$ \citep{Tuan1981, Villani2022}.
        In particular, this implies that \(\mathbb E [\mathbf I^T(\nu)] = \mathbf f(\nu)\).
        Moreover, it is noteworthy that the periodogram $\mathbf I^T(\nu)$ can be computed regardless of knowing the source of the event times.
        This paves the way for estimation.
        
        As it happens, in the scope of statistical inference, a parametric model for the matrix-valued spectral density function is considered:
        \[
            \mathcal{P} = \left\{
              \mathbf{f}_\theta^N : \RR \to \mathbb C^{d \times d},
              \theta = (\mu, \gamma, \lambda_0)\in \Theta
            \right\} \,,
        \]
        where $\gamma$ is a parameter that characterises the interaction functions.
        Then, for $\nu_k = k/T$ ($k \in \{1, 2, \dots\}$), it can be shown that $(\mathbf I^T(\nu_k))_{k \geq 1}$ are asymptotically independent,
        leading to the approximate spectral log-likelihood \citep{Brillinger2012, Duker2019, Villani2022}:
        \begin{equation}\label{eq:spectral_log_likelihood}
	        \ell_T(\theta) = -\frac{1}{T} \sum_{k=1}^{M}
	        \Big\{
	          \log\left(\det\left(\mathbf{f}_\theta^N(\nu_k)\right)\right) +
	          \Tr\left(\mathbf{f}_\theta^N(\nu_k)^{-1}\mathbf{I}^T(\nu_k)\right)
	        \Big\} \,,
        \end{equation}
        where $\det$ and $\Tr$ are respectively the determinant and trace of matrices.
        Then, the so-called Whittle (or spectral) estimator $\mathbf f_{\hat \theta}^N$ of $\mathbf f$ \citep{Whittle1952} is obtained for \(\hat \theta \in \Theta\) such that
        \[
          \hat \theta \in \arg\max_{\theta \in \Theta}~ \ell_T (\theta) \,.
        \]

      \subsection{Related works}
      \label{sec:related_works}
        
        The spectral analysis of point processes was introduced in \cite{Bartlett1963} and extended to 2-dimensional point processes in \cite{Bartlett1964}.
        Subsequent research works focusing on the theoretical properties of the Bartlett spectrum include \cite{Daley1971, DaleyV1, Tuan1981} for temporal settings and \cite{Mugglestone1996, Mugglestone2001, Rajala2023, Yang2024} for spatial contexts.
        
        Spectral results concerning the Hawkes model are built on the linearity of the intensity function along with its branching properties.
        The first spectral analysis of linear Hawkes processes appears in the original paper \cite{Hawkes1971} and was then developed in \cite{DaleyV1} in both univariate and multivariate contexts.
        A recent advance regarding the theoretical analysis of spectral methods for Hawkes processes is
        the work by \cite{Pinkney2024}, who study a penalised version of the spectral log-likelihood with multi-taper periodograms
        and establish estimation error bounds.
        Another advance to the Fourier analysis of point processes is due to \cite{Yang2024} who establish asymptotic results for the estimation of parameters from the spectral log-likelihood.
        
        Explicit expressions of the spectral density of a Hawkes process are available as long as the Fourier transform of the kernel functions are known, which allows to work with a wide array of parametrisations.
        Despite this, practical applications remain scarcer in the literature.
        \cite{Adamopoulos1976} studies earthquake arrivals through the analysis of Hawkes processes and \cite{Karavasilis2007} analyses the cross-correlation of bivariate point processes in the context of muscular stimulation.
        \cite{Pinkney2024} addresses estimation of the inverse spectral density matrix for neuronal activity analysis.
        
        In a recent contribution by \cite{Cheysson2022}, the authors employ the spectral analysis of point processes to infer an aggregated Hawkes process. More precisely, the observations are assumed to come from a standard Hawkes process but only the counts of events on fixed intervals are available.
        By leveraging the properties of the Bartlett spectrum, they propose an estimator obtained by means of maximisation of the spectral log-likelihood.

        The main advantage of the spectral approach is its obvious ability to handle different kinds of partial observations.
        That is why we propose an inference procedure for noisy Hawkes processes (altered by a homogeneous Poisson process) based on Bartlett's spectral density.
        This approach relies on a parametric model, the identifiability of which is studied mainly for classic exponential kernels (some insights are given about identifiability with other kernels).
        Asymptotic properties of the spectral estimator are also addressed based on the recent results by \cite{Yang2024}.

    \section{Spectral analysis}\label{sec:spectral_estimation}
      \subsection{The Bartlett spectrum}\label{sec:spectral_analysis}
        In this section, we formally introduce the concept of matrix-valued spectral measure
        $\boldsymbol{\Gamma}^N : \RR \to \mathbb C^{d \times d}$
        for a multivariate stationary point process $N = (N_1, \ldots, N_d)$.
        This is an extension of the Bartlett spectrum introduced by \cite{Bartlett1963} for the analysis of univariate point processes.
        Let $\mathcal S$ be
        the space of real functions on $\RR$ with rapid decay \cite[Chapter 8.6.1]{DaleyV1}:
        \[
          \mathcal S = \left\{ f \in C^\infty, \forall k \in \{1, 2, \dots\}, \forall r \in \{1, 2, \dots\},
          \sup_{x \in \RR} \left| x^r
          f^{(k)}(x)
          \right| < \infty \right\},
        \]
        where $C^\infty$ is the set of smooth functions from $\RR$ to $\RR$
        and $f^{(k)}$ is the \(k^\text{th}\) derivative of \(f \in C^\infty\).

        Then, the Bartlett spectrum of $N$ is the
        matrix-valued function $\boldsymbol{\Gamma}^N : \nu \in \RR \mapsto (\Gamma_{ij}^N(\nu))_{1 \le i, j \le d} \in \mathbb C^{d \times d}$, such that for all $1 \le i,j \le d$, $\Gamma_{ij}^N$ is a measure on $\RR$ verifying \cite[Equation 8.4.13]{DaleyV1}:
        \[
          \forall (\varphi, \psi) \in \mathcal S \times \mathcal S:
          \quad
          \Cov \left(
            \int_{\RR}{\varphi(x) \, N_i(\mathrm{d}x)} , \int_{\RR}{\psi(x) \, N_j(\mathrm{d}x)}
          \right)
          = \int_{\RR}{ \tilde{\varphi} (\nu) \tilde{\psi}(-\nu) \, \Gamma_{ij}^N (\mathrm{d}\nu)}\,,
        \]        
        where for all $f \in \mathcal S$, $\tilde f : \RR \to \mathbb C$ denotes the Fourier transform of $f$:
        \[
          \forall \nu \in \RR: \quad
          \tilde f(\nu) = \int_{\RR}{ f(x) \mathrm{e}^{-2\pi \iu x \nu} \, \dd x} \,.
        \]
        
        If, for all $1 \le i, j \le d$, the measure $\Gamma_{ij}^N$ is absolutely continuous,
        we can define the matrix-valued spectral density function of $N$, denoted $\mathbf f^N : \RR \to \mathbb C^{d \times d}$,
        such that for all $\nu \in \RR$, $\mathbf f^N(\nu) = (f^N_{ij}(\nu))_{1 \le i, j \le d}$
        with \(\Gamma_{ij}^N(\dd \nu) = f_{ij}^N(\nu) \, \dd \nu\).
      	
      	From a practical point of view, the spectral density $\mathbf f^N$ can be derived from the reduced covariance densities. 
        Let $\mathcal{B}_{\RR}^c$ be the collection of all bounded Borel sets on \(\RR\) and $\ell_{\RR} : \mathcal B_{\RR}^c \to \RR_{\ge0}$ the Lebesgue measure on \(\RR\).
        For all $i \in \{1, \dots, d\}$, the first moment measure of process $N_i$ is defined as $A \in \mathcal B_\RR^c \mapsto \EE[N_i(A)]$ and, by stationarity of the process,
        it comes:
        \[\forall A\in\mathcal{B}_{\RR}^c: \quad \EE[N_i(A)] = m_i^N \ell_{\RR}(A)\,,\]
        where $m_i^N = \EE[N_i([0, 1))]$ is the mean intensity of process $N_i$.
        Then, for all \((i, j) \in \{1, \dots, d\}^2\) the second order moment measure $M_{ij}^N : \mathcal B_\RR^c \times \mathcal B_\RR^c \to \RR_{\ge0}$ is defined by \cite[Section~5.4]{DaleyV1}:
        \[
          \forall (A, B) \in \mathcal{B}_{\RR}^c \times \mathcal{B}_{\RR}^c: \quad
    	    M_{ij}^N(A, B) = \EE[N_i(A)N_j(B)] = \int_{A\times B}{M_{ij}^N(\dd x, \dd y)} \,. 
        \]
        
        Now, as the process $N$ is stationary,
        $M_{ij}^N$ can be decomposed in a product of $\ell_{\RR}$ and a so-called reduced measure $\breve M_{ij}^N : \mathcal B_\RR^c \to \RR_{\ge0}$,
        such that for any bounded measurable function $g : \RR^2 \to \RR$ with bounded support \cite[Equation 8.1.1a]{DaleyV1}:
        \begin{equation}\label{eq:reduced_moment_measure_property}
            \int_{\RR^2}{g(x,y) \, M_{ij}^N(\dd x, \dd y)} = \int_\RR \int_\RR {g(x, x+u) \,\ell_{\RR}(\dd x) \, \breve M_{ij}^N(\dd u)}\,,
        \end{equation}
        which leads to the definition of the reduced covariance measure \(\breve C_{ij}^N : \mathcal B_\RR^c \to \RR_{\ge0}\):
        \begin{equation}\label{eq:breve_c_definition}
        		\forall B \in \mathcal B_\RR^c: \quad
        		\breve C_{ij}^N(B) = \breve M_{ij}^N(B) - m_i^N m_j^N\ell_{\RR}(B)\,.
        \end{equation}
        Since the right-hand side is the difference of two positive, positive-definite measures \cite[Section~8.6]{DaleyV1}, 
        we can define the Fourier transform of $\breve C_{ij}$ as the difference of their Fourier transforms (see for example \cite[Equation 5.2.1]{Pinsky2008} for the Fourier transform of a measure).
        The resulting quantity comes out to correspond exactly to the spectral density function \(f^N_{ij}\):
        \begin{equation}\label{equ:spectral_reduced}
          \forall \nu\in\RR: \quad
          f^N_{ij}(\nu) = \int_{\RR}{\mathrm{e}^{-2\pi \iu x \nu} \,\breve M_{ij}^N(\mathrm{d}x)} - m_i^N m_j^N \delta(\nu)\,,
        \end{equation}
        where \(\delta\) is the Dirac delta function.

      \subsection{Superposition of processes and noisy Hawkes process}\label{sec:superposition}
        The model we study considers the superposition of two point processes that we define as follows.
        \begin{definition}[Superposition of processes]\label{def:superposition}
          Let $X$ and $Y$ be two independent and stationary multivariate point processes with same dimension $d$.
          The \textit{superposition} of $X$ and $Y$, denoted $N = X + Y$, is the stationary multivariate point process defined for any integer $1 \le i \le d$ as:
          \[
            \forall A \in \mathcal B_\RR^c:
            \quad
            N_i(A) = X_i(A) + Y_i(A).
          \]
        \end{definition}
        
        It comes from the definition that if $X$ and $Y$ have respectively event times $(T_k^{X_1})_{k\ge1}$, \dots, $(T_k^{X_d})_{k\ge1}$ and $(T_k^{Y_1})_{k\ge1}$, \dots, $(T_k^{Y_d})_{k\ge1}$, the sequences of event times of $N = X + Y$ are the ordered unions of $(T_k^{X_1})_{k\geq1}$ and $(T_k^{Y_1})_{k\geq1}$ up to $(T_k^{X_d})_{k\geq1}$ and $(T_k^{Y_d})_{k\geq1}$.
        In addition, Proposition~\ref{PROP:SUM_OF_SPECTRAL_DENSITIES} below states that the spectral density of $N$ can be obtained easily from 
        the sum of the spectral densities of $X$ and $Y$ (see also \cite[Exercise~8.2.2]{DaleyV1} for univariate point processes).

        \begin{proposition}\label{PROP:SUM_OF_SPECTRAL_DENSITIES}
	        Let $X$ and $Y$ be two independent and stationary multivariate point processes with same dimension $d$, admitting respective spectral densities $\mathbf{f}^X$ and $\mathbf{f}^Y$.
	        
	        Then $N = X + Y$ admits a matrix-valued spectral density function $\mathbf f^N$ and
	        \begin{equation} \label{eq:sum_spectral_densities}
		        \mathbf{f}^N = \mathbf{f}^X + \mathbf{f}^Y\,.
	        \end{equation}
        \end{proposition}

        \begin{proof}
          The proof is given in Appendix~\ref{appendix:proof_sum_spectra}.
        \end{proof}

	      We are now ready to define the noisy Hawkes model, which is the superposition of a Hawkes process and a homogeneous Poisson process.
	      The latter is formally defined as
    	   a collection of \(d\) independent homogeneous Poisson processes on \(\RR\).

	      \begin{definition}[Noisy Hawkes process]
		      Let $H$ be a multivariate Hawkes process
		      and $P=(P_1, \ldots P_d)$ be a multivariate homogenous Poisson process
		      independent from $H$ and
		      with common intensity $\lambda_0 > 0$ (\ie for any integer $1 \le i \le d$, $P_i$ is a univariate Poisson process with constant intensity $\lambda_0$).
		      The superposition $N = H + P$ is called a \textit{noisy Hawkes process}.
	      \end{definition}
	      
	      Let us remark that, since the process $P$ is aimed at modeling
        some kind of background noise,
        it is naturally assumed that subprocesses of $P$ share the same intensity $\lambda_0$.
        However, if identifiability results from the forthcoming sections are specific to this assumption,
        the presentation done in this section can be trivially extended to a multivariate Poisson process with different intensities.
	      
	      Now, let us recall that we aim at analysing the process $N = H + P$ through an observation of event times $(T_k^{N_1})_{k\geq1}$, \dots, $(T_k^{N_d})_{k\geq1}$ with the inability of distinguishing between the times of $H$ and $P$.
	      To do so from a statistical estimation perspective, we leverage the spectral log-likelihood (Equation~\eqref{eq:spectral_log_likelihood}), which makes use of the spectral density of $N$.
	      The latter can be computed thanks to Proposition~\ref{PROP:SUM_OF_SPECTRAL_DENSITIES} and Example~8.3(c) from \cite{DaleyV1}.
	      Indeed, let $\tilde{\mathbf h} : \RR \to \mathbb C^{d \times d}$ be the matrix-valued Fourier transform of the interaction functions:
	      \[
	        \forall \nu \in \RR: \quad
	        \tilde{\mathbf h}(\nu) = \left( \tilde h_{ij}(\nu) \right)_{1 \le i, j \le d} \,.
	      \]
	      Under stationarity conditions, 
	      the spectral density $\mathbf{f}^H$ of the Hawkes process $H$ (defined as in Equation~\eqref{eq:hawkes_intensity}) is \cite[Equation (13)]{Hawkes1971b}:
	      \begin{equation}\label{eq:multivariate_spectral_matrix}
	        \forall \nu \in \RR: \quad
          \mathbf{f}^H(\nu) = \left( I_d - \tilde{\mathbf h}(\nu) \right)^{-1} \diag \left(m^H \right) \left(I_d - \tilde{\mathbf h}(-\nu)^T \right)^{-1}\,,
        \end{equation}
        where $I_d$ is the identity matrix of dimension $d$ and $\diag \left(m^H \right)$ is the diagonal matrix formed by the vector of the mean intensities:
        \[
          \begin{pmatrix}
            m_1^H\\
            \vdots\\
            m_d^H
          \end{pmatrix}
          = \left( I_d - \tilde{\mathbf h}(0) \right)^{-1} 
          \begin{pmatrix}
            \mu_1\\
            \vdots\\
            \mu_d
          \end{pmatrix}\,.
        \]
        
        Moreover, since the homogeneous Poisson process $P$ (with common intensity \(\lambda_0\)) is a Hawkes process with null interactions,
        its spectral density \(\mathbf f^P\) results easily from Equation~\eqref{eq:multivariate_spectral_matrix}:
        \[
          \forall \nu \in \RR: \quad
          \mathbf{f}^P(\nu) = \lambda_0 I_d \,,
        \]
        which leads to the spectral density of $N$:        
        \begin{equation}\label{eq:noisy_spectral_matrix}
          \forall \nu \in \RR: \quad
          \mathbf{f}^N(\nu)
          = \mathbf{f}^H(\nu) + \lambda_0 I_d \, ,
        \end{equation}
        with $\mathbf{f}^H$ expressed in Equation~\eqref{eq:multivariate_spectral_matrix}.
        Given this result, the inference procedure consists in maximising the spectral log-likelihood described in Equation~\eqref{eq:spectral_log_likelihood}.
        
        In the forthcoming sections, we focus on this statistical estimation problem in low-dimensional settings ($d=1$ and $d=2$) and provide identifiability results when interactions are exponential.

      \section{The univariate noisy Hawkes process}\label{sec:univariate_noisy_hawkes}
        \subsection{General setting}
	        Let us start with univariate processes ($d=1$).
	        In this case, Equation~\eqref{eq:noisy_spectral_matrix} simplifies, as stated in Corollary~\ref{corollary:noisy_spectral_density}.

          \begin{corollary}\label{corollary:noisy_spectral_density}
	          Let $N$ be a noisy Hawkes process defined by the superposition of a stationary Hawkes process $H$
	          (with baseline intensity $\mu > 0$ and kernel function $h : \RR \to \RR_{\ge 0}$)
	          and an independent homogeneous Poisson process $P$ (with constant intensity $\lambda_0 > 0$).
	          Then the spectral density $f^N$ of $N$ reads:
	          \begin{equation*}
		          \forall \nu \in \RR, \quad
		          f^N(\nu) = \frac{\mu}{\left (1 - \|h\|_1 \right) \left \lvert 1 - \tilde h(\nu) \right \rvert^2} + \lambda_0\,.
	          \end{equation*}
          \end{corollary}

          \begin{proof}
            This is straightforward from Equations~\eqref{eq:multivariate_spectral_matrix} and \eqref{eq:noisy_spectral_matrix}.
            See also \cite[Example 8.2(e)]{DaleyV1} for the spectral density of a univariate exciting Hawkes process.
          \end{proof}
          
          The estimation procedure also simplifies since
          the periodogram of $N$ and the spectral log-likelihood (see Equations~\eqref{eq:multivariate_periodogram} and \eqref{eq:spectral_log_likelihood}) respectively read:
          \[
            \forall \nu \in \RR, \quad
            I^T(\nu) = \frac{1}{T}\left\lvert \sum_{k=1}^{N(T)}{\mathrm{e}^{-2\pi \iu \nu T_k^N}}\right\rvert^2\,,
          \]
          where $(T_k^{N})_{k\geq1}$ is the sequence of event times of the noisy Hawkes process $N$,
          and:
          \begin{equation}\label{eq:univariate_log_likelihood}
            \forall \theta \in \Theta, \quad
          	\ell_T(\theta) = -\frac{1}{T}\sum_{k=1}^{M}{\left(\log\left(f_\theta^N(\nu_k)\right) + \frac{I^T(\nu_k)}{f_\theta^N(\nu_k)}\right)}\,.
		      \end{equation}
          As explained in Section~\ref{sec:setting}, the Whittle estimator $\hat \theta$ is then obtained by maximising the function $\ell_T$.

     \subsection{Asymptotic results} \label{sec:consistency}
     
    \cite{Yang2024} have considered the problem of parameter estimation in the frequency domain, using the spectral log-likelihood given by
    \begin{equation*}
        L_T(\theta) = - \int_D \left( \log f_\theta^N(\nu) + \frac{I^T(\nu)}{f_\theta^N(\nu)} \right) \dd\nu \,,
    \end{equation*}
    with $D$ a prespecified compact and $\theta \in \Theta \subset \mathbb R^p$, and the estimator given by
    \begin{equation*}
        \hat\theta_T \in \arg\,\max_{\theta\in\Theta} L_T(\theta) \,.
    \end{equation*}
    The spectral log-likelihood $\ell_T(\theta)$, defined in Equation~\eqref{eq:univariate_log_likelihood}, should then be seen as a Riemann sum approximation of $L_T(\theta)$.
    Define the best fitting parameter $\theta_0$ as 
    \begin{equation*}
        \theta_0 \in \arg\,\max_{\theta\in\Theta} \mathcal L(\theta), \qquad \text{where} \qquad \mathcal L(\theta) = - \int_D \left(\log f_\theta^N(\nu) + \frac{f^N(\nu)}{f_\theta^N(\nu)} \right) \dd\nu \,.
    \end{equation*}
    
    Here we recall the assumptions and asymptotic results on $\hat \theta_T$ that \cite{Yang2024} have established, written in the context of our noisy Hawkes process.
    
    \begin{assumption}
    \label{ass:theta}
    \cite[Assumption 6.1]{Yang2024}
    The parameter space $\Theta$ is a compact subset of $\mathbb R^p$, $p \in \mathbb N$.
    The parametric family of spectral density functions $\bigl\{ f_\theta^N \bigr\}$ is uniformly bounded above and bounded below from zero.
    For all $\nu$, $f_\theta^N(\nu)$ is twice differentiable with respect to $\theta$ and its first and second derivatives are continuous on $\Theta \times D$.
    $\theta_0$ is uniquely determined and lies in the interior of $\Theta$.
    Lastly, $\hat\theta_T$ exists for all $T > 0$ and lies in the interior of $\Theta$.
    \end{assumption}
    
    Note that, in the case where the mapping $\theta \mapsto f_\theta^N$ is injective and the model is correctly specified, then $\theta_0$ is uniquely specified and satisfies $f^N = f_{\theta_0}^N$.
    For the asymptotic variance of $\hat \theta_T$, we introduce the following quantities:
    \begin{align*}
        \Gamma(\theta) &= \frac{1}{4\pi} \int_D \Bigl[ \bigl( f^N(\nu) - f_\theta^N(\nu) \bigr) \nabla^2 (f_\theta^N)^{-1}(\nu) + \bigl( \nabla \log f_\theta^N(\nu) \bigr) \bigl(\nabla \log f_\theta^N(\nu) \bigr)^T \Bigr] \dd\nu\,,\\
        S_1(\theta) &= \frac{1}{4\pi} \int_D f^N(\nu)^2 \bigl( \nabla (f_\theta^N)^{-1}(\nu) \bigr) \bigl( \nabla (f_\theta^N)^{-1}(\nu) \bigr)^T \dd\nu\,,\\
        S_2(\theta) &= \frac{1}{8\pi} \int_{D^2} f_4^N(\nu_1, -\nu_1, \nu_2) \bigl( \nabla (f_\theta^N)^{-1}(\nu_1) \bigr) \bigl( \nabla (f_\theta^N)^{-1}(\nu_2) \bigr)^T\,,
    \end{align*}
    where $\nabla f_\theta$ and $\nabla^2 f_\theta$ are the first and second-order derivatives of $f_\theta$ with respect to $\theta$, and $f_4^N$ the fourth-order cumulant spectral density of the process.
    
    \begin{proposition}
    \label{prop:hawkes_asymptotic}
        Let $N$ be a noisy Hawkes process as in Corollary \ref{corollary:noisy_spectral_density}.
        Suppose that Assumption \ref{ass:theta} holds and that: \textit{(i)} $f^N(\nu) - m^N \in L^1(\mathbb R)$ and \textit{(ii)} $\partial^2 f^N(\nu)$ exists for $\nu \in \mathbb R$ and $\sup_\nu \lvert \partial^2 f^N(\nu) \rvert < \infty$. 
        Then,
        \begin{equation*}
            \hat \theta_T \xrightarrow[]{p} \theta_0, \quad T \to \infty.
        \end{equation*}
        Further assume that $\Gamma(\theta_0)$ is invertible, that $h(\cdot)$ has a $(4+\delta)$th-moment for some $\delta > 0$, that is $\int_0^\infty u^{4+\delta} h(u) \dd u < \infty$, and that: \textit{(i)} $f_4^N - \lambda \in L^1(\mathbb R^3)$ and \textit{(ii)} $f_4^N(\nu_1, \nu_2, \nu_3)$ is twice differentiable with respect to $\nu_2$ and the second-order partial derivative is bounded above.
        Then,
        \begin{equation*}
            \sqrt T \bigl( \hat \theta_T - \theta_0 \bigr) \xrightarrow[]{d} \mathcal N \Bigl( 0, \Gamma(\theta_0)^{-1} \bigl( S_1(\theta_0) + S_2(\theta_0) \bigr) \Gamma(\theta_0)^{-1} \Bigr), \quad T \to \infty.
        \end{equation*}
    \end{proposition}
    
    Proposition \ref{prop:hawkes_asymptotic} is a direct consequence of \cite[Theorem 6.1]{Yang2024}, where the authors verified that the original Theorem's assumptions hold for the Hawkes process \cite[Section 5.1]{Yang2024}.
    It is straightforward to show that these assumptions still hold for the noisy Hawkes process, except for the one about the decay rate of the $\alpha$-mixing coefficient of the process, for which we give a detailed proof below.
    
    \begin{proof}
        For any $\sigma$-algebras $\mathcal U$ and $\mathcal V$, denote $\mathcal U \vee \mathcal V$ the $\sigma$-algebra generated by $\mathcal U$ and $\mathcal V$.
        For any $s, t \in \mathbb R$, $s < t$, define $\mathcal F_s^t$, $\mathcal P_s^t$ and $\mathcal H_s^t$ the $\sigma$-algebras generated on the intervals $(s, t]$ by the processes $N$, $P$ and $H$ respectively, and $\mathcal F_s^\infty = \bigvee_{t > s} \mathcal F_s^t$ and $\mathcal F_{-\infty}^t = \bigvee_{s < t} \mathcal F_s^t$ (and similarly for $\mathcal H$ and $\mathcal P$).
        Then the $\alpha$-mixing coefficient for a stationary point process $N$ is defined as \citep{Poinas2019, Cheysson2022}
        \begin{equation*}
            \alpha_N(t) = \alpha(\mathcal F_{- \infty}^0, \mathcal F_t^\infty) \coloneqq \sup \bigl\{ \lvert \mathbb P(A \cap B) - \mathbb P(A) \mathbb P(B) \rvert : A \in \mathcal F_{- \infty}^0, B \in \mathcal F_t^\infty \bigr\}.
        \end{equation*}
        Since the $\alpha$-mixing coefficient is defined as a supremum over $\sigma$-algebras, and $\mathcal F_s^t \subset \mathcal H_s^t \vee \mathcal P_s^t$ for any $s, t \in \mathbb R \cup \{ - \infty, \infty \}$, 
        \begin{equation*}
            \alpha_N \, (t) 
            = \alpha \, (\mathcal F_{- \infty}^0, \mathcal F_t^\infty)
            \le \alpha \, (\mathcal H_{- \infty}^0 \vee \mathcal P_{- \infty}^0, \mathcal H_t^\infty \vee \mathcal P_t^\infty)
            \le \alpha \, (\mathcal H_{- \infty}^0, \mathcal H_t^\infty) + \alpha \, (\mathcal P_{- \infty}^0, \mathcal P_t^\infty)
            = \alpha_H \, (t) + \alpha_P \, (t)
            = \alpha_H \, (t),
        \end{equation*}
        where the last inequality comes from \cite[Theorem 5.1]{Bradley2005} and the independence between the Hawkes and the Poisson processes, and $\alpha_P \, (t) = 0$ from the Poisson property of independence.
    \end{proof}
    
    It remains that one must show in Assumption \ref{ass:theta} that $\theta_0$ is uniquely determined, begging the question of the model's identifiability when the model is correctly specified.
    In the next section, the exponential parametric model $\mathcal Q$ for $f^N$ is described and analysed in this context.

      \subsection{Exponential model}\label{sec:expon_1d}
        Let us consider the classic exponential kernel for the Hawkes process $H$:
          \begin{equation}\label{equ:exp_kernel_1d}
            \forall t \in \RR: \quad
            h(t) = \alpha\beta \mathrm{e}^{-\beta t} \II_{t \ge 0}\,,
          \end{equation}
          with $0 < \alpha < 1$ and $\beta > 0$.
          The kernel parameter is thus $\gamma = (\alpha, \beta)$ and the statistical model for a univariate noisy Hawkes process becomes:
          \[
        		\mathcal Q = \left\{
        		  f^N_\theta : \RR \to \mathbb C, 
              \theta = (\mu, \alpha, \beta, \lambda_0) \in \RR_{>0}\times (0,1) \times \RR_{>0} \times \RR_{>0}
            \right\}\,.
          \]
          
          The exponential kernel parametrisation has been widely studied from an inference point of view (see for instance \cite{Ozaki1979, Bacry2016}).
          In particular, its Fourier transform is:
          \begin{equation}\label{eq:fourier_exponential}
            \forall \nu \in \RR: \quad
            \tilde h(\nu) = \frac{\alpha \beta}{\beta + 2 \pi \iu \nu} \,,
          \end{equation}
          and the spectral density $f^H$ of a univariate Hawkes process
          with baseline intensity $\mu > 0$ and exponential kernel
          is \citep{Hawkes1971}:
          \[
            \forall \nu \in \RR: \quad
            f^H(\nu) = m^H \left(1 + \frac{\beta^2 \alpha (2 - \alpha)}{\beta^2 (1 - \alpha)^2 + 4 \pi^2 \nu^2}\right), \quad
            \text{where } m^H = \frac{\mu}{1 - \alpha}\,.
          \]
          It results that the spectral density $f_\theta^N \in \mathcal Q$ of a univariate noisy Hawkes process is:          
          \begin{equation}\label{eq:exponential_spectral_density}
            \forall \nu \in \RR: \quad
          	f_{\theta}^N(\nu) = \left( \frac{\mu}{1-\alpha}{\beta^2 \alpha (2-\alpha)} \right)
          	\frac{1}{\beta^2 (1-\alpha)^2 + 4 \pi^2 \nu^2}
          	+ \left(\frac{\mu}{1-\alpha} + \lambda_0\right)\,.
          \end{equation}
        
        Now, we should be ready for implementing the inference procedure based on maximising the spectral log-likelihood as expressed in Equation~\eqref{eq:univariate_log_likelihood}.
        However, it appears that the model $\mathcal Q$, as currently defined, is not identifiable (see Proposition~\ref{PROPOSITION:FIXED_UNIVARIATE_IDENTIFIABILITY}).
        Hopefully, this model becomes identifiable when restricted to only three parameters, thus allowing the practicability of the proposed estimation method,
        as numerically illustrated in Section~\ref{sec:univariate_numerical_results}.
        
        \begin{proposition}[Identifiability in the univariate setting]\label{PROPOSITION:FIXED_UNIVARIATE_IDENTIFIABILITY}
        	The model $\mathcal{Q}$ is not identifiable.
        	In particular, for any admissible parameter $\theta = (\mu, \alpha, \beta, \lambda_0)$ there exists an infinite number of admissible parameters $\theta'$ such that $f_{\theta}^N = f_{\theta'}^N$.
        	
        	However, the four collections of models defined below are identifiable:
          \begin{enumerate}
            \item for all $\mu^\circ > 0$, $\mathcal Q_{\mu} = \left\{ f_{(\mu, \alpha, \beta, \lambda_0)}^N \in \mathcal Q, \mu = \mu^\circ \right\}$;
            \item for all $\alpha^\circ \in (0, 1)$, $\mathcal Q_{\alpha} = \left\{ f_{(\mu, \alpha, \beta, \lambda_0)}^N \in \mathcal Q, \alpha = \alpha^\circ \right\}$;
            \item for all $\beta^\circ > 0$, $\mathcal Q_{\beta} = \left\{ f_{(\mu, \alpha, \beta, \lambda_0)}^N \in \mathcal Q, \beta = \beta^\circ \right\}$;
            \item for all $\lambda_0^\circ > 0$, $\mathcal Q_{\lambda_0} = \left\{ f_{(\mu, \alpha, \beta, \lambda_0)}^N \in \mathcal Q, \lambda_0 = \lambda_0^\circ \right\}$.
          \end{enumerate}
        \end{proposition}
        \begin{proof}
	        The proof is given in Appendix~\ref{appendix:identifiability_univariate}.
        \end{proof}
        
        The previous discussion raises the question of whether this non-identifiability also extends to the distribution of the noisy Hawkes process.
It turns out that, in the exponential case, Markov properties of the intensity function $\lambda^H$ of the underlying Hawkes process help ensuring identifiability.
Indeed, from the definition of the Hawkes process $H$ and the stationarity of the Poisson process $P$, $\left( \lambda^N(t) \right)_{t \ge 0}$ is a Markov process: it decreases with rate $\beta(\lambda^H(t) - \mu) = \beta (\lambda^N(t) - \mu - \lambda_0)$, and the jumps occurring from the Hawkes process, with rate $\lambda^H(t) = \lambda^N(t) - \lambda_0$, are of size $\alpha \beta$, while the jumps occurring from the Poisson component, with rate $\lambda_0$, have no impact on the intensity of the process. 
Then $\left( \lambda^N(t), N(t) \right)_{t \ge 0}$ is also a Markov process, more specifically a piecewise deterministic Markov process \citep{Davis1984}.
This allows us to use results from \cite{Dassios2011} on the distribution of exponential Hawkes processes to show that the distribution of the exponential noisy Hawkes process is identifiable. 

\begin{proposition}\label{PROP:STATN_IDENTIFIABILITY}
	Let $(\mu, \alpha, \beta, \lambda_0)$ and $(\mu', \alpha', \beta', \lambda_0')$ be two admissible $4$-tuples for the exponential noisy Hawkes model $\mathcal Q$,
	and $N$ and $N'$ respectively defined by these two tuples (Equation~\eqref{eq:intensity_noisy} with kernel \eqref{equ:exp_kernel_1d}).
	Then,
	if $N$ and $N'$ have same distribution
	and $\lambda^N(0)$ (respectively $\lambda^{N'}(0)$) is distributed according to the stationary distribution of $\left( \lambda^N(t) \right)_{t \ge 0}$ (respectively $\bigl( \lambda^{N'}(t) \bigr)_{t \ge 0}$),
	then $(\mu, \alpha, \beta, \lambda_0) = (\mu', \alpha', \beta', \lambda_0')$.
\end{proposition} 
\begin{proof}
  See Appendix~\ref{app:felix1}.
\end{proof}

A consequence of this result is that non-identifiability of our proposed method in the exponential case is a shortcoming of the spectral approach itself rather than an underlying property of the noisy Hawkes process, presumably stemming from the fact that the spectral density only encodes the second order moments of the process.
In the forthcoming section, we briefly investigate whether this issue arises when considering other reproduction kernels.

      \subsection{Beyond the exponential model}
        Let $N$ be a noisy Hawkes process defined by the superposition of a stationary 	Hawkes process with baseline intensity $\mu$ and kernel function $\alpha h$, with $\alpha \in (0,1)$, $h : \RR \to \RR_{\ge 0}$ and $\|h\|_1 = 1$, and a homogeneous Poisson process with intensity $\lambda_0$.
Per Corollary~\ref{corollary:noisy_spectral_density}, its spectral density is given by
\begin{equation*}
  \forall \nu \in \RR: \quad
	f^N(\nu) = \frac{\mu}{(1 - \alpha) \left \lvert 1 - \alpha \tilde h(\nu) \right \rvert^2} + \lambda_0 \,.
\end{equation*}

While it may be difficult to show that a model is identifiable from the spectral density expression, it may prove fruitful to look at its Taylor expansion around 0, and analyse the Taylor coefficients. 
For example, considering the uniform kernel and the corresponding Taylor expansion of the spectral density up to order 2, we get the following proposition.

\begin{proposition}\label{PROP:RECTANGLE}
  Let us consider a rectangle interaction function
  \[
    h : t \in \RR \mapsto \phi^{-1} \II_{0 \le t \le \phi}\,,
  \]
  for some kernel parameter $\phi > 0$,
  and the corresponding statistical model for a univariate noisy Hawkes process:
	\begin{equation*}
		\mathcal R = \left \{ f_\theta^N: \RR \to \mathbb C, \theta = (\mu, \alpha, \phi, \lambda_0) \in \mathbb R_{>0} \times (0, 1) \times \mathbb R_{>0} \times \mathbb R_{>0} \right\} \,.
	\end{equation*}
	Then $\mathcal R$ is identifiable.
\end{proposition}
\begin{proof}
  See Appendix~\ref{app:felix3}.
\end{proof}

This last proposition shows that non-identifiability of the spectral approach for the noisy exponential Hawkes process is more a consequence of the choice of the reproduction function $h$ rather than a general shortcoming of the spectral approach.
It is unexpected that the exponential reproduction function, usually chosen because the Markov properties for the resulting Hawkes intensity simplify calculations \citep{Ozaki1979, DaFonseca2013, Duarte2019}, seems to be here the main culprit of non-identifiability for our proposed spectral approach. 
Still, we will show how, by imposing some constraints on the modelling of multivariate noisy Hawkes processes, we are able to ensure identifiability of the model even in this case.

    \section{The bivariate noisy Hawkes process}\label{sec:dim2}
    
      This section addresses bivariate noisy Hawkes processes ($d=2$).
      More precisely, for such a process $N = H + P$, where $H$ is a bivariate Hawkes process (see Equation~\eqref{eq:hawkes_intensity}) and $P$ a Poisson process with shared intensity $\lambda_0$,
      Corollary~\ref{corollary:noisy_spectral_density_2} gives the closed-form expression of the spectral density $\mathbf f^N$.
      
      \begin{corollary}\label{corollary:noisy_spectral_density_2}
        Let $N = (N_1, N_2)$ be a bivariate noisy Hawkes process defined by the superposition of a stationary Hawkes process $H = (H_1, H_2)$
        (with baseline intensities $\mu_1 > 0$ and $\mu_2 > 0$, and kernel functions $h_{ij} : \RR \to \RR_{\geq 0}$, $1 \le i, j \le 2$)
        and an independent homogeneous Poisson process $P = (P_1, P_2)$ (with same constant intensity $\lambda_0 > 0$).
        Then the spectral density $\mathbf f^N$ of $N$ reads:
        \begin{equation*}
          \forall \nu \in \RR, \quad
          \mathbf f^N(\nu) =
          \begin{pmatrix}
            f_{11}^H(\nu) + \lambda_0 & f_{12}^H(\nu) \\
            f_{21}^H(\nu) & f_{22}^H(\nu) + \lambda_0
          \end{pmatrix} \,,
        \end{equation*}
        where
        \[
          \begin{cases}
            f_{11}^H (\nu) = \frac{m_1^H \left \lvert 1 - \tilde h_{22}(\nu) \right \rvert^2 + m_2^H \left \lvert  \tilde h_{12}(\nu) \right \rvert^2}{\left \lvert \left (1-\tilde h_{11}(\nu) \right) \left (1-\tilde h_{22}(\nu) \right) - \tilde h_{12}(\nu) \tilde h_{21}(\nu) \right \rvert^2}\\
            f_{12}^H(\nu) = \frac{m_1^H \left (1-\tilde h_{22}(\nu) \right)\tilde h_{21}(-\nu) + m_2^H \left (1-\tilde h_{11}(-\nu) \right) \tilde h_{12}(\nu)}{\left \lvert \left (1-\tilde h_{11}(\nu) \right) \left (1-\tilde h_{22}(\nu) \right) - \tilde h_{12}(\nu) \tilde h_{21}(\nu) \right \rvert^2}
          \end{cases} \,,
        \]
        and
        \[
          m_1^H = \frac{\mu_1\left( 1 - \|h_{22}\|_1 \right)  + \mu_2 \|h_{12}\|_1 }{\left( 1 - \|h_{11}\|_1 \right)\left( 1 - \|h_{22}\|_1 \right) - \|h_{12}\|_1 \|h_{21}\|_1} \, ,
        \]
        and $f_{22}^H$, $f_{21}^H$ and $m_2^H$ are obtained by symmetry of all indices.
      \end{corollary}
      \begin{proof}
        This is straightforward from Equations~\eqref{eq:multivariate_spectral_matrix} and \eqref{eq:noisy_spectral_matrix}.
      \end{proof}
      
      Then, the estimation procedure is exactly that described in Section~\ref{sec:setting},
      which is based on computing the cross-periodogram $\mathbf I^T$ and on maximising the spectral log-likelihood $\ell_T$ (see Equations~\eqref{eq:multivariate_periodogram} and \eqref{eq:spectral_log_likelihood}).
      Now, similarly to the univariate case detailed in Section~\ref{sec:expon_1d}, we consider exponential interaction functions, \ie for $1 \le i, j \le 2$:
      \[
        \forall t \in \RR: \quad
        h_{ij}(t) = \alpha_{ij} \beta_i \mathrm e^{-\beta_i t} \II_{t \ge 0},
      \]
      with $\alpha_{ij} \ge 0$ and $\beta_i > 0$.
      The kernel parameter is thus $\gamma = (\alpha, \beta)$, where $\alpha \in \RR_{\ge 0}^{2 \times 2}$ and $\beta \in \RR_{>0}^2$ and the statistical model for a bivariate noisy Hawkes process becomes:
      \[
        \mathcal Q_\Lambda =
        \left\{
          \mathbf{f}_{\theta}^N : \RR \to \mathbb C^{2 \times 2},
          \theta = (\mu, \alpha, \beta, \lambda_0)
          \in \RR_{>0}^2\times \Lambda \times \RR_{>0}^2 \times \RR_{>0},
          \beta \in \Omega_\alpha
        \right\}\,,
      \]
      where $\Lambda \subset \{\alpha \in \RR_{\geq 0}^{2\times 2} : \rho(\alpha) < 1\}$ is subset of matrices $\alpha$ that will be specified later,
      and for all $\alpha \in \RR_{\geq 0}^{2\times 2}$,
      \[
        \Omega_\alpha =
        \left\{
          \beta \in \RR_{>0}^2,
          \beta_1 = 1 \text{ if } \alpha_{11} = \alpha_{12} = 0,
          \beta_2 = 1 \text{ if } \alpha_{21} = \alpha_{22} = 0
        \right\} \,,
      \]
      is a subset of admissible values for $\beta$.
      The definition of $\Omega_\alpha$
      takes into account that when a row, say the first one, of the interaction matrix $\alpha$ is null, then the corresponding kernels verify $h_{11}=0$ and $h_{12}=0$ independently of the value of $\beta_1$.
      Thus, identifiability for the parameter $\beta_1$ is hopeless, which justifies that we get rid of it from the model (by fixing it to an arbitrary value).
      
      \begin{remark}
        Different versions of the multivariate exponential model exist. A first convention assumes that there is a unique $\beta\in\RR_{>0}$ shared by all kernel functions \citep{Chevallier2019, Bacry2020}. A second and less restrictive option, which is that we opt for in this paper, assumes that the recovery rate $\beta_i\in\RR_{>0}$ for each subprocess $N_i$ $(1 \le i \le d)$ is shared among received interactions \citep{Bonnet2023}. These choices allow for simplified derivations of estimators in the time domain and in the frequency domain, as shown below.  
      \end{remark}
      
      The aim of this section is to study identifiability of model $\mathcal Q_\Lambda$.
      A broad analysis (\ie for $\Lambda = \{\alpha \in \RR_{\geq 0}^{2\times 2} : \rho(\alpha) < 1\}$) being out of reach for complexity reasons,
      we exhibit some situations (\ie subsets $\Lambda$) for which non-identifiability (Proposition~\ref{PROPOSITION:BI_NON_IDENTIFIABLE}) or identifiability (Proposition~\ref{PROPOSITION:BI_IDENTIFIABLE}) can be proved.
      
      \begin{proposition}[Non-identifiability in the bivariate setting]\label{PROPOSITION:BI_NON_IDENTIFIABLE}
      	The model $\mathcal Q_\Lambda$ is not identifiable in the three situations:
        \begin{enumerate}
          \item \label{hyp:bi_non_identifiable_1} $\Lambda = \left\{ \begin{pmatrix} \alpha_{11} & 0 \\ 0 & \alpha_{22} \end{pmatrix}, 0 \le \alpha_{11}, \alpha_{22} < 1 \right\}$, that is for diagonal matrices $\alpha$ (with possibly null entries).
          \item \label{hyp:bi_non_identifiable_2} $\Lambda = \left\{ \begin{pmatrix} \alpha_{11} & \alpha_{12} \\ 0 & 0 \end{pmatrix}, 0 < \alpha_{11} < 1, \alpha_{12} > 0 \right\}$, that is for matrices with positive entries in the first row and null entries in the second.
          \item \label{hyp:bi_non_identifiable_2_bis} $\Lambda = \left\{ \begin{pmatrix} 0 & 0 \\ \alpha_{21} & \alpha_{22} \end{pmatrix}, \alpha_{21} > 0, 0 < \alpha_{22} < 1 \right\}$, that is for matrices with null entries in the first row and positive entries in the second.
        \end{enumerate}
      \end{proposition}
      \begin{proof}
      	The proof is given in Appendix~\ref{appendix:bi_non_identifiable}.
      \end{proof}
      
      \begin{remark} \label{rem:non_identifiability_submodels}
        The proof of Proposition~\ref{PROPOSITION:BI_NON_IDENTIFIABLE}, Situation~\ref{hyp:bi_non_identifiable_1} reveals that non-identifiability stands actually for each subprocess (considered as a univariate process),
        such that all the submodels with null cross-interactions built by fixing $\alpha_{11}$ or $\alpha_{22}$ to zero,
        or by keeping them away from zero are also not identifiable.
      \end{remark}
      
      \begin{proposition}[Identifiability in the bivariate setting]\label{PROPOSITION:BI_IDENTIFIABLE}
      	The model $\mathcal Q_\Lambda$ is identifiable in the four situations:
        \begin{enumerate}
          \item \label{HYP:BI_IDENTIFIABLE_1} $\Lambda = \left\{ \begin{pmatrix} \alpha_{11} & 0 \\ \alpha_{21} & 0 \end{pmatrix}, 0 \le \alpha_{11} < 1, \alpha_{21} > 0 \right\}$, that is for matrices $\alpha$ with null entries in the second column and a positive entry on the antidiagonal.
          \item \label{HYP:BI_IDENTIFIABLE_1_BIS} $\Lambda = \left\{ \begin{pmatrix} 0 & \alpha_{12} \\ 0 & \alpha_{22} \end{pmatrix}, \alpha_{12} > 0, 0 \le \alpha_{22} < 1 \right\}$, that is for matrices $\alpha$ with null entries in the first column and a positive entry on the antidiagonal.
          \item \label{HYP:BI_IDENTIFIABLE_2} $\Lambda = \left\{ \begin{pmatrix} \alpha_{11} & 0 \\ \alpha_{21} & \alpha_{22} \end{pmatrix}, 0 < \alpha_{11} < 1, \alpha_{21} > 0 , 0 \le \alpha_{22} < 1 \right\}$, that is for matrices $\alpha$ with positive entries in the first column and null upper right entry.
          \item \label{HYP:BI_IDENTIFIABLE_2_BIS} $\Lambda = \left\{ \begin{pmatrix} \alpha_{11} & \alpha_{12} \\ 0 & \alpha_{22} \end{pmatrix}, 0 \le \alpha_{11} < 1, \alpha_{12} > 0, 0 < \alpha_{22} < 1 \right\}$, that is for matrices $\alpha$ with a null lower left entry and positive entries in the second column.
        \end{enumerate}
        
        In addition, denoting respectively \(\Lambda_1\), \(\dots\), \(\Lambda_4\) the four sets described above,
        the model \(Q_{\cup_{j=1}^4 \Lambda_j} = \cup_{j=1}^4 Q_{\Lambda_j}\) is identifiable.
      \end{proposition}
      \begin{proof}
        The proofs of Situations~\ref{HYP:BI_IDENTIFIABLE_1} and \ref{HYP:BI_IDENTIFIABLE_2} are respectively in Appendices~\ref{appendix:bi_identifiable_1} and \ref{appendix:bi_identifiable_2}.
        The other situations are obtained by symmetry of all indices.
        
        The proof regarding \(Q_{\cup_{j=1}^4 \Lambda_j}\) is in Appendix~\ref{appendix:all_bi_identifiable}.
      \end{proof}

    \begin{remark}
        In the case where the Poisson process \(P\) does not share the same intensity across subprocesses (say \(\lambda_0\) and \(\lambda_0'\)),
        nonidentifiability still holds in the situations exhibited in Proposition~\ref{PROPOSITION:BI_NON_IDENTIFIABLE}.
        However, regarding identifiability, one can prove that it is still true for Situations~\ref{HYP:BI_IDENTIFIABLE_2} and \ref{HYP:BI_IDENTIFIABLE_2_BIS} from Proposition~\ref{PROPOSITION:BI_IDENTIFIABLE},
        while it presents an issue for Situations~\ref{HYP:BI_IDENTIFIABLE_1} and \ref{HYP:BI_IDENTIFIABLE_1_BIS}
        
        In any case, let us highlight that once identifiability is proved for a model, there is no additional methodological difficulty for the estimation procedure with different values of noise level for each subprocess.
    \end{remark}
    
    \begin{remark}
        Let us remark that, if we assume that the noise is known in advance, the identifiability of the full model $\mathcal{Q}$ holds. If the noise is known (whether they are identical between subprocesses or not), it suffices to prove that the model is identifiable for a multivariate Hawkes process. The equality of the spectral densities of two processes (unnoised) Hawkes processes implies the equality of their first and second order characteristics, which by the results from \cite[Corollary~1]{Bacry2016} implies equality of the baselines and the interaction functions. 
        
        The identifiability of the model for unnoised multivariate Hawkes processes holds, for any dimension $d\geq 1$, under the assumption that the parameterisation of the kernel functions is injective with respect to all parameters. This is in particular true for the exponential model used throughout this paper.
    \end{remark}
    
      Several lessons can be learnt from Propositions~\ref{PROPOSITION:BI_NON_IDENTIFIABLE} and \ref{PROPOSITION:BI_IDENTIFIABLE} and Remark~\ref{rem:non_identifiability_submodels}.
      First, the statistical model $\mathcal Q_\Lambda$ is not identifiable if
      $H$ reduces to a bivariate homogenous Poisson process (Proposition~\ref{PROPOSITION:BI_NON_IDENTIFIABLE}, Situation~\ref{hyp:bi_non_identifiable_1} with $\alpha_{11}=\alpha_{22}=0$)
      or to two independent univariate Hawkes processes (Proposition~\ref{PROPOSITION:BI_NON_IDENTIFIABLE}, Situation~\ref{hyp:bi_non_identifiable_1} with $\alpha_{11} >0$ and $\alpha_{22}>0$)
      even if the noise $P$ shares the same intensity $\lambda_0$ for both subprocesses.
      This result actually still holds true for a dimension $d > 2$.
      
      Second, Proposition~\ref{PROPOSITION:BI_IDENTIFIABLE} tells in a nutshell that
      the model $\mathcal Q_\Lambda$ is identifiable if there exist cross-interactions in the Hawkes process $H$ (\ie $\alpha_{12} > 0$ or $\alpha_{21} > 0$).
      However,
      interactions must not come from a Poisson subprocess and reach a self-exciting Hawkes subprocess (Proposition~\ref{PROPOSITION:BI_NON_IDENTIFIABLE}, Situations~\ref{hyp:bi_non_identifiable_2} and \ref{hyp:bi_non_identifiable_2_bis}),
      but rather
      it is sufficient to
      reach a self-neutral (\ie with null self-excitation) Hawkes subprocess (Proposition~\ref{PROPOSITION:BI_IDENTIFIABLE}, Situations~\ref{HYP:BI_IDENTIFIABLE_1} and \ref{HYP:BI_IDENTIFIABLE_1_BIS}) 
      or come from a self-exciting Hawkes subprocess
      (Proposition~\ref{PROPOSITION:BI_IDENTIFIABLE}, Situations~\ref{HYP:BI_IDENTIFIABLE_2} and \ref{HYP:BI_IDENTIFIABLE_2_BIS}).

     Let us mention that the identifiability conditions given in Proposition~\ref{PROPOSITION:BI_IDENTIFIABLE} cannot be generally verified without prior knowledge on the interaction matrix. However, we propose in the next section a practical method to recover the graph of interactions, which requires replicates of the process or, alternatively, a unique but long observation of the process.

\section{Numerical results}\label{sec:numerical_results}
        
  This section numerically illustrates the behaviour of the proposed estimator $\hat \theta$ (described in Section~\ref{sec:setting})
  for noisy exponential Hawkes processes.
  It investigates the effect of horizon $T$ and hyperparameter $M$ in the univariate setting (Section~\ref{sec:univariate_numerical_results}),
  and the impact of model sparsity and interaction strength in the bivariate setting (Section~\ref{sec:bivariate_numerical_results}).
  
  In the whole study, point processes are simulated thanks to Ogata's thinning method \citep{Ogata1981}
  and numerical optimisation of the spectral log-likelihood is performed via the L-BFGS-B method \citep{Byrd1995}, implemented in the \texttt{scipy.optimize.minimize} Python function.
  The computation of the periodogram is accelerated via a Fast Fourier Transform algorithm using the Python package \texttt{finufft} \citep{Barnett2019, Barnett2021}.
  Both simulation and estimation algorithms are freely available as a Python package on
  \cite{Code}. %GitHub\footnote{\url{https://github.com/migmtz/noisy-hawkes-process}}.
        
  \subsection{Univariate setting}\label{sec:univariate_numerical_results}

    \subsubsection{Simulation and estimation scenarios}\label{sec:uni_simulation}
      \paragraph{Data simulation}
        We consider observations $(T_k^N)_{1 \le k \le N(T)}$ coming from a univariate exponential noisy Hawkes process $N = H+P$,
        where $P$ is a Poisson process with intensity $\lambda_0 > 0$
        and $H$ is a Hawkes process with baseline intensity $\mu = 1$ and kernel given by Equation~\eqref{equ:exp_kernel_1d}
        with parameters $\alpha=0.5$ and $\beta=1$.

        In order to get close to stationarity while coping with inability to generate a process on the whole line $\RR$,
        $N$ is simulated on the window $[-100, T]$ with no points in $(-\infty, -100)$
        but only observations falling in $[0, T]$ are considered.
        
        The forthcoming section will illustrate the convergence of $\hat \theta$,
        thanks to its behaviour with respect to varying horizon $T \in \{250, 500, 1000, \dots, 8000\}$,
        and the impact of the intensity $\lambda_0$ of the noise process $P$ on estimation accuracy,
        via varying noise-to-signal ratio ${\lambda_0}/{m^H} \in \{ 0.2, 0.4, \dots, 2.0 \}$
        (given the average intensity $m^H = {\mu}/{(1-\alpha)}=2$ of $H$).

      \paragraph{Statistical models}
        According to Proposition~\ref{PROPOSITION:FIXED_UNIVARIATE_IDENTIFIABILITY},
        which states four collections of identifiable models for univariate exponential Hawkes processes,
        estimation is successively performed in models $\mathcal Q_{\mu}$, $\mathcal Q_{\alpha}$, $\mathcal Q_{\beta}$, $\mathcal Q_{\lambda_0}$,
        where the known parameter is fixed to the value of the generated process (see above).
        
        In addition, the behaviour of $\hat \theta$ will be assessed thanks to its relative error ${\|\hat \theta - \theta^\star\|_2}/C_{\mathcal Q}^\star$
        (where \(C_{\mathcal Q}^\star\) is the norm of the vector of parameters of the generated process, that are considered free in the model --
        for instance \(C_{\mathcal Q}^\star = \sqrt{\alpha^2 + \beta^2 + \lambda_0^2}\) for \(\mathcal Q_\mu\))
        averaged over $50$ different trials.
        
      \subsubsection{Convergence, computation time and influence of parameter $M$}
      
        Up to now, the hyperparameter $M$,
        appearing in the spectral log-likelihood (Equations~\eqref{eq:spectral_log_likelihood} and \eqref{eq:univariate_log_likelihood})
        and determining the number of frequencies tested with the spectral density,
        has been let free.
        However, the theoretical literature suggests that its choice may have a lot of influence on the performance of the estimation procedure.
        In \citep{Tuan1981}, it is suggested
        that \(M \to \infty\)
        as $T \to \infty$, in particular to include enough frequencies with respect to the information available.
        A rational choice is to consider $M/T \xrightarrow[T \to \infty]{} \infty$, which allows for a strictly increasing window of frequencies (see Theorem 5 of the aforementioned paper,
        which ensures convergence of a modified version of the spectral log-likelihood \(\ell_T(\theta)\) to its integral version).

        Since the rate of convergence is not specified but has an effect on the computational efficiency of the spectral estimator,
        we propose to study the compatible case $M = N(T) \log N(T)$
        and the economy case $M = N(T)$.
		 
        Figure~\ref{fig:errors_wrt_T} displays the relative errors of \(\hat \theta\)
        with respect to both the simulation horizon $T$ (top panels) and the estimation time (bottom panels). 
		 As expected, the quality of the estimations improves as $T$ increases, and this independently of which parameter is fixed.
		 Estimations are slightly better when considering $M = N(T) \log N(T)$ (orange line) especially for smaller values of $T$ but the trade-off is a ten times higher computation time. 
          Therefore, the performance benefit of taking $M = N(T) \log N(T)$ rather than $M = N(T)$ seems minor when compared to the computational cost. In practice, what seems to be important is that for a fixed window of observation $[0, T]$, the maximal frequency $M/T$ is big enough so that, for any $\nu > M/T$, the spectral density $f(\nu)\approx m_1$, in other words, so it is almost constant.
         The forthcoming numerical results will be performed with $M = N(T)$, value above which the considered spectral functions are practically constant.

    \begin{figure}[!ht]
        \centering
        \includegraphics[width=0.8\textwidth]{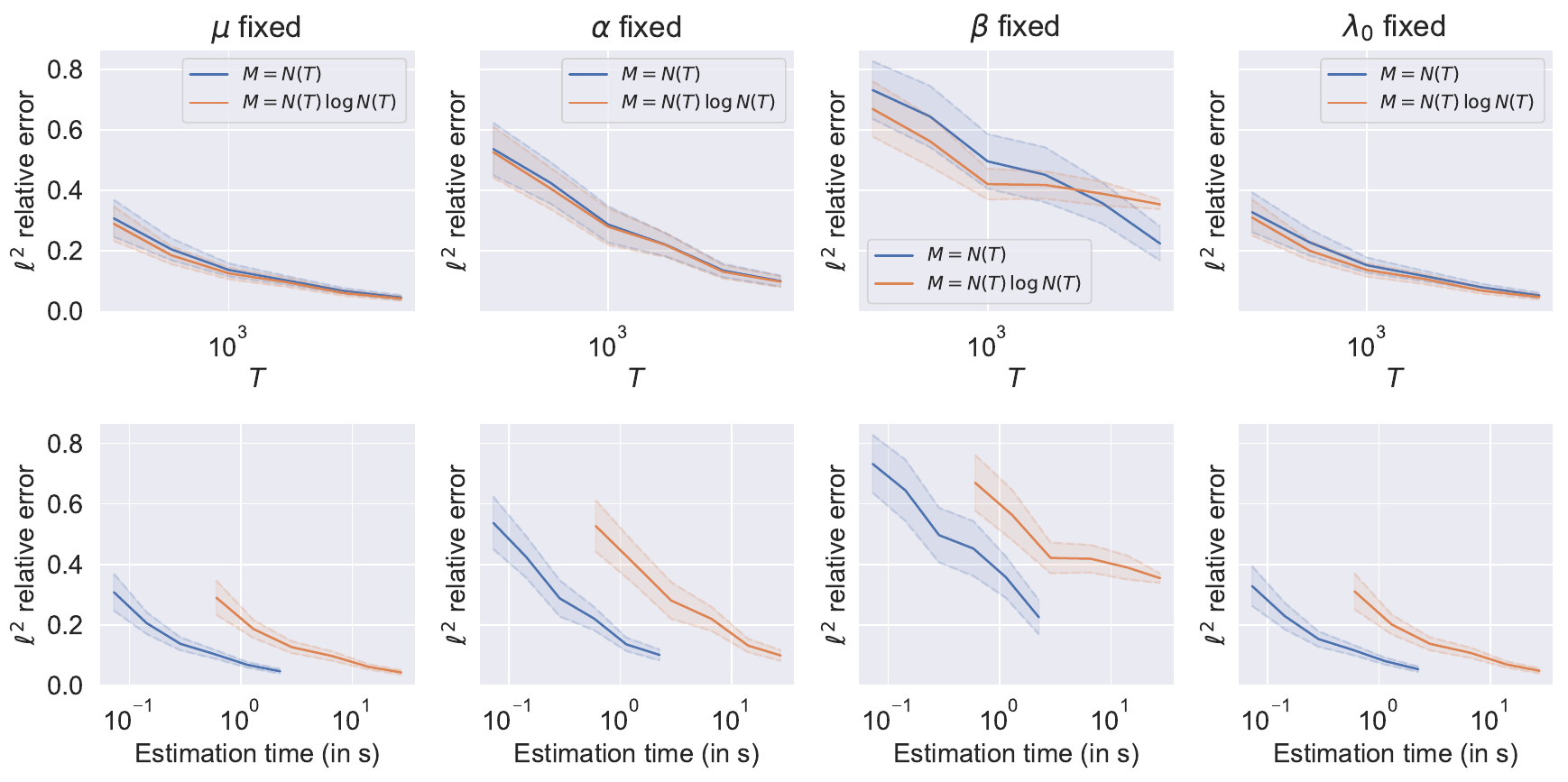}
        \caption{Relative estimation error with confidence bands ($\pm 1.96$ empirical standard deviation) respectively for \(\mu\), \(\alpha\), \(\beta\) and \(\lambda_0\) fixed (columns from left to right)
        with respect to the time horizon $T$ (top) and the computation time (bottom) in logarithmic scale. Level of noise $\lambda_0 = 1.6$.}
        \label{fig:errors_wrt_T}
    \end{figure}

    \subsubsection{Influence of the noise level}\label{sec:influence_noise}
       Figure~\ref{fig:errors_wrt_noise_N} shows the relative error with respect to the ratio $\lambda_0/m_H$, obtained when increasing the value of $\lambda_0$ while
       keeping the other parameters fixed, for a given horizon $T=8000$.

       \begin{figure}[!ht]
        \centering
        \includegraphics[width=0.8\textwidth]{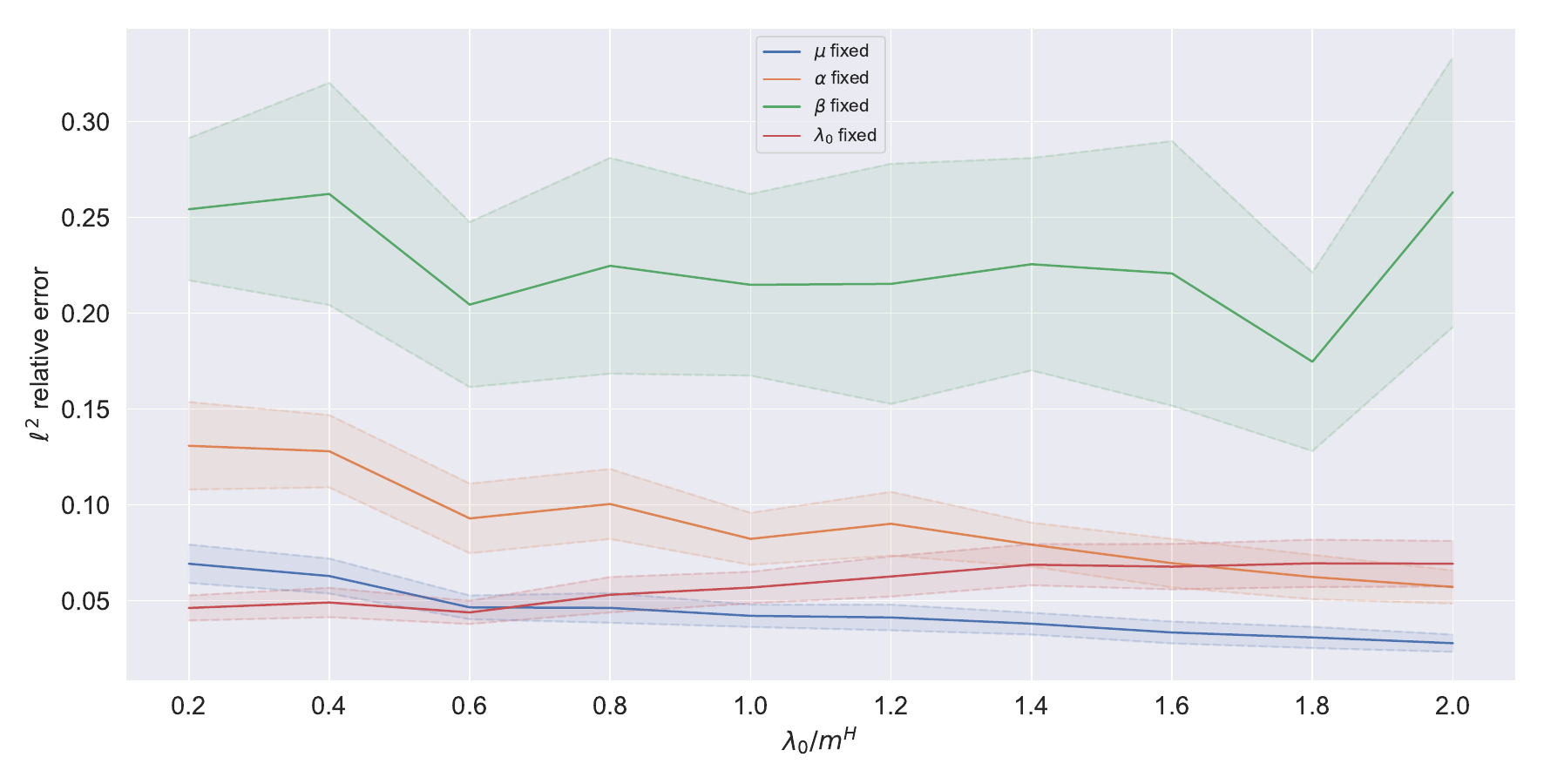}
        \caption{Relative estimation error with confidence bands ($\pm 1.96$ empirical standard deviation) respectively for \(\mu\), \(\alpha\), \(\beta\) and \(\lambda_0\) fixed 
        with respect to the noise-to-signal ratio for the maximal horizon $T=8000$.}
        \label{fig:errors_wrt_noise_N}
        \end{figure}
        
		First, it seems that the value of $\lambda_0$ has a low impact on the quality of estimations, as we cannot see a clear trend in the four curves displayed in Figure~\ref{fig:errors_wrt_noise_N}.
        Taking a step back,
        the case when \(\lambda_0\) is fixed (red curve -- for which only the parameters of the Hawkes process have to be estimated) looks slightly increasing,
        which is probably due to the increase of points coming from the Poisson process (as \(\lambda_0\) grows), thus leading to a more difficult estimation of the parameters of the Hawkes process.
        On the contrary, the cases where \(\alpha\) and \(\mu\) are fixed (orange and blue curves) look to have a slight downward trend, suggesting that the better estimation of \(\lambda_0\) dominates the worse estimations of the parameters of the Hawkes process.
   
    Second, when $\beta$ is fixed, the average error is substantially larger than when any of the other parameters is fixed.
     This could be explained by a compensation phenomenon inside the triplet $(\mu, \alpha,\lambda_0)$ which occurs as our method implicitly adjusts the estimation to the mean intensity of the noisy Hawkes process: 
     \[m^N = \lambda_0 + \frac{\mu}{1-\alpha},\]
     which is indeed a quantity independent of $\beta$.
    
    This numerical compensation is illustrated in Figure~\ref{fig:uni_compensation_beta_N}, where we can see that overestimating $\mu$ is systematically balanced by underestimating $\alpha$ and $\lambda_0$ and vice versa, whereas the estimated mean intensities remain accurate.
    In this experiment, the level of noise has been arbitrarily fixed to $\lambda_0 = 1.2$ but the observed behaviour appears similarly for all possible values of $\lambda_0$ used in the previous section.

    \begin{figure}[!ht]
        \centering
        \includegraphics[width=0.8\textwidth]{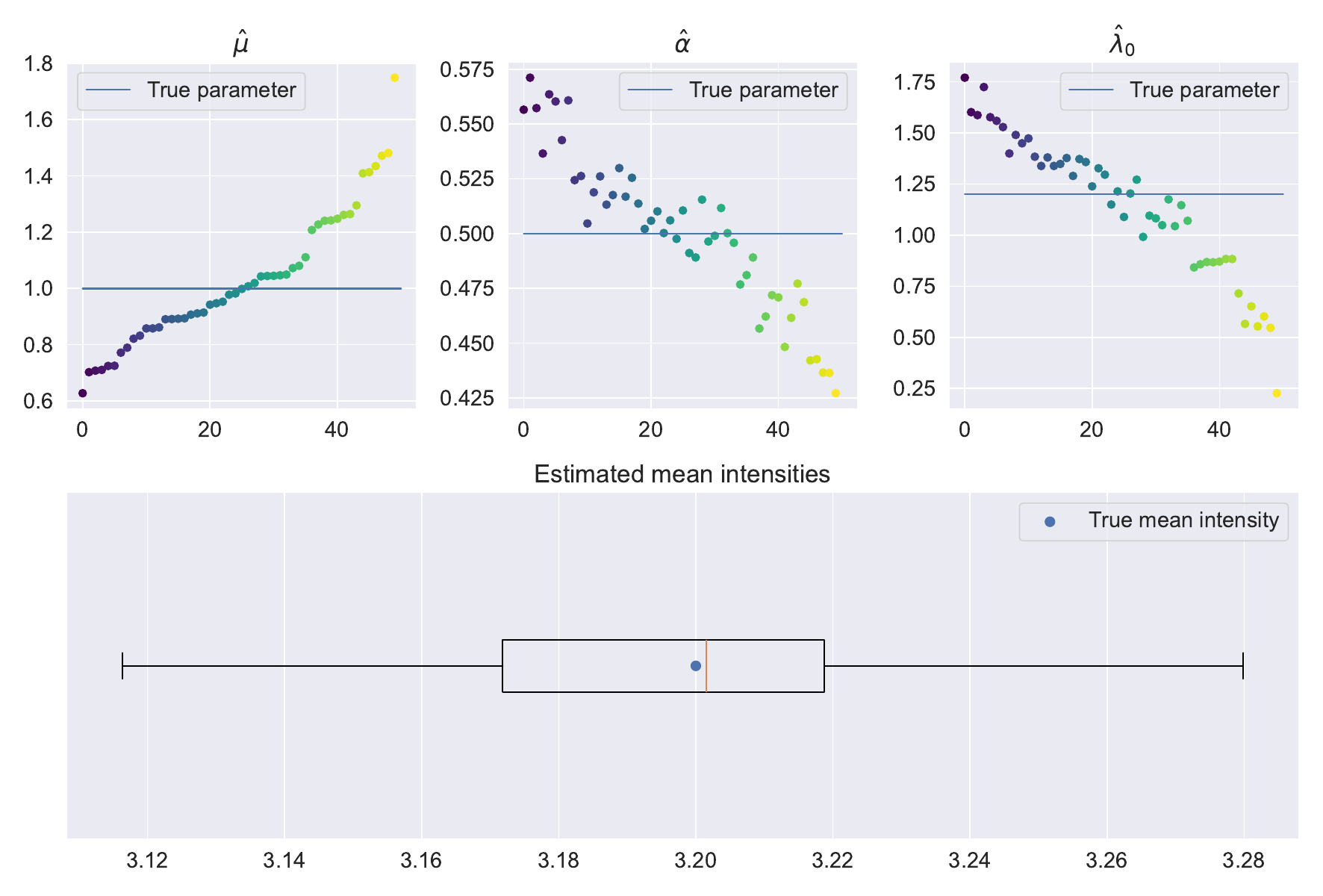}
        \caption{Estimations $\hat \mu, \hat \alpha$ and $\hat \lambda_0$ (of $\mu$, $\alpha$, and $\lambda_0$) for
        $\lambda_0 = 1.2$ when $\beta$ is fixed,
        sorted by the values of $\hat{\mu}$ (top).
        In all plots, each color corresponds to one of the 50 repetitions. Boxplot of estimated mean intensities (bottom).}
        \label{fig:uni_compensation_beta_N}
        \end{figure}

    When performing estimation with $\alpha$ fixed (the case for which the average error is the second largest, as illustrated in Figure~\ref{fig:errors_wrt_noise_N}), the compensation appears only between $\hat \mu$ and $\hat \lambda_0$ whereas $\hat \beta$ does not seem impacted, as shown in Figure~\ref{fig:uni_compensation_alpha_N}.

        \begin{figure}[!ht]
        \centering
        \includegraphics[width=0.8\textwidth]{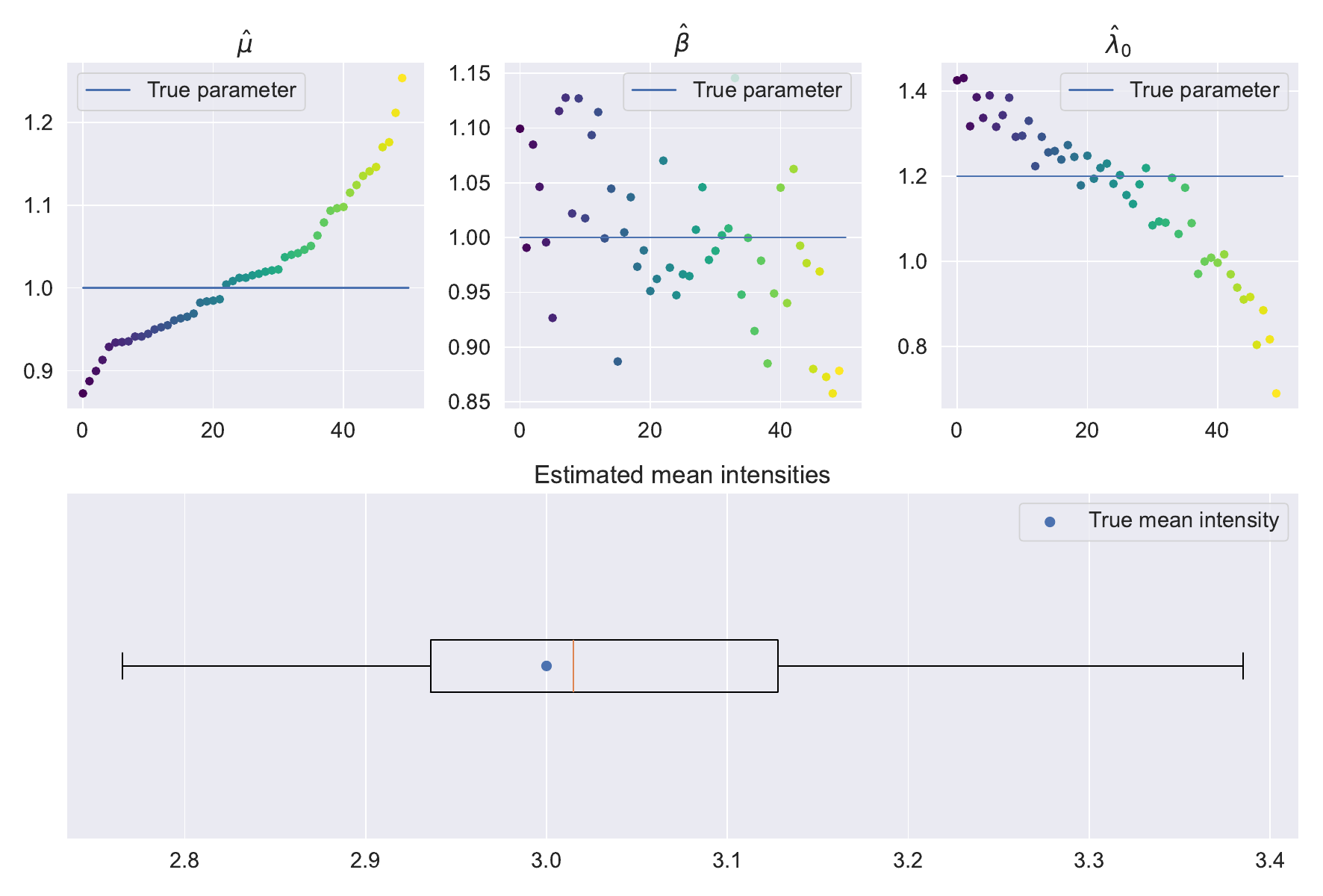}
        \caption{Estimations $\hat \mu, \hat \beta$ and $\hat \lambda_0$ (of $\mu$, $\beta$, and $\lambda_0$) for
        $\lambda_0 = 1.2$ when $\alpha$ is fixed,
        sorted by the values of $\hat{\mu}$ (top).
        In all plots, each color corresponds to one of the 50 repetitions. Boxplot of estimated mean intensities (bottom).}
        \label{fig:uni_compensation_alpha_N}
        \end{figure}

	\subsection{Bivariate setting} \label{sec:bivariate_numerical_results}
		This section illustrates numerical estimation of bivariate exponential noisy Hawkes processes (see Section~\ref{sec:dim2}) when conditions of identifiability are met (Proposition~\ref{PROPOSITION:BI_IDENTIFIABLE}). We carry out two different studies, exploring different scenarios: Section~\ref{sec:bi_alpha_parameter} studies the influence of the strength of the cross-interaction between the two subprocesses and Section~\ref{sec:bi_two_scenarios} investigates the performance of the estimator with and without knowledge of the null components. Indeed, since identifiability conditions depend on knowing which components are non-null, an information that is unlikely to be available in practical applications, we compare the performance of the estimator for both the reduced model $\mathcal{Q}_{\Lambda}$, where the null components are known, and the complete model, \[\mathcal Q
        = \left\{
          \mathbf{f}_{\theta}^N : \RR \to \mathbb C^{2 \times 2},
          \theta = (\mu, \alpha, \beta, \lambda_0)
          \in \RR_{>0}^2\times \RR_{\geq 0}^{2\times 2} \times \RR_{>0}^2 \times \RR_{>0}, \rho(\alpha) < 1,
          \beta \in \Omega_\alpha
       	 \right\}\,,\]with no prior information.
		
    Throughout this section, we consider a Hawkes process with $\mu = \begin{pmatrix} 1.0 \\ 1.0 \end{pmatrix}$ and $\beta = \begin{pmatrix} 1.0 \\ 1.3 \end{pmatrix}$.
    In addition, fortified by the analysis of the univariate setting,
    the Poisson intensity is chosen to be $\lambda_0 = 0.5$
    (the level of noise does not appear to have a significant impact on the quality of estimations, Figure~\ref{fig:errors_wrt_noise_N})
    and it is considered $M = N(T)$ (which provides accurate estimations in a reasonable amount of time, Figure~\ref{fig:errors_wrt_T}).

		\subsubsection{Influence of the cross-interaction}
		\label{sec:bi_alpha_parameter} 

        Let us consider one of the identifiable scenarios where the only non-null interaction in the Hawkes process is one of the two cross-interactions (see Proposition~\ref{PROPOSITION:BI_IDENTIFIABLE}, Situation~\ref{HYP:BI_IDENTIFIABLE_1}).
        More precisely, we consider the reduced model $\mathcal{Q}_{\Lambda}$,
        where
        \[
          \Lambda = \left\{ \begin{pmatrix} 0 & 0 \\ \alpha_{21} & 0 \end{pmatrix} : \alpha_{21} > 0 \right\} \,.
        \]
        The Hawkes process is then simulated with different levels of cross-interaction: $\alpha_{21} \in \{0.2, 0.4, 0.6, 0.8\}$,
        and estimations are obtained by optimising the spectral log-likelihood on the non-null parameters $\mu_1$, $\mu_2$, $\alpha_{21}$, $\beta_2$, and $\lambda_0$.
        
        Figure~\ref{fig:bi_phase_transition} illustrates the influence of the true parameter $\alpha_{21}$ on the quality of the estimations,
	      through the relative error of the estimations for the different values of $\alpha_{21}$ and 
	      an increasing range of horizons $T$.
	      As a complement to what has been observed in Figure~\ref{fig:errors_wrt_T},
	      our estimator appears to behave particularly well for higher values of $T$, but also for higher values of $\alpha_{21}$.

        \begin{figure}[!ht]
			\centering
			\includegraphics[width=0.8\textwidth]{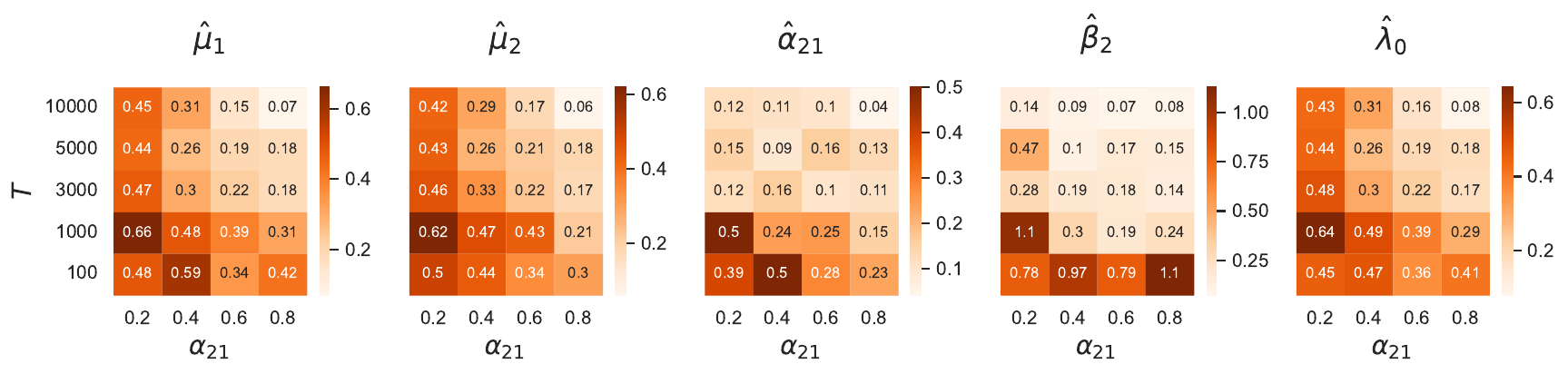}
 		   	\caption{Heatmap of relative errors for each estimation $\hat\mu_1$, $\hat\mu_2$, $\hat\alpha_{21}$, $\hat\beta_2$, and $\hat\lambda_0$ for different levels of interaction $\alpha_{21}$ (x-axis) and horizons $T$ (y-axis).}
 		  \label{fig:bi_phase_transition}
		\end{figure}
        
        This is not surprising since, for smaller values of $\alpha_{21}$, the Hawkes process behaves closely to a homogeneous Poisson process, and as proven in Proposition~\ref{PROPOSITION:BI_NON_IDENTIFIABLE}, the superposition of two Poisson processes leads to a non-identifiable model. Lower interactions necessitate then higher values of $T$ to obtain satisfactory results. Inversely, for average and high interaction magnitudes, we start to obtain small errors for horizon values around $T=3000$.

		\subsubsection{Influence of null interactions}
		\label{sec:bi_two_scenarios}

        In this section we simulate $50$ repetitions with a fixed horizon $T=3000$ for two identifiable scenarios regarding the Hawkes process $H = (H_1, H_2)$.
        \begin{description}
          \item[Scenario 1] The matrix of interactions is:
          \[
            \alpha = \begin{pmatrix} 0.5 & 0 \\ 0.4 & 0 \end{pmatrix}\,,
          \]
          corresponding to Proposition~\ref{PROPOSITION:BI_IDENTIFIABLE}, Situation~\ref{HYP:BI_IDENTIFIABLE_1}
          In other terms, $H_1$ excites both subprocesses whereas $H_2$ has no influence on the dynamics (See Figure~\ref{fig:bi_scenario}, left).
          \item[Scenario 2] The matrix of interactions is:
          \[
            \alpha = \begin{pmatrix} 0.5 & 0 \\ 0.4 & 0.4 \end{pmatrix}\,,
          \]
          corresponding to Proposition~\ref{PROPOSITION:BI_IDENTIFIABLE}, Situation~\ref{HYP:BI_IDENTIFIABLE_2}
          In other terms, $H_1$ excites both subprocesses and $H_2$ excites itself (See Figure~\ref{fig:bi_scenario}, right).
        \end{description}
        		
        \begin{figure}[!ht]
  			\centering
			  \includegraphics[width=0.4\linewidth]{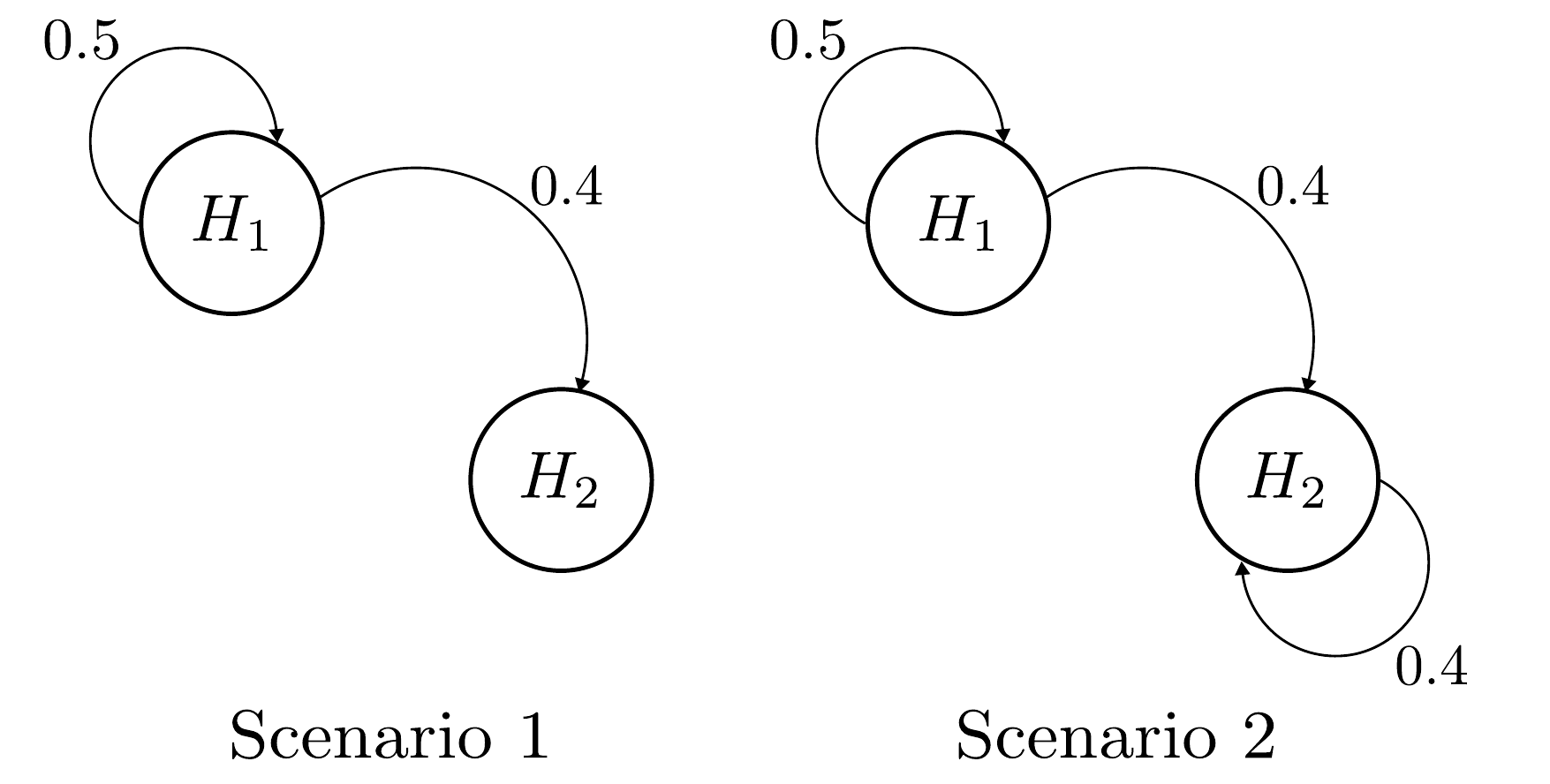}
 		   	\caption{
 		   	Interactions in the two numerical scenarios considered.
 		   	}
 		  	\label{fig:bi_scenario}
	   	\end{figure}

      Graphics in the left column of Figure~\ref{fig:column_triangle_model_estimation} present the boxplots of each parameter estimation when considering their respective reduced models $\mathcal{Q}_{\Lambda}$.		
      These results show that our method provides unbiased estimates of all parameters and is particularly efficient at inferring the interaction matrix $\alpha$ (estimations have very low variance).
      The larger variances are observed for parameters $\mu_1$, $\mu_2$ and $\lambda_0$, which is probably due to compensation effects already mentioned in Section~\ref{sec:influence_noise}.

      \begin{figure}[!ht]
        \centering
        \includegraphics[width=0.8\textwidth]{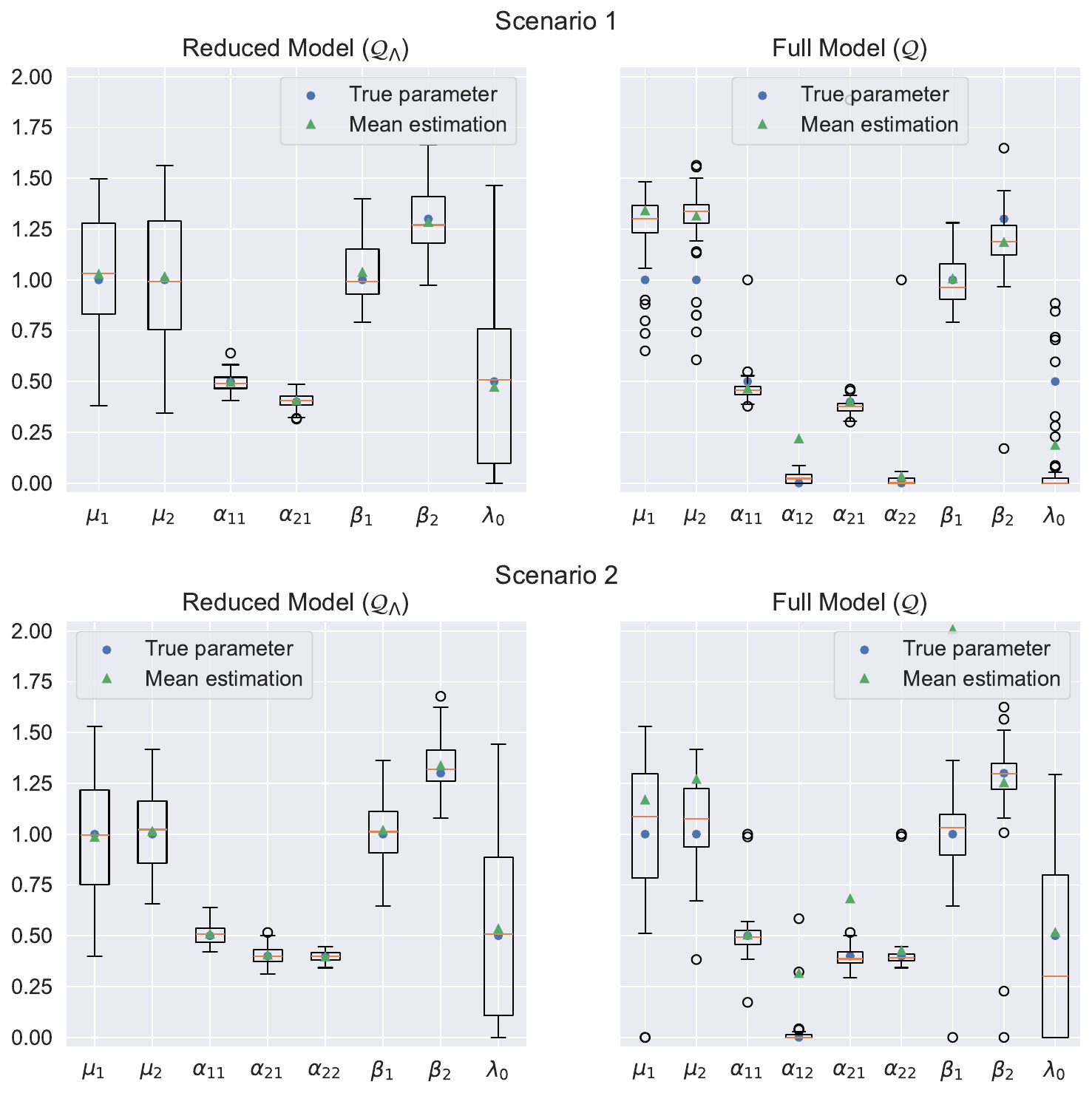}
        \caption{Boxplots of parameter estimations in the reduced model $\mathcal{Q}_{\Lambda}$ (left) and full model $\mathcal{Q}$ (right)
        in Scenario~1 (top) and Scenario~2 (bottom) ($50$ trials).
        Average estimation (green triangle) is to be compared to true parameter (blue point).}
        \label{fig:column_triangle_model_estimation}
      \end{figure}
            
            Estimating in the reduced model $\mathcal Q_\Lambda$ requires prior knowledge on the null parameters $\alpha_{ij}$ ($1 \le i, j \le 2$), which is unlikely in practical applications. Therefore, we compare the results obtained in the reduced model to those in the complete model $\mathcal Q$ (see the right column graphics in Figure~\ref{fig:column_triangle_model_estimation}). If the estimates of $\alpha$ and $\beta$ seem still empirically unbiased, we observed several deteriorations compared to the previous results. First, we notice a bias in the estimates of $\mu_1$, $\mu_2$ and $\lambda_0$: more precisely, $\mu_1$ and $\mu_2$ are overestimated while $\lambda_0$ is underestimated in both scenarios. Moreover, we observe in Scenario 2 some outlier estimations for the $\alpha_{ij}$ ($1 \le i, j \le 2$) coefficients, which did not appear when considering the reduced model. 
		
		Fortunately, Figure~\ref{fig:column_triangle_model_estimation} also suggests that our estimator is able to detect the null interactions in the full model, which allows to re-estimate the parameters in the reduced model.
    To do so, we propose to look at the $5\%$-empirical quantile and proportion of null estimations of each term of the estimated interaction matrix $\hat \alpha$, which are summarised in Table~\ref{tab:full_model_quantile}. 	We can observe that these empirical quantiles allow to distinguish between null and non-null parameters: indeed, we see that when the true parameter is non-null, at least $95\%$ of the estimations are also non-null, which does not happen when the true parameter is null. If the choice of the $5\%$-threshold can seem arbitrary, the proportion of null estimations (in brackets) suggests that the procedure is robust with respect to the choice of threshold. Indeed, at least $28\%$ of each null coefficient are actually estimated to $0$, while none of the non-null coefficient has a single null estimate.
		This suggests that when enough repetitions are available, it is possible to use these empirical quantiles to estimate the null interactions.
		An estimation procedure when no prior information is known about the support of the interaction graph would then consist in a three-step approach:
		first, estimating all parameters in the complete model $\mathcal{Q}$;
		second, computing the $5\%$-quantiles for all estimated interactions parameters in matrix $\hat \alpha$,
		and defining the support of $\alpha$ to be entries corresponding to a positive empirical quantile;
		finally, re-estimating all parameters in the reduced model defined by the support of $\alpha$.
		Let us remark that the proposed support estimation step boils down to correspond to a multiple test that,
		when the noisy Hawkes process has significantly more that $2$ dimensions,
		can be corrected thanks to usual procedures such as Bonferroni and Benjamini-Hochberg methods.
                
        \begin{table}[!ht]
          \centering
          \begin{tabular}{c @{\hspace{7mm}} cccc}
            \toprule
            & $\alpha_{11}$ & $\alpha_{12}$ & $\alpha_{21}$ & $\alpha_{22}$ \\
            \toprule
            Scenario 1 & 0.40 (0.00) & 0.00 (0.28) & 0.32 (0.00) & 0.00 (0.48) \\
            Scenario 2 & 0.41 (0.00) & 0.00 (0.56) & 0.34 (0.00) & 0.35 (0.00) \\
          \end{tabular}
          \caption{$5\%$-empirical quantile (in bracket, proportion of null estimations) of each parameter of the estimated interaction matrix $\hat \alpha$. The true parameter $\alpha_{12}$ is null in both scenarios and $\alpha_{22}$ is null only in Scenario 1.}
          \label{tab:full_model_quantile}
        \end{table}

        \begin{remark}
       \label{rem:subsampling}
            If there is no available replicates but the number of event times is large, an alternative approach consists in partitioning the observation window $[0, T]$ and performing the estimation procedure on each subsample. The corresponding results are given in Table~\ref{tab:partition_quantile}, which shows similar results to those obtained with replicates.
        \end{remark}

	    \begin{table}[!ht]
          \centering
          \begin{tabular}{c @{\hspace{7mm}} cccc}
            \toprule
            & $\alpha_{11}$ & $\alpha_{12}$ & $\alpha_{21}$ & $\alpha_{22}$ \\
            \toprule
            Partitioned Scenario 1 & 0.34 (0.00) & 0.00 (0.50) & 0.27 (0.00) & 0.00 (0.45) \\
            Partitioned Scenario 2 & 0.26 (0.00) & 0.00 (0.60) & 0.26 (0.00) & 0.30 (0.00) \\
          \end{tabular}
          \caption{$5\%$-quantile (in bracket, proportion of null estimations) of each parameter of the estimated interaction matrix $\hat \alpha$ for $20$ subsampled replicates from a single observation in $[0, 6000]$. Observation window is partitioned in equally-sized windows of length $300$. The true parameter $\alpha_{12}$ is null in both scenarios and $\alpha_{22}$ is null only in Scenario 1.}
          \label{tab:partition_quantile}
    \end{table}

    To conclude and in order to numerically ground the three-step approach proposed above,
    we provide an additional numerical experiment in Appendix~\ref{appendix:additional_experiment},
    which quantifies the proportion of correct estimations of the support of the interaction matrix \(\alpha\) for a wide range of true parameters randomly chosen according to a spike-and-slab procedure.

  \section{Application to neuronal data}

        \subsection{Dataset and preprocessing}
        
        We work on a publicly available dataset analysed in \cite{zhang2023cerebellum}, which describes the neuronal activity of mice when accomplishing a navigation task. Several sessions are recorded for several mice. The events correspond to the spike times of different clusters of neurons: since the number of clusters varies from one mouse to another, but remains constant during all sessions for  each mouse, we focus on the recording of one mouse and we consider the five sessions as replicates. We select one mouse for which the number of clusters is small to remain close to the setting that we study in the paper: we then keep Mouse 22, all recordings of which contain 3 clusters of neurons. Finally, since Cluster 1 has very few observations (around $20$ in each experiment for $T=725$) compared to the two other clusters (more than $1000$ spikes for each cluster), we decide to keep only Clusters 2 and 3 for this study.
        
        \subsection{Results}

        Since the number of replicates is small but the number of spikes is large, we decide to use the partitioning procedure mentioned in the previous section: the observation window for each replicate is divided into 5 windows in order to obtain a total of $25$ subsamples.
        Before carrying out the estimation procedures, we have to make a sensible choice about parameter $M$, to make sure that the explored window of frequencies in Equation~\eqref{eq:spectral_log_likelihood} is big enough to take into account most of the real spectra. We begin by looking at the average periodogram estimated with the 25 subsamples, which is often used as an estimation of the spectral matrix $\mathbf{f}$. In Figure~\ref{fig:real_data_periodogram}, we display the diagonal terms of the periodogram $\mathbf{I}^T$. 
        By choosing $M = N(T) \log(N(T))$, the maximal frequency, on average, is $M(T) / T = 11.7$, value above which we can see that both $f_{11}$ and $f_{22}$ seem to be in an almost constant regime. We decide then to use $M = N(T) \log(N(T))$ for the rest of this study.

        \begin{figure}[!ht] 

		    \centering
		    \includegraphics[width=0.9\linewidth]{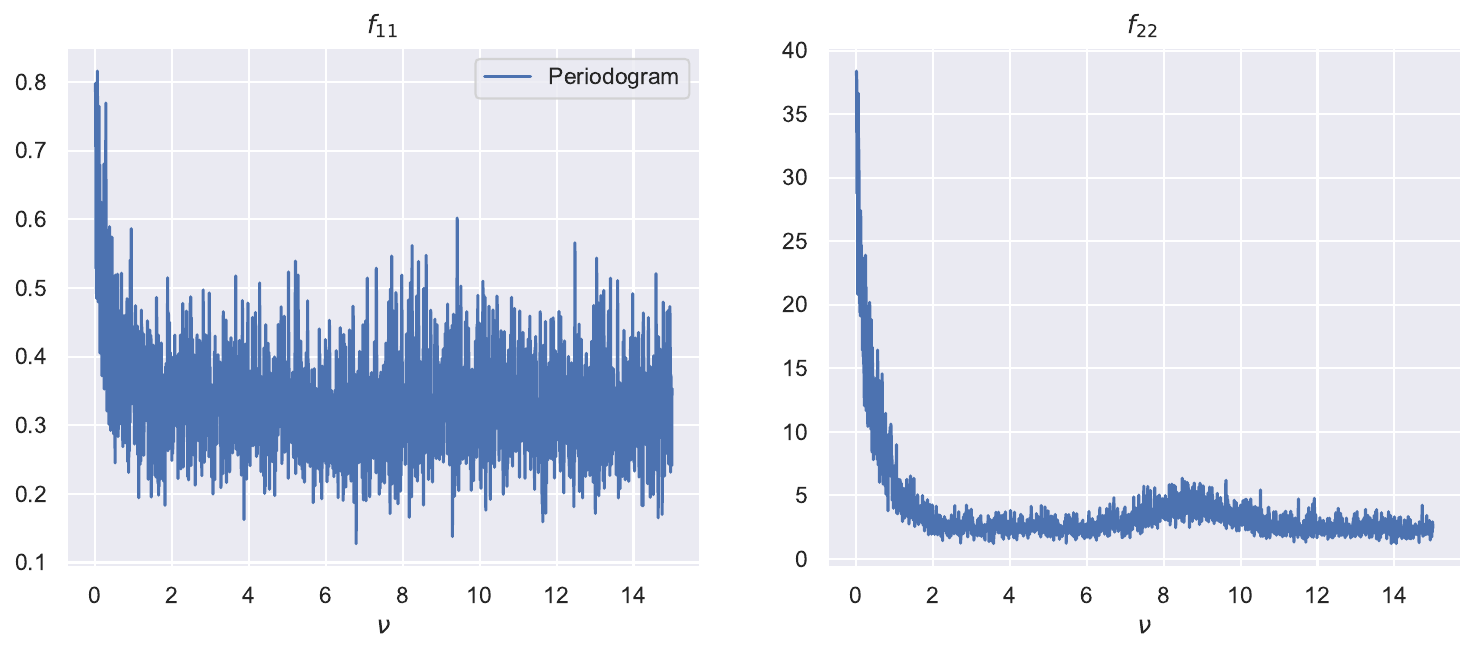}
		    \caption{Diagonal terms $I_{11}^T$ and $I^T_{22}$ of averaged periodogram for the 25 subsamples of the neuronal data.}
		    \label{fig:real_data_periodogram}
        \end{figure}
        We perform the estimation procedure on each subsample.
        Table~\ref{tab:full_model_quantile_realdata} shows the estimated quantiles and the proportion of null estimations of the interaction parameters in matrix $\hat \alpha$. Consequently, we then set $\hat \alpha_{12} = 0$ and re-estimate the parameters in the reduced model $\mathcal{Q}_{\Lambda}$.
        
        \begin{table}[!ht]
          \centering
          \begin{tabular}{c @{\hspace{7mm}} cccc}
            \toprule
            & $\alpha_{11}$ & $\alpha_{12}$ & $\alpha_{21}$ & $\alpha_{22}$ \\
            \toprule
            & $9.26\times10^{-4}$ (0.00) & 0.00 (0.08)& $2.59\times10^{-2}$ (0.00)& $8.30\times10^{-1}$ (0.00)\\
            
          \end{tabular}
          \caption{$5\%$-empirical quantile (in bracket, proportion of null estimations) of each parameter of the interaction matrix $\hat \alpha$ for the neuronal data.}
          \label{tab:full_model_quantile_realdata}
    \end{table}

    Figure~\ref{fig:real_data_application} displays the estimations obtained within the full model $\mathcal{Q}$ and the reduced model $\mathcal{Q}_\Lambda$ where the coefficient $\hat \alpha_{12} = 0$ is set to $0$. We observe that the noise $\lambda_0$ is estimated to a non-null value in almost all subsamples, which illustrates the interest of modeling the noise for neuronal data. 
    We can comment on the fact that we detect self-excitating behaviours for both neurons ($\alpha_{11}$ and $\alpha_{22}$ being positive), but also one cross-interaction, indicating that Neuron $1$ has an exciting effect on Neuron $2$. This supports the interest of modeling neuronal activity as a multivariate process rather than individually for each neuron. The large values for $\beta$ also indicate that the interactions from the past vanish quickly. However, these results must be considered cautiously due to both the large variances of the estimators and the very small values of some of the coefficients $\alpha$, which makes the support identification very difficult. This could be explained by some obvious limitations of the model for detecting neuronal interactions, in particular the lack of modeling for inhibiting interactions.
    We believe it still motivates the interest of further research for developing appropriate modeling of noised data.
        
        \begin{figure}[!ht] 

		    \centering
		    \includegraphics[width=\linewidth]{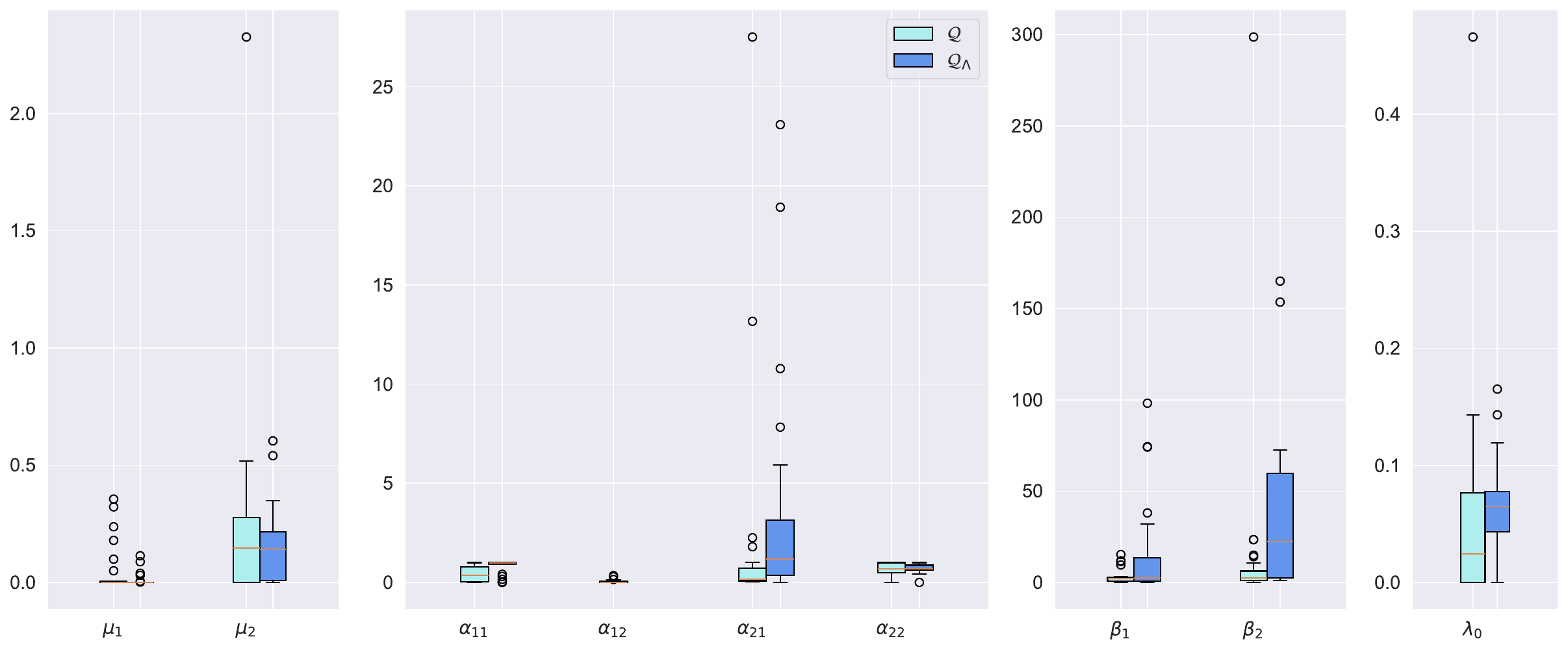}
		    \caption{Boxplots of estimated parameters for $25$ subsamples on neuronal activation data. The coefficients are estimated in the full model $\mathcal{Q}$, then re-estimated in the reduced model $\mathcal{Q}_\Lambda$ after the detection support step.}
		    \label{fig:real_data_application}
        \end{figure}

		\section{Discussion}

In this paper, we propose a spectral approach to estimate the parameters of a noisy Hawkes process, the performance of which is illustrated on an extended numerical study. Although we highlight the great benefit of considering a spectral analysis when standard inference methods are not available, we also bring out  identifiability issues that may arise, either from the model itself or from the spectral approach. While we exhibit several identifiable and non-identifiable scenarios in both univariate and bivariate contexts, a general result on identifiability is still to be established, in particular in higher dimensions. For this purpose, the number of non-null cross-interactions and the choice of the kernel functions appear to be key elements in order to obtain identifiability guarantees.

More generally, we believe that the spectral analysis can provide efficient estimators in many frameworks of inaccurately or partially observed data. A natural extension of this work is to consider alternative forms for the noise process, for instance an heterogeneous Poisson process. Another topic of interest is to investigate other mechanisms for the noise when some points are randomly missing. This would be complementary to our work since it would allow to model both false positive and false negative occurrences. In practice, this could be of great interest for applications, especially for the tracking of epidemics. Finally, let us mention that this paper focuses on the linear Hawkes model, which excludes notably inhibition phenomena. Though the nonlinear framework is particularly interesting for many applications, for instance in neuroscience, all the spectral theory for point processes only exists in a linear context so that we believe that developing a spectral inference procedure for nonlinear processes would be very challenging and remains a widely open topic.

\section*{Acknowledgments}

This work is partially supported by the French Agence Nationale de la Recherche (ANR),
project under reference ANR-23-CE40-0007.

%\backmatter
%\bmsection*{Author contributions}

%This is an author contribution text. This is an author contribution text. This is an author contribution text. This is an author contribution text. This is an author contribution text.

\bibliography{bibliography}

\begin{thebibliography}{59}
\expandafter\ifx\csname natexlab\endcsname\relax\def\natexlab#1{#1}\fi
\expandafter\ifx\csname url\endcsname\relax
  \def\url#1{{\tt #1}}\fi
\expandafter\ifx\csname urlprefix\endcsname\relax\def\urlprefix{URL }\fi

\bibitem[\protect\citeauthoryear{Adamopoulos}{Adamopoulos}{1976}]{Adamopoulos1976}
Adamopoulos, L. (1976), \enquote{Cluster models for earthquakes: Regional
  comparisons,} {\em Journal of the International Association for Mathematical
  Geology\/}, 8, 463--475.

\bibitem[\protect\citeauthoryear{Antoniadis and Bigot}{Antoniadis and
  Bigot}{2006}]{Antoniadis2006}
Antoniadis, A. and Bigot, J. (2006), \enquote{Poisson inverse problems,} {\em
  Annals of Statistics\/}, 34, 2132--2158.

\bibitem[\protect\citeauthoryear{Bacry, Bompaire, Gaïffas, and Muzy}{Bacry
  et~al.}{2020}]{Bacry2020}
Bacry, E., Bompaire, M., Gaïffas, S., and Muzy, J. (2020), \enquote{Sparse and
  low-rank multivariate {H}awkes processes,} {\em Journal of Machine Learning
  Research\/}, 21, 1--32.

\bibitem[\protect\citeauthoryear{Bacry, Mastromatteo, and Muzy}{Bacry
  et~al.}{2015}]{Bacry2015}
Bacry, E., Mastromatteo, I., and Muzy, J.-F. (2015), \enquote{Hawkes Processes
  in Finance,} {\em Market Microstructure and Liquidity\/}, 01, 1550005.

\bibitem[\protect\citeauthoryear{Bacry and Muzy}{Bacry and
  Muzy}{2016}]{Bacry2016}
Bacry, E. and Muzy, J. (2016), \enquote{First- and Second-Order Statistics
  Characterization of {H}awkes Processes and Non-Parametric Estimation,} {\em
  IEEE Transactions on Information Theory\/}, 62, 2184--2202.

\bibitem[\protect\citeauthoryear{Barnett}{Barnett}{2021}]{Barnett2021}
Barnett, A.~H. (2021), \enquote{Aliasing error of the exp($\beta \sqrt{1 -
  z^2}$) kernel in the nonuniform fast Fourier transform,} {\em Applied and
  Computational Harmonic Analysis\/}, 51, 1--16.

\bibitem[\protect\citeauthoryear{Barnett, Magland, and Klinteberg}{Barnett
  et~al.}{2019}]{Barnett2019}
Barnett, A.~H., Magland, J.~F., and Klinteberg, L.~a. (2019), \enquote{A
  parallel non-uniform fast Fourier transform library based on an
  “exponential of semicircle” kernel,} {\em SIAM Journal on Scientific
  Computing\/}, 41, C479--C504.

\bibitem[\protect\citeauthoryear{Bartlett}{Bartlett}{1963}]{Bartlett1963}
Bartlett, M.~S. (1963), \enquote{The Spectral Analysis of Point Processes,}
  {\em Journal of the Royal Statistical Society. Series B (Methodological)\/},
  25, 264--296.

\bibitem[\protect\citeauthoryear{Bartlett}{Bartlett}{1964}]{Bartlett1964}
--- (1964), \enquote{The Spectral Analysis of Two-Dimensional Point Processes,}
  {\em Biometrika\/}, 51, 299--311.

\bibitem[\protect\citeauthoryear{Bonnet, Lacour, Picard, and Rivoirard}{Bonnet
  et~al.}{2022}]{Bonnet2022}
Bonnet, A., Lacour, C., Picard, F., and Rivoirard, V. (2022), \enquote{Uniform
  deconvolution for Poisson point processes,} {\em Journal of Machine Learning
  Research\/}, 23.

\bibitem[\protect\citeauthoryear{Bonnet, Martinez~Herrera, and Sangnier}{Bonnet
  et~al.}{2023}]{Bonnet2023}
Bonnet, A., Martinez~Herrera, M., and Sangnier, M. (2023), \enquote{Inference
  of multivariate exponential {H}awkes processes with inhibition and
  application to neuronal activity,} {\em Statistics and Computing\/}, 33.

\bibitem[\protect\citeauthoryear{Bradley}{Bradley}{2005}]{Bradley2005}
Bradley, R.~C. (2005), \enquote{{Basic properties of strong mixing conditions.
  A survey and some open questions},} {\em Probability Surveys\/}, 2, 107--144.

\bibitem[\protect\citeauthoryear{Brillinger}{Brillinger}{2012}]{Brillinger2012}
Brillinger, D.~R. (2012), {\em Statistical Inference for Stationary Point
  Processes\/}, New York, NY: Springer New York, 499--543.

\bibitem[\protect\citeauthoryear{Brémaud and Massoulié}{Brémaud and
  Massoulié}{1996}]{Bremaud1996}
Brémaud, P. and Massoulié, L. (1996), \enquote{Stability of nonlinear
  {H}awkes processes,} {\em The Annals of Probability\/}, 24, 1563--1588.

\bibitem[\protect\citeauthoryear{Byrd, Lu, Nocedal, and Zhu}{Byrd
  et~al.}{1995}]{Byrd1995}
Byrd, R., Lu, P., Nocedal, J., and Zhu, C. (1995), \enquote{A Limited Memory
  Algorithm for Bound Constrained Optimization,} {\em SIAM Journal on
  Scientific and Statistical Computing\/}, 16, 1190--1208.

\bibitem[\protect\citeauthoryear{Chevallier, Duarte, Löcherbach, and
  Ost}{Chevallier et~al.}{2019}]{Chevallier2019}
Chevallier, J., Duarte, A., Löcherbach, E., and Ost, G. (2019), \enquote{Mean
  field limits for nonlinear spatially extended Hawkes processes with
  exponential memory kernels,} {\em Stochastic Processes and their
  Applications\/}, 129, 1--27.

\bibitem[\protect\citeauthoryear{Cheysson and Lang}{Cheysson and
  Lang}{2022}]{Cheysson2022}
Cheysson, F. and Lang, G. (2022), \enquote{{Spectral estimation of Hawkes
  processes from count data},} {\em The Annals of Statistics\/}, 50, 1722 --
  1746.

\bibitem[\protect\citeauthoryear{Chornoboy, Schramm, and Karr}{Chornoboy
  et~al.}{1988}]{Chornoboy1988}
Chornoboy, E.~S., Schramm, L.~P., and Karr, A.~F. (1988), \enquote{Maximum
  likelihood identification of neural point process systems,} {\em Biological
  Cybernetics\/}, 59, 265--275.

\bibitem[\protect\citeauthoryear{Da~Fonseca and Zaatour}{Da~Fonseca and
  Zaatour}{2013}]{DaFonseca2013}
Da~Fonseca, J. and Zaatour, R. (2013), \enquote{{H}awkes Process: Fast
  Calibration, Application to Trade Clustering, and Diffusive Limit,} {\em
  Journal of Futures Markets\/}, 34, 548--579.

\bibitem[\protect\citeauthoryear{Daley and Vere-Jones}{Daley and
  Vere-Jones}{2003}]{DaleyV1}
Daley, D. and Vere-Jones, D. (2003), {\em An introduction to the theory of
  point processes. {V}ol. {I}\/}, New {Y}ork: Springer-Verlag, second edition.

\bibitem[\protect\citeauthoryear{Daley}{Daley}{1971}]{Daley1971}
Daley, D.~J. (1971), \enquote{Weakly Stationary Point Processes and Random
  Measures,} {\em Journal of the Royal Statistical Society. Series B
  (Methodological)\/}, 33, 406--428.

\bibitem[\protect\citeauthoryear{Dassios and Zhao}{Dassios and
  Zhao}{2011}]{Dassios2011}
Dassios, A. and Zhao, H. (2011), \enquote{{A dynamic contagion process},} {\em
  Advances in Applied Probability\/}, 43, 814--846.

\bibitem[\protect\citeauthoryear{Davis}{Davis}{1984}]{Davis1984}
Davis, M. H.~A. (1984), \enquote{{Piecewise-Deterministic Markov Processes : A
  General Class of Non-Diffusion Stochastic Models},} {\em Journal of the Royal
  Statistical Society Series B (Methodological)\/}, 46, 353--388.

\bibitem[\protect\citeauthoryear{Deutsch and Ross}{Deutsch and
  Ross}{2020}]{Deutsch2020}
Deutsch, I. and Ross, G. (2020), \enquote{ABC Learning of Hawkes Processes with
  Missing or Noisy Event Times,} Preprint at
  \url{https://arxiv.org/abs/2006.09015}.

\bibitem[\protect\citeauthoryear{Duarte, L{\"{o}}cherbach, and Ost}{Duarte
  et~al.}{2019}]{Duarte2019}
Duarte, A., L{\"{o}}cherbach, E., and Ost, G. (2019), \enquote{{Stability,
  convergence to equilibrium and simulation of non-linear Hawkes processes with
  memory kernels given by the sum of Erlang kernels},} {\em ESAIM: Probability
  and Statistics\/}, 23, 770--796.

\bibitem[\protect\citeauthoryear{Düker and Pipiras}{Düker and
  Pipiras}{2019}]{Duker2019}
Düker, M.-C. and Pipiras, V. (2019), \enquote{Asymptotic results for
  multivariate local Whittle estimation with applications,} in {\em 2019 IEEE
  8th International Workshop on Computational Advances in Multi-Sensor Adaptive
  Processing (CAMSAP)\/}.

\bibitem[\protect\citeauthoryear{Embrechts, Liniger, and Lin}{Embrechts
  et~al.}{2011}]{Embrechts2011}
Embrechts, P., Liniger, T., and Lin, L. (2011), \enquote{Multivariate Hawkes
  processes: an application to financial data,} {\em Journal of Applied
  Probability\/}, 48, 367–378.

\bibitem[\protect\citeauthoryear{Guo, Hu, Xu, and Zhang}{Guo
  et~al.}{2018}]{Guo2018}
Guo, X., Hu, A., Xu, R., and Zhang, J. (2018), \enquote{Consistency and
  Computation of Regularized MLEs for Multivariate {H}awkes Processes,}
  Preprint at \url{https://arxiv.org/abs/1810.02955}.

\bibitem[\protect\citeauthoryear{Gupta, Farajtabar, Dilkina, and Zha}{Gupta
  et~al.}{2018}]{Gupta2018}
Gupta, A., Farajtabar, M., Dilkina, B., and Zha, H. (2018), \enquote{Discrete
  Interventions in Hawkes Processes with Applications in Invasive Species
  Management,} in {\em Proceedings of the Twenty-Seventh International Joint
  Conference on Artificial Intelligence, {IJCAI-18}\/}, International Joint
  Conferences on Artificial Intelligence Organization.

\bibitem[\protect\citeauthoryear{Hawkes}{Hawkes}{1971{\natexlab{a}}}]{Hawkes1971b}
Hawkes, A.~G. (1971{\natexlab{a}}), \enquote{{Point Spectra of Some Mutually
  Exciting Point Processes},} {\em Journal of the Royal Statistical Society.
  Series B (Methodological)\/}, 33, 438--443.

\bibitem[\protect\citeauthoryear{Hawkes}{Hawkes}{1971{\natexlab{b}}}]{Hawkes1971}
--- (1971{\natexlab{b}}), \enquote{Spectra of Some Self-Exciting and Mutually
  Exciting Point Processes,} {\em Biometrika\/}, 58, 83--90.

\bibitem[\protect\citeauthoryear{Hohage and Werner}{Hohage and
  Werner}{2016}]{Hohage2016}
Hohage, T. and Werner, F. (2016), \enquote{Inverse problems with {P}oisson
  data: statistical regularization theory, applications and algorithms,} {\em
  Inverse Problems\/}, 32, 093001.

\bibitem[\protect\citeauthoryear{Karavasilis and Rigas}{Karavasilis and
  Rigas}{2007}]{Karavasilis2007}
Karavasilis, G.~J. and Rigas, A.~G. (2007), \enquote{Spectral analysis
  techniques of stationary point processes used for the estimation of
  cross-correlation: Application to the study of a neurophysiological system,}
  in {\em 2007 15th European Signal Processing Conference\/}.

\bibitem[\protect\citeauthoryear{Kirchner}{Kirchner}{2017}]{Kirchner2017}
Kirchner, M. (2017), \enquote{An estimation procedure for the {H}awkes
  process,} {\em Quantitative Finance\/}, 17, 571--595.

\bibitem[\protect\citeauthoryear{Lambert, Tuleau-Malot, Bessaih, Rivoirard,
  Bouret, Leresche, and Reynaud-Bouret}{Lambert et~al.}{2018}]{Lambert2018}
Lambert, R., Tuleau-Malot, C., Bessaih, T., Rivoirard, V., Bouret, Y.,
  Leresche, N., and Reynaud-Bouret, P. (2018), \enquote{Reconstructing the
  functional connectivity of multiple spike trains using Hawkes models,} {\em
  Journal of Neuroscience Methods\/}, 297, 9--21.

\bibitem[\protect\citeauthoryear{Lewis and Mohler}{Lewis and
  Mohler}{2011}]{Lewis2011}
Lewis, E. and Mohler, G. (2011), \enquote{A Nonparametric EM Algorithm for
  Multiscale {H}awkes Processes,} {\em Journal of Nonparametric Statistics\/},
  1, 1--20.

\bibitem[\protect\citeauthoryear{Loison, Staerman, and Moreau}{Loison
  et~al.}{2025}]{staerman2024}
Loison, V., Staerman, G., and Moreau, T. (2025), \enquote{{UNH}aP: Unmixing
  Noise from Hawkes Process to Model Physiological Events,} in {\em The 28th
  International Conference on Artificial Intelligence and Statistics\/}.

\bibitem[\protect\citeauthoryear{Lund and Rudemo}{Lund and
  Rudemo}{2000}]{Lund2000}
Lund, J. and Rudemo, M. (2000), \enquote{Models for Point Processes Observed
  with Noise,} {\em Biometrika\/}, 87, 235--249.

\bibitem[\protect\citeauthoryear{Martinez~Herrera}{Martinez~Herrera}{2025}]{Code}
Martinez~Herrera, M. (2025), \enquote{noise-hawkes-process,} {v}1.0.0. Zenodo.
  https://doi.org/10.5281/zenodo.16737873.

\bibitem[\protect\citeauthoryear{Mei, Qin, and Eisner}{Mei
  et~al.}{2019}]{mei2019}
Mei, H., Qin, G., and Eisner, J. (2019), \enquote{Imputing Missing Events in
  Continuous-Time Event Streams,} in Chaudhuri, K. and Salakhutdinov, R.
  (editors), {\em Proceedings of the 36th International Conference on Machine
  Learning\/}, volume~97 of {\em Proceedings of Machine Learning Research\/},
  PMLR.

\bibitem[\protect\citeauthoryear{Mugglestone and Renshaw}{Mugglestone and
  Renshaw}{1996}]{Mugglestone1996}
Mugglestone, M. and Renshaw, E. (1996), \enquote{A practical guide to the
  spectral analysis of spatial point processes,} {\em Computational Statistics
  \& Data Analysis\/}, 21, 43--65.

\bibitem[\protect\citeauthoryear{Mugglestone and Renshaw}{Mugglestone and
  Renshaw}{2001}]{Mugglestone2001}
Mugglestone, M.~A. and Renshaw, E. (2001), \enquote{Spectral tests of
  randomness for spatial point patterns,} {\em Environmental and Ecological
  Statistics\/}, 8, 237--251.

\bibitem[\protect\citeauthoryear{Ogata}{Ogata}{1978}]{Ogata1978}
Ogata, Y. (1978), \enquote{The asymptotic behaviour of maximum likelihood
  estimators for stationary point processes,} {\em Annals of the Institute of
  Statistical Mathematics\/}, 30, 243--261.

\bibitem[\protect\citeauthoryear{Ogata}{Ogata}{1981}]{Ogata1981}
--- (1981), \enquote{On {L}ewis' simulation method for point processes,} {\em
  IEEE Transactions on Information Theory\/}, 27, 23--30.

\bibitem[\protect\citeauthoryear{Ogata}{Ogata}{1988}]{Ogata1988}
--- (1988), \enquote{Statistical Models for Earthquake Occurrences and Residual
  Analysis for Point Processes,} {\em Journal of the American Statistical
  Association\/}, 83, 9--27.

\bibitem[\protect\citeauthoryear{Ogata}{Ogata}{1998}]{Ogata1998}
--- (1998), \enquote{Space-Time Point-Process Models for Earthquake
  Occurrences,} {\em Annals of the Institute of Statistical Mathematics\/}, 50,
  379--402.

\bibitem[\protect\citeauthoryear{Olinde and Short}{Olinde and
  Short}{2020}]{Olinde2020}
Olinde, J. and Short, M. (2020), \enquote{A Self-limiting {H}awkes Process:
  Interpretation, Estimation, and Use in Crime Modeling,} in {\em 2020 IEEE
  International Conference on Big Data (Big Data)\/}.

\bibitem[\protect\citeauthoryear{Ozaki}{Ozaki}{1979}]{Ozaki1979}
Ozaki, T. (1979), \enquote{Maximum likelihood estimation of {H}awkes'
  self-exciting point processes,} {\em Annals of the Institute of Statistical
  Mathematics\/}, 31, 145–155.

\bibitem[\protect\citeauthoryear{Pinkney, Euan, Gibberd, and Shojaie}{Pinkney
  et~al.}{2024}]{Pinkney2024}
Pinkney, C., Euan, C., Gibberd, A., and Shojaie, A. (2024),
  \enquote{Regularised Spectral Estimation for High-Dimensional Point
  Processes,} Preprint at \url{https://arxiv.org/abs/2403.12908}.

\bibitem[\protect\citeauthoryear{Pinsky}{Pinsky}{2008}]{Pinsky2008}
Pinsky, M. (2008), {\em Introduction to Fourier Analysis and Wavelets\/},
  Graduate studies in mathematics, American Mathematical Society.

\bibitem[\protect\citeauthoryear{Poinas, Delyon, and Lavancier}{Poinas
  et~al.}{2019}]{Poinas2019}
Poinas, A., Delyon, B., and Lavancier, F. (2019), \enquote{{Mixing properties
  and central limit theorem for associated point processes},} {\em
  Bernoulli\/}, 25, 1724--1754.

\bibitem[\protect\citeauthoryear{Rajala, Olhede, Grainger, and Murrell}{Rajala
  et~al.}{2023}]{Rajala2023}
Rajala, T., Olhede, S., Grainger, J., and Murrell, D. (2023), \enquote{What is
  the Fourier Transform of a Spatial Point Process?} {\em IEEE Transactions on
  Information Theory\/}.

\bibitem[\protect\citeauthoryear{Reynaud-Bouret, Rivoirard, Grammont, and
  Tuleau-Malot}{Reynaud-Bouret et~al.}{2014}]{Reynaud2014}
Reynaud-Bouret, P., Rivoirard, V., Grammont, F., and Tuleau-Malot, C. (2014),
  \enquote{Goodness-of-Fit Tests and Nonparametric Adaptive Estimation for
  Spike Train Analysis,} {\em The Journal of Mathematical Neuroscience\/}, 4,
  3.

\bibitem[\protect\citeauthoryear{Trouleau, Etesami, Grossglauser, Kiyavash, and
  Thiran}{Trouleau et~al.}{2019}]{Trouleau2019}
Trouleau, W., Etesami, J., Grossglauser, M., Kiyavash, N., and Thiran, P.
  (2019), \enquote{Learning {H}awkes Processes Under Synchronization Noise,} in
  Chaudhuri, K. and Salakhutdinov, R. (editors), {\em Proceedings of the 36th
  International Conference on Machine Learning\/}, volume~97 of {\em
  Proceedings of Machine Learning Research\/}, PMLR.

\bibitem[\protect\citeauthoryear{Tuan}{Tuan}{1981}]{Tuan1981}
Tuan, P. (1981), \enquote{{Estimation of the Spectral Parameters of a
  Stationary Point Process},} {\em The Annals of Statistics\/}, 9, 615 -- 627.

\bibitem[\protect\citeauthoryear{Villani, Quiroz, Kohn, and Salomone}{Villani
  et~al.}{2022}]{Villani2022}
Villani, M., Quiroz, M., Kohn, R., and Salomone, R. (2022), \enquote{Spectral
  Subsampling MCMC for Stationary Multivariate Time Series with Applications to
  Vector ARTFIMA Processes,} {\em Econometrics and Statistics\/}.

\bibitem[\protect\citeauthoryear{Whittle}{Whittle}{1952}]{Whittle1952}
Whittle, P. (1952), \enquote{Some results in time series analysis,} {\em
  Scandinavian Actuarial Journal\/}, 1952, 48--60.

\bibitem[\protect\citeauthoryear{Yang and Guan}{Yang and Guan}{2024}]{Yang2024}
Yang, J. and Guan, Y. (2024), \enquote{{Fourier analysis of spatial point
  processes},} Preprint at \url{https://arxiv.org/abs/2401.06403}.

\bibitem[\protect\citeauthoryear{Zhang, Fournier, Fallahnezhad, Paradis,
  Rochefort, and Rondi-Reig}{Zhang et~al.}{2023}]{zhang2023cerebellum}
Zhang, L., Fournier, J., Fallahnezhad, M., Paradis, A.-L., Rochefort, C., and
  Rondi-Reig, L. (2023), \enquote{The cerebellum promotes sequential foraging
  strategies and contributes to the directional modulation of hippocampal place
  cells,} {\em Iscience\/}, 26.

\end{thebibliography}

\appendix

\section{Proof of Proposition \ref{PROP:SUM_OF_SPECTRAL_DENSITIES}}\label{appendix:proof_sum_spectra}
          
          \begin{lemma}\label{lemma:sum_second_order_moment}
            Let $X$ and $Y$ be two independent and stationary multivariate point processes with same dimension $d$.
            If they admit second order moment measures, denoted respectively $(M_{ij}^{X})_{1 \le i, j \le d}$ and $(M_{ij}^{Y})_{1 \le i, j \le d}$,
            then the process $N = X + Y$ also admits a second order moment measure,
            noted $(M_{ij}^N)_{1 \le i, j \le d}$.
            Moreover, for any pair $(A, B) \in (\mathcal{B}_{\RR}^c)^2$ and all $1 \le i, j \le d$:
	        
	        \begin{equation}\label{eq:relation_second_moment}
		        M_{ij}^N(A, B) = M_{ij}^{X}(A, B) + M_{ij}^{Y}(A, B) + (m_i^X m_j^Y + m_j^X m_i^Y)\ell_\RR(A)\ell_\RR(B)\,,
	        \end{equation}
	        where for all $i \in \{1, \dots, d\}$ and $m_i^X = \EE[X_i([0, 1))]$ and $m_i^Y = \EE[Y_i([0, 1))]$.
            Furthermore, the reduced measure of $N$ reads:

                \begin{equation}\label{eq:relation_reduced_moment}
		        \breve M_{ij}^N(B) = \breve M_{ij}^X(B) + \breve M_{ij}^Y(B)+ (m_i^X m_j^Y + m_j^X m_i^Y)\ell_{\RR}(B)\,.
	        \end{equation}
	       
          \end{lemma}

          \begin{proof}

          Let $(A,B)\in (\mathcal{B}_{\RR}^c)^2$. Then for any $1 \leq i,j \leq d$:
		\begin{align*}
			M_{ij}^N(A, B) &= \EE[N_i(A)N_j(B)] \\
			&= \EE[(X_i(A) + Y_i(A))(X_j(B) + Y_j(B))]\\
			&= \EE[X_i(A)X_j(B)] + \EE[X_i(A)Y_j(B)] + \EE[X_j(A)Y_i(B)] + \EE[Y_j(A)Y_j(B)]\\
			&= M_{ij}^X(A,B) + \EE[X_i(A)]\EE[Y_j(B)] + \EE[X_j(A)]\EE[Y_i(B)] + M_{ij}^Y(A,B)\\
            &= M_{ij}^X(A,B) + m_i^X m_j^Y\ell(A)\ell(B) + m_j^X m_i^Y\ell(A)\ell(B) + M_{ij}^Y(A,B)\,,
		\end{align*}
		where the last line comes from the stationarity, which implies that $\EE[X_i(A)] = m_i^X \ell(A)$.
         
            By applying Equation~\eqref{eq:reduced_moment_measure_property} to the function $g(x,y) = \II_{x \in [0,1]}\II_{y-x\in B}$ for any $B\in\mathcal{B}_\RR^c$, we can remark that:
		\begin{equation*}
			\int_{\RR^2}{g(x, y) \, M_{ij}^N(\mathrm{d}x, \mathrm{d}y)} = \int_{[0,1]}{\ell_{\RR}(\mathrm{d}x)}\int_{B}{\breve M_{ij}^N(\mathrm{d}u)} = \breve M_{ij}^N(B)\,.
		\end{equation*}
		In particular this equality is also true if we replace $N$ by $X$ and $Y$. 
            By leveraging Equation~\eqref{eq:relation_second_moment} on the left-side integral, we obtain:
		\begin{align*}
			\int_{\RR^2}{g(x, y) \, M_{ij}^N(\mathrm{d}x, \mathrm{d}y)} &= \breve M_{ij}^{X}(B) +  \breve M_{ij}^{Y}(B) + (m_i^X m_j^Y + m_j^X m_i^Y)\int_{x\in[0,1]}{\int_{y\in B + x}{\ell_\RR(\mathrm{d}y)} \, \ell_\RR(\mathrm{d}x)}  \\
			&= \breve M_{ij}^X(B) + \breve M_{ij}^Y(B)+ (m_i^X m_j^Y + m_j^X m_i^Y)\ell_{\RR}([0,1])\ell_{\RR}(B) \\
			&= \breve M_{ij}^X(B) + \breve M_{ij}^Y(B)+ (m_i^X m_j^Y + m_j^X m_i^Y)\ell_{\RR}(B) \,,
		\end{align*}
		which achieves the proof.

        \end{proof}
				
  \subsubsection*{Proof of Proposition~\ref{PROP:SUM_OF_SPECTRAL_DENSITIES}}
		 By definition of the spectral density, for all $\nu\in\RR$:
		 
		\begin{align}
            f_{ij}^N(\nu) = &\int_\RR \mathrm{e}^{-2\pi \iu x \nu} \, \breve M_{ij}^N(\mathrm{d}x) - m_i^N m_j^N \delta(\nu)\nonumber \\
            = &\int_\RR \mathrm{e}^{-2\pi \iu x \nu} \, \breve M_{ij}^X(\mathrm d x) + \int_\RR \mathrm{e}^{-2\pi \iu x \nu} \, \breve M_{ij}^Y(\mathrm d x)\nonumber  \\
            &+ (m_i^X m_j^Y + m_j^X m_i^Y) \int_\RR \mathrm{e}^{-2\pi \iu x \nu} \, \ell_{\RR}(\mathrm{d}x) - (m_i^X + m_i^Y)(m_j^X + m_j^Y) \delta(\nu)\nonumber  \\
            = &\int_\RR \mathrm{e}^{-2\pi \iu x \nu} \, \breve M_{ij}^X(\mathrm d x) - m_i^X m_j^X \delta(\nu) + \int_\RR \mathrm{e}^{-2\pi \iu x \nu} \, \breve M_{ij}^Y(\mathrm d x) - m_i^Y m_j^Y \delta(\nu) \nonumber \\
            &+ (m_i^X m_j^Y + m_j^X m_i^Y) \int_\RR \mathrm{e}^{-2\pi \iu x \nu} \, \ell_{\RR}(\mathrm{d}x) - (m_i^X m_j^Y + m_i^Y m_j^X) \delta(\nu) \nonumber \\
            = &f_{ij}^X(\nu) + f_{ij}^Y(\nu) + (m_i^X m_j^Y + m_j^X m_i^Y) \int_\RR \mathrm{e}^{-2\pi \iu x \nu} \,\mathrm{d}x - (m_i^X m_j^Y + m_i^Y m_j^X) \delta(\nu)\,.\label{eq:dirac_difference} 
          \end{align}
                    
          By properties of the Dirac measure, for all $\nu\in\RR$:
          \[\int_\RR \mathrm{e}^{-2\pi \iu x \nu} \delta(x) \,\mathrm{d}x = 1\,,\]
          and so by duality of the Fourier transform \cite[Proposition 5.2.4.]{Pinsky2008} it follows that:
          \[\int_\RR \mathrm{e}^{-2\pi \iu x \nu} \,\mathrm{d}x  = \delta(\nu)\,.\]

          The last two terms of Equation~\eqref{eq:dirac_difference} being equal, we obtain $f_{ij}^N = f_{ij}^X + f_{ij}^Y$.

\section{Proof of Proposition~\ref{PROPOSITION:FIXED_UNIVARIATE_IDENTIFIABILITY}}\label{appendix:identifiability_univariate}  
  We first show that equality of two spectral densities with different parameters is equivalent to a system of equations (Equations~\eqref{eq:system_1d}).
  This result will be used then to prove both parts of Proposition~\ref{PROPOSITION:FIXED_UNIVARIATE_IDENTIFIABILITY}.

  Let $(\mu, \alpha, \beta, \lambda_0)$ and $(\mu', \alpha', \beta', \lambda_0')$ be two admissible $4$-tuples for the exponential noisy Hawkes model. 
	        Let us assume that for all $\nu\in\RR$:
	        \[
	        		 f_{(\mu, \alpha, \beta, \lambda_0)}^N(\nu) = f_{(\mu', \alpha', \beta', \lambda_0')}^N(\nu).
	        \]
	        Thanks to Equation~\eqref{eq:exponential_spectral_density}, this equality implies the following system of equations:

                \begin{numcases}{}
                        \frac{\mu}{1-\alpha} + \lambda_0 = \frac{\mu'}{1-\alpha'} + \lambda_0' &\text{$(\nu\to +\infty)$}\nonumber\\
                        \frac{\mu}{1-\alpha} \frac{\beta^2 \alpha(2-\alpha)}{\beta^2(1-\alpha)^2} +\frac{\mu}{1-\alpha} + \lambda_0= \frac{\mu'}{1-\alpha'} \frac{{\beta'}^2 \alpha' (2-\alpha')}{{\beta'}^2(1-\alpha')^2} +\frac{\mu'}{1-\alpha'} + \lambda_0' &\text{$(\nu = 0)$}\nonumber\\
                        \frac{\mu}{1-\alpha} \frac{\beta^2 \alpha (2-\alpha)}{\beta^2(1-\alpha)^2 + 4\pi^2} +\frac{\mu}{1-\alpha} + \lambda_0 = \frac{\mu'}{1-\alpha'} \frac{{\beta'}^2\alpha'  (2-\alpha')}{{\beta'}^2(1-\alpha')^2 + 4\pi^2}+\frac{\mu'}{1-\alpha'} + \lambda_0' &\text{$(\nu = 1)$}.\nonumber
                \end{numcases}
                
                The first equality can be used to simplify the second and third equalities,
                leading to:
                
		      \begin{numcases}{}
                    \frac{\mu}{1-\alpha} + \lambda_0 = \frac{\mu'}{1-\alpha'} + \lambda_0' \nonumber\\
			 \frac{\mu}{1-\alpha} \frac{\alpha (2-\alpha)}{(1-\alpha)^2} = \frac{\mu'}{1-\alpha'} \frac{\alpha' (2-\alpha')}{(1-\alpha')^2}\nonumber\\
		         \frac{\mu}{1-\alpha} \frac{\beta^2 \alpha (2-\alpha)}{\beta^2(1-\alpha)^2 + 4\pi^2}= \frac{\mu'}{1-\alpha'} \frac{{\beta'}^2 \alpha' (2-\alpha')}{{\beta'}^2(1-\alpha')^2 + 4\pi^2}\,. \nonumber
		      \end{numcases}
		      
		      Given that $\alpha\in(0,1)$ (same for $\alpha'$), the last two equalities lead to:
		      
		      \[ \frac{\beta^2(1-\alpha)^2}{\beta^2 (1-\alpha)^2 + 4\pi^2} = \frac{{\beta'}^2(1-\alpha')^2}{{\beta'}^2 (1-\alpha')^2 + 4\pi^2}\,,\]
		      which in turn implies $\beta(1-\alpha) = \beta' (1-\alpha')$.
		      
		      Thus, it results the following system of equations:
		      
		       \begin{subnumcases}{\label{eq:system_1d}}
                    \frac{\mu}{1-\alpha} + \lambda_0 = \frac{\mu'}{1-\alpha'} + \lambda_0' \label{eq:system_1} \\
			 \frac{\mu \alpha (2-\alpha)}{(1-\alpha)^3} = \frac{\mu' \alpha' (2-\alpha')}{(1-\alpha')^3}\label{eq:system_2} \\
		         \beta(1-\alpha) = \beta' (1-\alpha') \label{eq:system_3}\,.
		      \end{subnumcases}
		      
		      Now, Equations~\eqref{eq:system_1d} lead to
		      \begin{numcases}{}
	          \frac{\mu}{1-\alpha} + \lambda_0 = \frac{\mu'}{1-\alpha'} + \lambda_0' \nonumber\\
		        \frac{\mu}{1-\alpha} \beta^2 \alpha (2-\alpha) = \frac{\mu'}{1-\alpha'} {\beta'}^2\alpha' (2-\alpha') \nonumber\\
	          \beta(1-\alpha) = \beta' (1-\alpha') \,, \nonumber
		      \end{numcases}
		      which, by Equation~\eqref{eq:exponential_spectral_density}, implies straightforwardly $f_{(\mu, \alpha, \beta, \lambda_0)}^N = f_{(\mu', \alpha', \beta', \lambda_0')}^N$.
		      Consequently,
		      \[
		        f_{(\mu, \alpha, \beta, \lambda_0)}^N = f_{(\mu', \alpha', \beta', \lambda_0')}^N
		        \iff
		        \text{Equations~}\eqref{eq:system_1d} \,.
		      \]
		      
  \subsubsection*{The model $\mathcal{Q}$ is not identifiable.}
              	Let $\tau> -\lambda_0$ and $\lambda_0' = \lambda_0 + \tau > 0$.
          
          		Then, by denoting $\kappa = \frac{\mu}{1-\alpha} \beta^2\alpha (2-\alpha)$, Equations~\eqref{eq:system_1d} are equivalent to:
          		\begin{equation}\label{eq:system_non_identifiable_1d}
	          	\begin{cases}
            		\mu' = \frac{\beta(1-\alpha) (\frac{\mu}{1-\alpha}-\tau)}{\sqrt{\beta^2(1-\alpha)^2 + \frac{\kappa}{\frac{\mu}{1-\alpha}-\tau}}} \\
            		\alpha' = 1 - \frac{\beta(1 - \alpha)}{\sqrt{\beta^2(1-\alpha)^2 + \frac{\kappa}{\frac{\mu}{1-\alpha}-\tau}}} \\
            		\beta' = \sqrt{\beta^2(1-\alpha)^2 + \frac{\kappa}{\frac{\mu}{1-\alpha}-\tau}}\,.
	          	\end{cases}
          		\end{equation}
			From \eqref{eq:system_non_identifiable_1d}, it is clear that for all $\tau\in \left(-\lambda_0, \frac{\mu}{1-\alpha} \right)\setminus\{0\}$,
          		$(\mu', \alpha', \beta', \lambda_0 + \tau) \neq (\mu, \alpha, \beta, \lambda_0) $ is an admissible parameter for $\mathcal Q$ and
          		$f_{(\mu', \alpha', \beta', \lambda_0 + \tau)}^N = f_{(\mu, \alpha, \beta, \lambda_0)}^N$.
          		Consequently, $\mathcal Q$ is not identifiable.

  \subsubsection*{The reduced model defined by a triplet of admissible parameters is identifiable.}
                		
		It will be shown that, for admissible parameters, $f_{(\mu, \alpha, \beta, \lambda_0)}^N = f_{(\mu', \alpha', \beta', \lambda_0')}^N$ implies \[\mu = \mu' \iff \alpha = \alpha' \iff \beta = \beta' \iff \lambda_0 = \lambda_0'\,,\]
		from which we can deduce identifiability of the four collections of models mentioned in Proposition~\ref{PROPOSITION:FIXED_UNIVARIATE_IDENTIFIABILITY}.
		
		So, let us assume that $f_{(\mu, \alpha, \beta, \lambda_0)} = f_{(\mu', \alpha', \beta', \lambda_0')}$,
		for some admissible parameters.
		As shown beforehand this implies Equations~\eqref{eq:system_1d}.
		From Equation~\eqref{eq:system_3}, we establish that \[\alpha = \alpha' \iff \beta = \beta'\,,\] since, $\beta > 0$ and $\alpha < 1$.
		
		From Equation~\eqref{eq:system_2}, since $\alpha > 0$, it is clear that $\alpha = \alpha' \implies \mu = \mu'$.  Conversely, if $\mu=\mu' > 0$, Equation~\eqref{eq:system_2} becomes:
		\[g(\alpha) =  \frac{\alpha (2-\alpha)}{(1-\alpha)^3} = \frac{\alpha' (2-\alpha')}{(1-\alpha')^3}  = g(\alpha')\,.\]
		In particular, $g$ is a strictly increasing function on $(0,1)$ and so it can be deduced that $\alpha = \alpha'$. So \[\mu = \mu' \iff \alpha = \alpha'\,.\]
		
		Finally, as $\mu = \mu' \iff \alpha = \alpha'$, from Equation~\eqref{eq:system_1} we conclude that $\alpha = \alpha' \implies \lambda_0 = \lambda_0'$. 
		Let us now assume that $\lambda_0 = \lambda_0'$, Equation~\eqref{eq:system_1} then reads \[\frac{\mu}{1-\alpha} = \frac{\mu'}{1-\alpha'}\,.\]
		As $\mu > 0$ and $\mu' > 0$, Equation~\eqref{eq:system_2} can be then reduced to \[s(\alpha) := \frac{\alpha(2-\alpha)}{(1-\alpha)^2} = \frac{\alpha(2-\alpha)}{(1-\alpha)^2} = s(\alpha')\,,\]
		where $s$ is a strictly increasing function on $(0,1)$ and so $\alpha = \alpha'$. With this we have proved that $\alpha = \alpha' \iff \lambda_0 = \lambda_0'$ which achieves the proof.

  \section{Proof of Proposition~\ref{PROP:STATN_IDENTIFIABILITY}}
    \label{app:felix1}

  Let $\mathcal G_t = \mathcal H_t \vee \mathcal P_t$ denote the $\sigma$-algebra generated by $\mathcal H_t$ and $\mathcal P_t$.
  Note that this filtration is different from the natural filtration $\mathcal F_t$ of the superposed process $(N(t))_{t\ge0}$ as it yields information about the process, $H$ or $P$, which has generated each point.
  Let us now consider the conditional survival function of the first non-negative jump $\tau_1$ of $N$ given the past (before $0$):
	for all $t \ge 0$,
	\begin{align*}
		\mathbb P \left(\tau_1 > t \,\big|\, \mathcal G_0 \right)
		&= \mathbb P \left(N(t) = 0 \,\big|\, \mathcal G_0 \right)\\
		&= \mathbb P \left(P(t) = 0, H(t) = 0 \,\big|\, \mathcal G_0 \right)\\
		&= \mathbb P \left(P(t) = 0 \,\big|\, \mathcal P_0 \right) \mathbb P \left( H(t) = 0 \,\big|\, \mathcal H_0 \right) & (\text{by independence})\\
		&= \mathbb P \left(P(t) = 0 \right) \mathbb P \left( H(t) = 0 \,\big|\, \mathcal H_0 \right) & (\text{by definition})\\
		&= \mathrm e^{-\lambda_0 t} e^{- \mu \left( t - u_t \right) - u_t \lambda^H(0)} \,,
	\end{align*}
        with $u_t = {(1 - \mathrm e^{-\beta t})}{\beta^{-1}}$ and
	where we have used that,
	given that $H(t) = 0$,
	for all $u \in [0, t]$,
	\[
	  \lambda^H (u) = \mu + \left( \lambda^H(0) - \mu \right) \mathrm e^{-\beta u}.
	\]
	Let us remark that the last equality can also be deduced from \cite[Corollary~3.3]{Dassios2011} with the same notation $u_t$:
	\[
	  \mathbb P \left( H(t) = 0 \,\big|\, \mathcal H_0 \right) = \exp \left( - \int_0^{u_t} \frac{\mu\beta v}{1 - \beta v} \, \dd v \right) \mathrm e^{-u_t \lambda^H(0)} \,.
	\]

	It appears that $\mathbb P \left(\tau_1 > t \,\big|\, \mathcal G_0 \right)$ depends on the past only through $\lambda^H(0)$.
	But $\lambda^H(0) = \lambda^N(0) - \lambda_0$ and $\lambda^N(0)$ is distributed according to the stationary distribution of $\left( \lambda^N(t) \right)_{t \ge 0}$.
	It results that $\lambda^H(0)$ is distributed according to the stationary distribution of $\left( \lambda^H(t) \right)_{t \ge 0}$, the Laplace transform of which is given by \cite[Corollary 3.1]{Dassios2011}.
	Plugging in this result, we have:
	\begin{align*}
		\mathbb P \left(\tau_1 > t \right)
		&= \mathrm e^{-\lambda_0 t - \mu \left( t - u_t \right)} \mathbb E \left[ \mathrm e^{- u_t \lambda^H(0)} \right]\\
		&= \mathrm e^{-\lambda_0 t - \mu \left( t - u_t \right)} \exp \left( - \int_0^{u_t} \frac{\mu\beta v}{\beta v + \mathrm e^{-\alpha\beta v} - 1} \dd v \right) \,.
	\end{align*}	
	For ease of derivation, let $G : t \in \RR_{\ge 0} \mapsto - \log \mathbb P \left(\tau_1 > t \right)$.
	The function $G$ is differentiable and has derivative $D$ given by:
	\[
	  D(t) = \lambda_0 + \mu \left( 1 - \mathrm e^{-\beta t} \right) + \frac{\mu \left( 1 - \mathrm e^{-\beta t} \right) \mathrm e^{-\beta t}}{\exp \left(-\alpha \left( 1 - \mathrm e^{-\beta t} \right) \right) - \mathrm e^{-\beta t}} \,.
	\]
	
	Now define in the same manner $\tau_1'$, $G'$, and $D'$ from the process $N'$ with parameter $(\mu', \alpha', \beta', \lambda_0')$, with $\lambda^{N'}(0)$ being distributed according to the stationary distribution of $\bigl( \lambda^{N'}(t) \bigr)_{t \ge 0}$,
	and assume that $N$ and $N'$ have the same distribution.
	Then it follows that both $\tau_1$ and $\tau_1'$ have the same distribution, so that $G(t) = G'(t)$ for all $t \ge 0$.
	Since $G$ and $G'$ are everywhere differentiable, it results that $D(t) = D'(t)$ for all $t \ge 0$.
	
	We want to establish a system of four equations satisfied by the parameters that leads to the equality of the 4-tuples.
	First, noting that $\lim_{t \to \infty} D(t) = \lambda_0 + \mu$, we get 
	\begin{equation}\label{eq:statN_eq_1}
	\lambda_0 + \mu = \lambda_0' + \mu' \,.
	\end{equation}
	Then, since $\lim_{t \to 0} D(t) = \lambda_0 + \mu / (1 - \alpha)$, we get, using Equation~\eqref{eq:statN_eq_1},
	\begin{equation}\label{eq:statN_eq_2}
	\frac{\mu\alpha}{1 - \alpha} = \frac{\mu'\alpha'}{1 - \alpha'} \,.
	\end{equation}
	
	To highlight two other equations on the parameters, we establish the Taylor expansion of $D(t)$ around $t = 0$ up to order 2.
	After some calculation, we find that 
	\begin{equation*}
		D(t) = \lambda_0 + \frac{\mu}{1-\alpha} + \frac{\mu\alpha}{1-\alpha} \left( \frac{1}{2} \beta t \frac{\alpha-2}{1-\alpha} + \frac{1}{12} \beta^2 t^2 \frac{\alpha^3 - \alpha^2 - 3\alpha + 6}{(1-\alpha)^2} + o\bigl(t^2\bigr) \right) \,.
	\end{equation*}
	From the first-order term of the expansion, Equation~\eqref{eq:statN_eq_2} and the equality of $D$ and $D'$, we find that
	\begin{equation}\label{eq:statN_eq_3}
		\beta \frac{\alpha - 2}{1 - \alpha} = \beta' \frac{\alpha' - 2}{1 - \alpha'} \,.
	\end{equation}
	Then the second-order term of the expansion can be rewritten
	\begin{equation*}
		\frac{\mu\alpha}{1-\alpha} \frac{1}{12} \left( \frac{\beta(\alpha-2)}{1-\alpha} \right)^2 \frac{\alpha^3 - \alpha^2 - 3\alpha + 6}{(\alpha-2)^2} \,,
	\end{equation*}
	so that from Equations~\eqref{eq:statN_eq_2} and \eqref{eq:statN_eq_3}, we find that 
	\begin{equation}\label{eq:statN_eq_4}
		g(\alpha) = \frac{\alpha^3 - \alpha^2 - 3\alpha + 6}{(\alpha-2)^2} = \frac{\alpha'^3 - \alpha'^2 - 3\alpha' + 6}{(\alpha'-2)^2} = g(\alpha') \,.
	\end{equation}
	$\alpha \mapsto g(\alpha)$ can be shown to be strictly increasing for $\alpha \in (0, 1)$, so that Equation~\eqref{eq:statN_eq_4} yields that $\alpha = \alpha'$.
	With this remark, one can easily show from the system composed by Equations~\eqref{eq:statN_eq_1}--\eqref{eq:statN_eq_4} that the tuples $(\mu, \alpha, \beta, \lambda_0)$ and $(\mu', \alpha', \beta', \lambda_0')$ are equal.

  \section{Proof of Proposition~\ref{PROP:RECTANGLE}}
    \label{app:felix3}

\begin{lemma}\label{lemma:taylor}
	Assume all moments of $h$ are finite: $\forall n \ge 0, m_n = \int t^n h(t) \, \dd t < \infty$.
	Then
	the spectral density $f^N$ of the noisy Hawkes process $N$ has a Taylor expansion around $0$ given by:
	\begin{equation}\label{eq:taylor}
	  \forall t \in \RR: \quad
		f^N(t) = \frac{\mu}{(1-\alpha)^3} + \frac{\mu}{(1-\alpha)^3} \sum_{q \ge 1} \left( \frac{\alpha}{1 - \alpha} \sum_{n \ge 1} b_n t^{2n} \right)^q + \lambda_0 \,,
	\end{equation}
	with 
	\begin{equation}\label{eq:taylor_an_bn}
		b_n = \frac{(-1)^n (2\pi)^{2n} a_n}{(2n)!} \,, \qquad \text{and} \qquad
		a_n = 2 m_{2n} - \frac{\alpha}{1 - \alpha} \sum_{k=1}^{2n-1} \binom{2n}{k} (-1)^k m_k m_{2n-k} \,.
	\end{equation}
\end{lemma}
\begin{proof}
	Introduce the Taylor expansion of the Fourier transform of $h$:
	\begin{equation*}
	  \forall t \in \RR: \quad
		\tilde h(t) = \sum_{n \ge 0} (\iu \tau)^n \frac{m_n}{n!},
	\end{equation*}
  where $\tau = -2 \pi t$.
	Then, given that $m_0=1$,
	\begin{align*}
		\left \lvert 1 - \alpha \tilde h(t) \right \rvert^2 
		&= \left( 1 - \alpha \sum_{n \ge 0} (\iu \tau)^n \frac{m_n}{n!} \right) \left( 1 - \alpha \sum_{n \ge 0} (-\iu \tau)^n \frac{m_n}{n!} \right)\\
		&= 1 - 2 \alpha \sum_{n \ge 0} (-1)^n \tau^{2n} \frac{m_{2n}}{(2n)!} + \alpha^2 \sum_{n \ge 0} \sum_{k=0}^n (-1)^k (\iu \tau)^n \frac{m_{n-k}m_k}{(n-k)!k!}\\
		&= 1 - \alpha \sum_{n \ge 0} 2 m_{2n} \frac{(-1)^n \tau^{2n}}{(2n)!} + \alpha^2 \sum_{n \ge 0} \sum_{k=0}^{2n} (-1)^k m_k m_{2n-k} \binom{2n}{k} \frac{(-1)^n \tau^{2n}}{(2n)!}\\
		&= (1 - \alpha)^2 \left( 1 - \frac{\alpha}{1 - \alpha} \sum_{n \ge 1} a_n \frac{(-1)^n \tau^{2n}}{(2n)!} \right).
	\end{align*}
	
	Inverting this expression and taking the Taylor expansion of $x \mapsto (1 - x)^{-1}$ around 0 yields the desired result.
\end{proof}

\begin{remark}
	For $n = 1, 2$, $a_n$ is given by
	\[
	  \begin{cases}
		  a_1 &= 2 m_2 + \frac{\alpha}{1 - \alpha} 2 m_1^2 \\
		  a_2 &= 2 m_4 + \frac{\alpha}{1 - \alpha} \bigl(8 m_1 m_3 - 6 m_2^2\bigr) \,,
	  \end{cases}
	\]
	so that the Taylor expansion of $f^N$ up to order $5$ is:
	\begin{equation*}
		f^N(t) = a + c_1 t^2 + c_2 t^4 + o(t^5) \,,
	\end{equation*}
	with
	\begin{align*}
	  a &= \frac{\mu}{(1-\alpha)^3} + \lambda_0 \\
		c_1 &= 4 \frac{\mu\alpha \pi^2}{(1-\alpha)^4} \left[ - m_2 - \frac{\alpha}{1 - \alpha} m_1^2 \right]\\
		c_2 &= 16 \frac{\mu\alpha \pi^4}{(1-\alpha)^4} \left[ \frac{m_4}{12} + \frac{\alpha}{1 - \alpha} \left( \frac{m_1 m_3}{3} + \frac{3 m_2^2}{4} \right) + \left( \frac{\alpha}{1 - \alpha} \right)^2 2 m_2 m_1^2 + \left( \frac{\alpha}{1 - \alpha} \right)^3 m_1^4 \right] \,.
	\end{align*}
\end{remark}

\subsubsection*{Proof of Proposition~\ref{PROP:RECTANGLE}}
Let $(\mu, \alpha, \phi, \lambda_0)$ and $(\mu', \alpha', \phi', \lambda_0')$ be two admissible $4$-tuples for the uniform noisy Hawkes model $\mathcal R$, and assume that, for all $\nu \in \mathbb R$,
	\begin{equation*}
		f_{(\mu, \alpha, \phi, \lambda_0)}^N(\nu) = f_{(\mu', \alpha', \phi', \lambda_0')}^N(\nu)\,.
	\end{equation*}
	First, noting that $\lim_{\nu\to\infty} f_{(\mu, \alpha, \phi, \lambda_0)}^N(\nu) = \lambda_0 + \mu/(1 - \alpha)$, we get
	\begin{equation}\label{eq:idunif_eq1}
		\lambda_0 + \frac{\mu}{1 - \alpha} = \lambda_0' + \frac{\mu'}{1 - \alpha'}\,.
	\end{equation}
	Then, since $f_{(\mu, \alpha, \phi, \lambda_0)}^N(0) = \lambda_0 + \mu/(1-\alpha)^3$ and using Equation~\eqref{eq:idunif_eq1}, we get:
	\begin{equation}\label{eq:idunif_eq2}
		\frac{\mu\alpha(2-\alpha)}{(1-\alpha)^3} = \frac{\mu'\alpha'(2-\alpha')}{(1-\alpha')^3}\,.
	\end{equation}
	
	Now, since the Taylor expansions, given by Equation~\eqref{eq:taylor}, of $f_{(\mu, \alpha, \phi, \lambda_0)}^N$ and $f_{(\mu', \alpha', \phi', \lambda_0')}^N$ around 0 coincide, their respective Taylor coefficients $(c_1, c_2, \ldots)$ and $(c_1', c_2', \ldots)$ are equal.
	Plugging in the moments of the uniform distribution on $[0, \phi]$, 
	\begin{equation*}
		m_n = \frac{\phi^n}{n+1}\,, \qquad \text{for all}~n \ge 0\,,
	\end{equation*}
	the first order coefficient $c_1$ can be written
	\begin{equation*}
		c_1 = - \frac{\pi}{3} \frac{\mu\alpha(2-\alpha)}{(1-\alpha)^3} \phi^2 \frac{4-\alpha}{(1-\alpha)^2 (2-\alpha)}\,,
	\end{equation*}
	so that, with the use of Equation~\eqref{eq:idunif_eq2}, we get
	\begin{equation}\label{eq:idunif_eq3}
		\phi^2 \frac{4-\alpha}{(1-\alpha)^2 (2-\alpha)} = \phi'^2 \frac{4-\alpha'}{(1-\alpha')^2 (2-\alpha')}\,.
	\end{equation}
	Similarly, by plugging the moments in the second order coefficient $c_2$, we get 
	\begin{equation*}
		c_2 = \frac{\pi^4}{15} \frac{\mu\alpha(2-\alpha)}{(1-\alpha)^3} \left( \phi^2 \frac{4-\alpha}{(1-\alpha)^2(2-\alpha)} \right)^2 \frac{(2-\alpha)(\alpha^3 - 8\alpha^2 + 18\alpha + 4)}{(4 - \alpha)^2}\,,
	\end{equation*}
	so that, by Equations \eqref{eq:idunif_eq2} and \eqref{eq:idunif_eq3},
	\begin{equation}\label{eq:idunif_eq4}
		g(\alpha) = \frac{(2-\alpha)(\alpha^3 - 8\alpha^2 + 18\alpha + 4)}{(4 - \alpha)^2} = \frac{(2-\alpha')(\alpha'^3 - 8\alpha'^2 + 18\alpha' + 4)}{(4 - \alpha')^2} = g(\alpha')\,.
	\end{equation}
	But $\alpha \mapsto g(\alpha)$ can be shown to be strictly increasing for $\alpha \in (0, 1)$, so that Equation~\eqref{eq:idunif_eq4} yields that $\alpha = \alpha'$.
	Finally, it is easily proven from the system composed by Equations~\eqref{eq:idunif_eq1}--\eqref{eq:idunif_eq4} that the tuples $(\mu, \alpha, \beta, \lambda_0)$ and $(\mu', \alpha', \beta', \lambda_0')$ are equal.

  \section{Proof of Proposition~\ref{PROPOSITION:BI_NON_IDENTIFIABLE}}\label{appendix:bi_non_identifiable}

Let $N$ be a bivariate noisy Hawkes process parametrised by the exponential model $\mathcal{Q}_\Lambda$. We will prove that if either of the conditions in Proposition~\ref{PROPOSITION:BI_NON_IDENTIFIABLE} is fulfilled then $\mathcal{Q}_\Lambda$ is not identifiable.

\begin{itemize}

\item \textbf{Condition~\ref{hyp:bi_non_identifiable_1}:} Let 
\[\Lambda = \left\{ \begin{pmatrix} \alpha_{11} & 0 \\ 0 & \alpha_{22} \end{pmatrix}, 0 \le \alpha_{11}, \alpha_{22} < 1 \right\}\,.\]
For any admissible $\theta = (\mu, \alpha, \beta, \lambda_0)$, the spectral matrix $\mathbf{f}_{\theta}^N$ is diagonal such that, for any integer $i\in\{1,2\}$:

\[{ f_{\theta}^N }_{ii}(\nu) = \frac{\mu_i}{1-\alpha_{ii}}\left(\frac{\beta_i^2 \alpha_{ii} (2-\alpha_{ii})}{\beta_i^2 (1-\alpha_{ii})^2 + 4 \pi^2 \nu^2} \right) + \left(\frac{\mu_i}{1-\alpha_{ii}} + \lambda_0\right)\,.\] 

We will show that there exists an admissible $\theta' = (\mu', \alpha', \beta', \lambda_0')$ such that $\theta'\neq \theta$ and $\mathbf{f}_\theta^N = \mathbf{f}_{\theta'}^N$. Let $\tau > -\lambda_0$ and $\lambda_0' = \lambda_0 + \tau$.

For every $i\in\{1,2\}$,
on the one hand, if $\alpha_{ii}\neq 0$ then, as shown in the univariate case (Appendix~\ref{appendix:identifiability_univariate}),
letting $\kappa_i = \frac{\mu_i}{1-\alpha_{ii}}\beta_i^2 \alpha_{ii} (2-\alpha_{ii})$ and defining the parameters:
\begin{equation}
\begin{cases}
	\mu_i' = \frac{\beta_i(1-\alpha_{ii}) (\frac{\mu_i}{1-\alpha_{ii}}-\tau)}{\sqrt{\beta_i^2(1-\alpha_{ii})^2 + \frac{\kappa_i}{\frac{\mu_i}{1-\alpha_{ii}}-\tau}}} \\
	\alpha_{ii}' = 1 - \frac{\beta_i(1 - \alpha_{ii})}{\sqrt{\beta_i^2(1-\alpha_{ii})^2 + \frac{\kappa_i}{\frac{\mu_i}{1-\alpha_{ii}}-\tau}}} \\
	\beta_i' = \sqrt{\beta_i^2(1-\alpha_{ii})^2 + \frac{\kappa_i}{\frac{\mu_i}{1-\alpha_{ii}}-\tau}}\,,
\end{cases}
\end{equation}
leads to ${ f_{\theta}^N }_{ii} = { f_{\theta'}^N }_{ii}$.

On the other hand, if $\alpha_{ii}=0$ then,
${ f_{\theta}^N }_{ii} = \mu_i + \lambda_0$ and it is enough to consider:
\[
\begin{cases}
	\mu_i' = \mu_i - \tau \\
	\alpha_{ii}' = 0 \\
	\beta_i' = \beta_i\,,
\end{cases}
\]
to get ${ f_{\theta}^N }_{ii} = { f_{\theta'}^N }_{ii}$.

In both cases, for $\tau \in\left(-\lambda_0, \min_{1\leq i \leq 2}\left\{\frac{\mu_i}{1-\alpha_{ii}}\right\} \right)\setminus\{0\}$, we obtain an admissible parameter $\theta' \neq \theta$ and
$\mathbf{f}_\theta^N = \mathbf{f}_{\theta'}^N$.
Thus, the model is not identifiable.

\item \textbf{Condition~\ref{hyp:bi_non_identifiable_2} and \ref{hyp:bi_non_identifiable_2_bis}:} 
Without loss of generality, let 
\[\Lambda = \left\{ \begin{pmatrix} \alpha_{11} & \alpha_{12} \\ 0 & 0 \end{pmatrix}, 0 < \alpha_{11} < 1, \alpha_{12} > 0 \right\}\,.\] 

For any admissible parameter $\theta = (\mu, \alpha, \beta, \lambda_0)$ and for any $\nu\in\RR$, the spectral matrix reads: 

\begin{equation*}
	\mathbf{f}_\theta^N(\nu) = 
	\begin{bmatrix}
		{f_\theta^N}_{11}(\nu) & {f_\theta^N}_{12}(\nu)\\
		{f_\theta^N}_{21}(\nu) & {f_\theta^N}_{22}(\nu)
	\end{bmatrix}	=
	\begin{bmatrix}
		{f_\theta^N}_{11}(\nu) &  \dfrac{\mu_2 \beta_1 \alpha_{12}}{\beta_1(1-\alpha_{11}) + 2 \pi \iu \nu}\\
		 \dfrac{\mu_2 \beta_1 \alpha_{12}}{\beta_1(1-\alpha_{11}) - 2 \pi \iu \nu} & \mu_2 + \lambda_0
	\end{bmatrix}
\end{equation*}
with 
\begin{align*}
	{f_\theta^N}_{11}(\nu) = &\left(\dfrac{\mu_1}{1-\alpha_{11}} + \dfrac{\mu_2\alpha_{12}}{1-\alpha_{11}} \right)\dfrac{\beta_1^2 \alpha_{11}(2-\alpha_{11})}{\beta_1^2 (1-\alpha_{11})^2 + 4 \pi^2\nu^2} + \dfrac{\mu_2 \beta_1^2 \alpha_{12}^2}{\beta_1^2 (1-\alpha_{11})^2 + 4 \pi^2\nu^2} \\
	&+ \dfrac{\mu_1}{1-\alpha_{11}} + \dfrac{\mu_2\alpha_{12}}{1-\alpha_{11}} + \lambda_0\,.
\end{align*}

Let us introduce the following constants:

\begin{equation*}
	\begin{cases}
		A = \mu_2 + \lambda_0\\
		B = \mu_2 \beta_1 \alpha_{12}\\
		C = \beta_1 (1-\alpha_{11})\\
		D = \frac{\mu_1 + \mu_2 \alpha_{12}}{1-\alpha_{11}} + \lambda_0\\
		E = (\frac{\mu_1 + \mu_2 \alpha_{12}}{1-\alpha_{11}} ) \beta_1^2 \alpha_{11} (2-\alpha_{11}) + \mu_2 \beta_1^2 \alpha_{12}^2\,,
	\end{cases}
\end{equation*}
which allow us to rewrite the spectral matrix as:

\begin{equation}\label{eq:bi_matrix_constants}
	\mathbf{f}_\theta^N(\nu) = 
	\begin{bmatrix}
		\dfrac{E}{C^2 + 4 \pi^2 \nu^2} + D & \dfrac{B}{C + 2 \pi \iu \nu}\\
		\dfrac{B}{C - 2 \pi \iu \nu} & A
	\end{bmatrix}\,.
\end{equation}

Let $\tau\in\RR\setminus \left\{\mu_2, (\mu_1 + \mu_2 \alpha_{12})/(1-\alpha_{11}) \right\}$ and: \[\kappa_\tau = \frac{(\mu_1 + \mu_2 \alpha_{12})\alpha_{11}(2-\alpha_{11})(\mu_2 - \tau) - \tau\mu_2^2 \alpha_{12}^2(1-\alpha_{11})}{(\mu_2 - \tau)(\mu_1 + \mu_2\alpha_{12} - \tau(1-\alpha_{11}))}\,.\]

Now, consider the parameter $\theta' = (\mu', \alpha', \beta', \lambda_0')$, defined as follows:

\begin{equation}\label{eq:bi_system_constants}
	\begin{cases}
		\mu_1' = \frac{\mu_1 - \tau(1-\alpha_{11})}{\sqrt{(1-\alpha_{11})^2 + \kappa_\tau}}\\
		\mu_2' = \mu_2 - \tau\\
		\alpha_{11}' = 1 - \frac{1-\alpha_{11}}{\sqrt{(1-\alpha_{11})^2 + \kappa_\tau}}\\
		\alpha_{12}' = \frac{\mu_2 \alpha_{12}}{(\mu_2 - \tau)\sqrt{(1-\alpha_{11})^2 + \kappa_\tau}}\\
		\beta_1' = \beta_1 \sqrt{(1-\alpha_{11})^2 + \kappa_\tau}\\
		\lambda_0' = \lambda_0 + \tau\,.
	\end{cases}
\end{equation}

The goal will be to show that there exists $\tau$ such that $\theta'$ is well defined, admissible for $\mathcal Q_\Lambda$, such that $\theta\neq\theta'$ and $\mathbf{f}_\theta^N = \mathbf{f}_{\theta'}^N$.

First, in order for $\theta'$ to be an admissible parameter, let us remark that $\rho(S) = \alpha_{11}$ and so we obtain the following constrains:

\begin{equation}\label{eq:system_tau_non_identifiability}
\begin{cases}
		\mu_1' > 0\\
		\mu_2' > 0\\
		0 < \alpha_{11}' < 1\\
		\alpha_{12}' > 0\\
		\beta_1' > 0\\
		\lambda_0' >0
\end{cases}\iff
\begin{cases}
		\mu_1 - \tau(1-\alpha_{11}) > 0\\
		\mu_2 - \tau > 0\\
		\frac{1-\alpha_{11}}{\sqrt{(1-\alpha_{11})^2 + \kappa_\tau}} < 1\\
		(1-\alpha_{11})^2 + \kappa_\tau > 0\\
		\lambda_0 + \tau > 0
\end{cases}\iff
\begin{cases}
		\tau < \frac{\mu_1}{1-\alpha_{11}}\\
		\tau < \mu_2\\
		\kappa_\tau > 0\\
		\tau > -\lambda_0\,.
\end{cases}
\end{equation}

Since $\tau < \mu_1/(1-\alpha_{11}) \implies \mu_1 + \mu_2 \alpha_{12} - \tau(1-\alpha_{11}) > 0$, the third inequality becomes

\[
	 \tau < \frac{(\mu_1 + \mu_2 \alpha_{12})\alpha_{11}(2-\alpha_{11})\mu_2}{(\mu_1 + \mu_2 \alpha_{12})\alpha_{11}(2-\alpha_{11}) + \mu_2^2(1-\alpha_{11}) \alpha_{12}^2}\,.
\]
Then for 
\[\tau \in \left(-\lambda_0, \min\left\{\frac{\mu_1}{1-\alpha_{11}}, \mu_2,  \frac{(\mu_1 + \mu_2 \alpha_{12})\alpha_{11}(2-\alpha_{11})\mu_2}{(\mu_1 + \mu_2 \alpha_{12})\alpha_{11}(2-\alpha_{11}) + \mu_2^2(1-\alpha_{11}) \alpha_{12}^2} \right\} \right)\setminus\{0\}\,,\]
the right-hand side of Equations~\eqref{eq:system_tau_non_identifiability} is verified and so $\theta'$ defined by Equations~\eqref{eq:bi_system_constants} is well defined, admissible for $\mathcal Q_\Lambda$ and $\theta' \neq \theta$.

Then, we can show that:

\begin{equation*}
	\begin{cases}
		\mu_2' + \lambda_0' = \mu_2 + \lambda_0 = A\\
		 \mu_2' \beta_1' \alpha_{12}' = \mu_2 \beta_1 \alpha_{12} = B\\
		\beta_1' (1-\alpha_{11}') = \beta_1 (1-\alpha_{11}) = C\\
		\frac{\mu_1' + \mu_2' \alpha_{12}'}{1-\alpha_{11}'} + \lambda_0' = \frac{\mu_1 + \mu_2 \alpha_{12}}{1-\alpha_{11}} + \lambda_0 = D\\
		\left(\frac{\mu_1' + \mu_2' \alpha_{12}'}{1-\alpha_{11}'} \right) {\beta_1'}^2 \alpha'_{11} (2-\alpha'_{11}) + \mu'_2 {\beta_1'}^2 {\alpha_{12}'}^2 = \left(\frac{\mu_1 + \mu_2 \alpha_{12}}{1-\alpha_{11}} \right) \beta_1^2 \alpha_{11} (2-\alpha_{11}) + \mu_2 \beta_1^2 \alpha_{12}^2 = E\,.
	\end{cases}
\end{equation*}

This assures that, for all $\nu \in \RR$:

\[
	\mathbf{f}_{\theta'}^N(\nu) = \begin{bmatrix}
		\dfrac{E}{C^2 + 4 \pi^2 \nu^2} + D & \dfrac{B}{C + 2 \pi \iu \nu}\\
		\dfrac{B}{C - 2 \pi \iu \nu} & A
	\end{bmatrix} 
	= \mathbf{f}_{\theta}^N(\nu) \,,
\]
Thus, the model is not identifiable.
\end{itemize}

\section{Proof of Proposition~\ref{PROPOSITION:BI_IDENTIFIABLE}, Situations~\ref{HYP:BI_IDENTIFIABLE_1} and \ref{HYP:BI_IDENTIFIABLE_1_BIS}}\label{appendix:bi_identifiable_1}

Without loss of generality, let 
\[\Lambda = \left\{ \begin{pmatrix} \alpha_{11} & 0 \\ \alpha_{21} & 0 \end{pmatrix}, 0 \le \alpha_{11} < 1, \alpha_{21} > 0 \right\}\,.\]
Then, for any admissible parameter $\theta = (\mu, \alpha, \beta, \lambda_0)$ and all $\nu\in\RR$:
\[\begin{dcases}
  {f_\theta^N}_{11}(\nu) = \frac{\mu_1}{1-\alpha_{11}}\frac{\beta_1^2 + 4 \pi^2 \nu^2}{\beta_1^2(1-\alpha_{11})^2 + 4\pi^2\nu^2} + \lambda_0\\
  {f_\theta^N}_{12}(\nu) = \frac{\mu_1}{1-\alpha_{11}}\frac{\beta_1^2 + 4 \pi^2 \nu^2}{\beta_1^2(1-\alpha_{11})^2 + 4\pi^2\nu^2} \frac{\alpha_{21}\beta_2}{\beta_2 - 2\pi \iu \nu}\\
  {f_\theta^N}_{22}(\nu) = \frac{\mu_1}{1-\alpha_{11}}\frac{\beta_1^2 + 4 \pi^2 \nu^2}{\beta_1^2(1-\alpha_{11})^2 + 4\pi^2\nu^2} \frac{\alpha_{21}^2\beta_2^2}{\beta_2^2 + 4\pi^2\nu^2} + \mu_2 + \frac{\mu_1 \alpha_{21}}{1-\alpha_{11}} + \lambda_0\,,
\end{dcases}\]
and ${f_\theta^N}_{21}(\nu) = \overline{{f_\theta^N}_{12}(\nu)}$.

These expressions can be reformulated as follows:

\[\begin{dcases}
  {f_\theta^N}_{11}(\nu) = m_1^H \frac{1}{\lvert 1 - \tilde h_{\theta \, 11}(\nu)\rvert^2} + \lambda_0\\
  {f_\theta^N}_{12}(\nu) = ({f_\theta^N}_{11}(\nu) - \lambda_0)\tilde h_{\theta \, 21}(-\nu)\\
  {f_\theta^N}_{22}(\nu) = \frac{\lvert {f_\theta^N}_{12}(\nu)\rvert^2}{{f_\theta^N}_{11}(\nu) - \lambda_0} + \lim_{\nu' \to +\infty}{f_\theta^N}_{22}(\nu')\,.
\end{dcases}\]
where $m_1^H = {\mu_1}/{(1 - \alpha_{11})}$, $m_2^H =  \mu_2 + \mu_1 \alpha_{21} / ( 1 - \alpha_{11} )$ and $\lim_{\nu' \to +\infty}{f_\theta^N}_{22}(\nu') = m_2^H + \lambda_0$.

Let $\theta' = (\mu', \alpha', \beta', \lambda_0')$ be an admissible parameter such that $\mathbf{f}_{\theta}^N = \mathbf{f}_{\theta'}^N$. For $\nu = 0$, it comes
${f_\theta^N}_{22}(0) = {f_{\theta'}^N}_{22}(0)$, which implies that
\[
	\frac{\lvert {f_\theta^N}_{12}(0)\rvert^2}{{f_\theta^N}_{11}(0) - \lambda_0} = \frac{\lvert {f_{\theta'}^N}_{12}(0)\rvert^2}{{f_{\theta'}^N}_{11}(0) - \lambda_0'}\,.
\]
Since $\lvert {f_\theta^N}_{12}(0)\rvert^2 = {\mu_1 \alpha_{21}}{(1-\alpha_{11})^{-3}}\neq 0$ (as $\alpha_{11} < 1$ and $\alpha_{21} > 0$),
${f_\theta^N}_{12}(0) = {f_{\theta'}^N}_{12}(0)$
and ${f_\theta^N}_{11}(0) = {f_{\theta'}^N}_{11}(0)$,
it results that: \[\lambda_0 = \lambda_0'\,.\]

Now, for all $\nu \in \RR$, since ${f_\theta^N}_{12}(\nu) = {f_{\theta'}^N}_{12}(\nu)$, it comes $({f_\theta^N}_{11}(\nu) - \lambda_0)\tilde h_{\theta \, 21}(-\nu) = ({f_{\theta'}^N}_{11}(\nu) - \lambda_0')\tilde h_{\theta', 21}(-\nu)$, then $\tilde h_{\theta \, 21}(\nu) = \tilde h_{\theta', 21}(\nu)$.
By Equation~\eqref{eq:fourier_exponential}, it results that: 
\[
  \frac{\alpha_{21} \beta_2}{\beta_2 + 2 \pi \iu \nu} = \frac{\alpha_{21}' \beta_2'}{\beta_2' + 2 \pi \iu \nu}\,.
\]

For $\nu = 0$, this simplifies to \[\alpha_{21} = \alpha_{21}'\,.\]
For $\nu = 1$, as $\alpha_{21} \neq 0$ it comes
${\beta_2}/(\beta_2 + 2 \pi \iu) = {\beta_2'}/(\beta_2' + 2 \pi \iu)$, then
\[
  \beta_2 = \beta_2' \,.
\]

Now, by considering the limits as $\nu \to \infty$, we obtain the two following equations:

\[\begin{dcases}
	\frac{\mu_1}{1-\alpha_{11}} + \lambda_0 = \frac{\mu_1'}{1-\alpha_{11}'} + \lambda_0\\
	\mu_2 + \frac{\mu_1 \alpha_{21}}{1-\alpha_{11}} + \lambda_0 = \mu_2' + \frac{\mu_1' \alpha_{21}}{1-\alpha_{11}'} + \lambda_0\,,
\end{dcases}
\]
which, using that $\mu_1 > 0$, can be simplified to:
\[
	\begin{dcases}
	\frac{\mu_1}{1-\alpha_{11}}= \frac{\mu_1'}{1-\alpha_{11}'}\\
	\mu_2 = \mu_2'\,. 
	\end{dcases}
\]

Now, as ${f_\theta^N}_{11}(0) = {\mu_1}/{(1-\alpha_{11})^3}$, we have the two following equations:

\[\begin{dcases}
	\frac{\mu_1}{1-\alpha_{11}}= \frac{\mu_1'}{1-\alpha_{11}'}\\
	\frac{\mu_1}{(1-\alpha_{11})^3} = \frac{\mu_1'}{(1-\alpha_{11}')^3} \,,
	\end{dcases}
\]
which imply
\[
  \begin{dcases}
	\mu_1 = \mu_1'\\
	\alpha_{11} = \alpha_{11}'\,.
	\end{dcases}
\]

At last, since ${f_\theta^N}_{11}(1) = {f_{\theta'}^N}_{11}(1)$ and $\mu_1 > 0$,
\[
  \frac{\beta_1^2 + 4 \pi^2}{\beta_1^2(1-\alpha_{11})^2 + 4\pi^2}
  =
  \frac{{\beta_1'}^2 + 4 \pi^2}{{\beta_1'}^2(1-\alpha_{11})^2 + 4\pi^2} \,,
\]
which implies
\[
  \alpha_{11}(2 - \alpha_{11}) \beta_1^2
  =
  \alpha_{11}(2 - \alpha_{11}) {\beta_1'}^2 \,.
\]

Either $\alpha_{11} > 0$, so the previous equation leads to
\[
  \beta_1 = \beta_1' \,,
\]
or $\alpha_{11} = 0$, so $\alpha_{11}' = \alpha_{11} = 0$ and $\beta_1 = \beta_1' = 1$ (since $\theta$ and $\theta'$ are admissible for the model) .

This proves that $\mathbf{f}_\theta^N = \mathbf{f}_{\theta'}^N \implies \theta = \theta'$, which achieves the proof.

  \section{Proof of Proposition~\ref{PROPOSITION:BI_IDENTIFIABLE}, Situations~\ref{HYP:BI_IDENTIFIABLE_2} and \ref{HYP:BI_IDENTIFIABLE_2_BIS}}\label{appendix:bi_identifiable_2}

Without loss of generality, let
\[\Lambda = \left\{ \begin{pmatrix} \alpha_{11} & 0 \\ \alpha_{21} & \alpha_{22} \end{pmatrix}, 0 < \alpha_{11} < 1, \alpha_{21} > 0 , 0 \le \alpha_{22} < 1 \right\}\,.\]
Then, for any admissible parameter $\theta = (\mu, \alpha, \beta, \lambda_0)$ and all $\nu\in\RR$:
\[\begin{dcases}
  {f_\theta^N}_{11}(\nu) = \dfrac{\mu_1}{1-\alpha_{11}} \dfrac{\beta_1^2 \alpha_{11}(2-\alpha_{11})}{\beta_1^2(1-\alpha_{11})^2 + 4 \pi^2\nu^2} + \dfrac{\mu_1}{1-\alpha_{11}} + \lambda_0\\
  {f_\theta^N}_{12}(\nu) = \left( {f_\theta^N}_{11}(\nu) - \lambda_0 \right) \dfrac{\alpha_{21} \beta_2}{\beta_2^2(1-\alpha_{22})^2 + 4\pi^2\nu^2}\left(\beta_2(1-\alpha_{22}) + 2 \pi \iu \nu\right)\\
  {f_\theta^N}_{22}(\nu) = \left( \dfrac{\mu_1 \alpha_{21}}{(1-\alpha_{22})(1-\alpha_{11})} + \dfrac{\mu_2}{1-\alpha_{22}} \right)\dfrac{\beta_2^2 + 4\pi^2\nu^2}{\beta_2^2(1-\alpha_{22})^2 + 4\pi^2\nu^2} \,+ \\
  \phantom{{f_\theta^N}_{22}(\nu) = } \dfrac{\mu_1 \alpha_{21}^2 \beta_2^2 }{1-\alpha_{11}} \dfrac{\beta_1^2 + 4\pi^2\nu^2}{(\beta_2^2(1-\alpha_{22})^2 + 4\pi^2\nu^2)(\beta_1^2(1-\alpha_{11})^2 + 4\pi^2\nu^2)} + \lambda_0\,,
\end{dcases}
\]
and ${f_\theta^N}_{21}(\nu) = \overline{{f_\theta^N}_{12}(\nu)}$.

Let $\theta' = (\mu', \alpha', \beta', \lambda_0')$ be an admissible parameter such that $\mathbf{f}_{\theta}^N = \mathbf{f}_{\theta'}^N$.
We start by showing that $\lambda_0 = \lambda_0'$.
From $\mathfrak{Re} \left( {f_\theta^N}_{12}(1) \right) = \mathfrak{Re} \left({f_{\theta'}^N}_{12}(1) \right)$ and $\mathfrak{Im} \left( {f_\theta^N}_{12}(1) \right) = \mathfrak{Im} \left( {f_{\theta'}^N}_{12}(1) \right)$,
the following system of equations is established:
\begin{align}\label{eq:sys_re_im}
  \begin{cases}
    \dfrac{({f_\theta^N}_{11}(1) - \lambda_0)\alpha_{21} \beta_2}{\beta_2^2(1-\alpha_{22})^2 + 4\pi^2} \beta_2(1-\alpha_{22}) = \dfrac{({f_{\theta'}^N}_{11}(1) - \lambda_0')\alpha_{21}' \beta_2'}{{\beta_2'}^2(1-\alpha_{22}')^2 + 4\pi^2} \beta_2'(1-\alpha_{22}') \\
    \dfrac{({f_\theta^N}_{11}(1) - \lambda_0)\alpha_{21} \beta_2}{\beta_2^2(1-\alpha_{22})^2 + 4\pi^2} = \dfrac{({f_{\theta'}^N}_{11}(1) - \lambda_0')\alpha_{21}' \beta_2'}{ {\beta_2'}^2(1-\alpha_{22}')^2 + 4\pi^2}\,.
  \end{cases}
\end{align}
Since
\[
  {f_\theta^N}_{11}(1) - \lambda_0 = \frac{\mu_1}{1-\alpha_{11}}\frac{\beta_1^2 + 4 \pi^2}{\beta_1^2(1-\alpha_{11})^2 + 4\pi^2} \neq 0 \,,
\]
and $\alpha_{21} \beta_2 > 0$,
Equations~\eqref{eq:sys_re_im} imply that
\begin{align}\label{eq:E_id_1}
  \beta_2(1-\alpha_{22}) = \beta_2'(1-\alpha_{22}') \,.
\end{align}
Now, from ${f_\theta^N}_{12}(0) = \left( {f_\theta^N}_{11}(0) - \lambda_0 \right) \alpha_{21} / (1-\alpha_{22})$
and ${f_\theta^N}_{11}(0) - \lambda_0 = \mu_1 / (1-\alpha_{11})^3 \neq 0$,
it comes $\alpha_{21} = {f_\theta^N}_{12}(0) (1-\alpha_{22}) / \left( {f_\theta^N}_{11}(0) - \lambda_0 \right)$ and the second equality of Equations~\eqref{eq:sys_re_im} implies that
\begin{align}\label{eq:E_id_2}
  \frac{{f_\theta^N}_{11}(1) - \lambda_0}{{f_\theta^N}_{11}(0) - \lambda_0} \frac{{f_\theta^N}_{12}(0)\beta_2(1-\alpha_{22})}{\beta_2^2(1-\alpha_{22})^2 + 4\pi^2} = \frac{{f_{\theta'}^N}_{11}(1) - \lambda_0'}{{f_{\theta'}^N}_{11}(0) - \lambda_0'} \frac{{f_{\theta'}^N}_{12}(0){\beta_2'}(1-\alpha_{22}')}{{\beta_2'}^2(1-\alpha_{22}')^2 + 4\pi^2}\,.
\end{align}
By Equation~\eqref{eq:E_id_1}
\[
  \frac{{f_\theta^N}_{12}(0)\beta_2(1-\alpha_{22})}{\beta_2^2(1-\alpha_{22})^2 + 4\pi^2} = \frac{{f_{\theta'}^N}_{12}(0){\beta_2'}(1-\alpha_{22}')}{{\beta_2'}^2(1-\alpha_{22}')^2 + 4\pi^2}\,,
\]
and since $\alpha_{21} > 0$, ${f_\theta^N}_{12}(0) \neq 0$ and Equation~\eqref{eq:E_id_2} leads to
\[
  \frac{{f_\theta^N}_{11}(1) - \lambda_0}{{f_\theta^N}_{11}(0) - \lambda_0} = \frac{{f_{\theta'}^N}_{11}(1) - \lambda_0'}{{f_{\theta'}^N}_{11}(0) - \lambda_0'}\,.
\]

Since ${f_\theta^N}_{11}$ is strictly decreasing, ${f_\theta^N}_{11}(1)-{f_\theta^N}_{11}(0) = {f_{\theta'}^N}_{11}(1)-{f_{\theta'}^N}_{11}(0) \neq 0$, and it comes $\lambda_0 = \lambda_0'$.

Now, the expression of ${f_\theta^N}_{11}$ is similar to that of the univariate spectral density $f_\theta^N$ in Proposition~\ref{PROPOSITION:FIXED_UNIVARIATE_IDENTIFIABILITY}.
Thus, following the proof in Appendix~\ref{appendix:identifiability_univariate}, from ${f_\theta^N}_{11} = {f_{\theta'}^N}_{11}$ and $\lambda_0 = \lambda_0'$
(since $\alpha_{11} > 0$ and $\alpha_{11}'> 0$)
it comes:
\[
\begin{dcases}
  \mu_1 = \mu_1' \nonumber\\
  \alpha_{11} = \alpha_{11}'\nonumber\\
  \beta_1 = \beta_1'\nonumber\,.
\end{dcases}
\]

Next, from ${f_\theta^N}_{12}(0) = {f_{\theta'}^N}_{12}(0)$, it comes:
\begin{equation*}
    \left( {f_\theta^N}_{11}(0) - \lambda_0 \right) \frac{\alpha_{21}}{1-\alpha_{22}}
    =
    \left( {f_{\theta'}^N}_{11}(0) - \lambda_0' \right) \frac{\alpha_{21}'}{1-\alpha_{22}'} \,,
\end{equation*}
which implies, with ${f_\theta^N}_{11}(0) - \lambda_0 = {f_{\theta'}^N}_{11}(0) - \lambda_0' \neq 0$:
\begin{equation}\label{eq:bi_last_system}
  \frac{\alpha_{21}}{1-\alpha_{22}} = \frac{\alpha_{21}'}{1-\alpha_{22}'}\,.
\end{equation}

Now, from $\lim_{\nu' \to \infty} {f_\theta^N}_{22}(\nu') = \lim_{\nu' \to \infty} {f_{\theta'}^N}_{22}(\nu')$,
\begin{align*}
	\frac{\mu_1 \alpha_{21}}{(1-\alpha_{11})(1-\alpha_{22})} + \frac{\mu_2}{1-\alpha_{22}} &= \frac{\mu_1 \alpha_{21}'}{(1-\alpha_{11})(1-\alpha_{22}')} + \frac{\mu_2'}{1-\alpha_{22}'} \,,
\end{align*}
and leveraging Equation~\eqref{eq:bi_last_system}, it comes:
\begin{align}\label{eq:bi_last_equation}
  \frac{\mu_2}{1-\alpha_{22}} &= \frac{\mu_2'}{1-\alpha_{22}'}\,.
\end{align}

Moreover,
\begin{align*}
	{f_\theta^N}_{22}(0)
	&= \left( \frac{\mu_1 \alpha_{21}}{(1-\alpha_{11})(1-\alpha_{22})} + \frac{\mu_2}{1-\alpha_{22}} + \frac{\mu_1\alpha_{21}^2}{(1-\alpha_{11})^3} \right) \frac{1}{(1-\alpha_{22})^2} + \lambda_0 \\
	&= \left(
	  \frac{\mu_1}{1-\alpha_{11}} \frac{\alpha_{21}}{1-\alpha_{22}} +
	\frac{\mu_2}{1-\alpha_{22}}
	\right) \frac{1}{(1-\alpha_{22})^2} +
	\frac{\mu_1}{(1-\alpha_{11})^3} \frac{\alpha_{21}^2}{(1-\alpha_{22})^2} + \lambda_0 \,.
\end{align*}

Thus, by Equations~\eqref{eq:bi_last_system} and \eqref{eq:bi_last_equation}, we obtain from ${f_\theta^N}_{22}(0) = {f_{\theta'}^N}_{22}(0)$:
\[
  \frac{1}{(1-\alpha_{22})^2} = \frac{1}{(1-\alpha_{22}')^2} \,,
\]
which implies
\[
  \alpha_{22} = \alpha_{22}' \,.
\]

To conclude, from Equations~\eqref{eq:E_id_1}, \eqref{eq:bi_last_equation} and \eqref{eq:bi_last_system},
$\beta_2 = \beta_2'$, $\mu_2 = \mu_2'$ and $\alpha_{21} = \alpha_{21}'$.
This proves that $\mathbf{f}_\theta^N = \mathbf{f}_{\theta'}^N \implies \theta = \theta'$, which achieves the proof.

\section{Proof of Proposition~\ref{PROPOSITION:BI_IDENTIFIABLE}, final argument}\label{appendix:all_bi_identifiable}
    Let 
    \[
        \begin{array}{l}
            \Lambda_1 = \left\{ \begin{pmatrix} \alpha_{11} & 0 \\ \alpha_{21} & 0 \end{pmatrix}, 0 \le \alpha_{11} < 1, \alpha_{21} > 0 \right\} \\
            \Lambda_2 = \left\{ \begin{pmatrix} 0 & \alpha_{12} \\ 0 & \alpha_{22} \end{pmatrix}, \alpha_{12} > 0, 0 \le \alpha_{22} < 1 \right\} \\
            \Lambda_3 = \left\{ \begin{pmatrix} \alpha_{11} & 0 \\ \alpha_{21} & \alpha_{22} \end{pmatrix}, 0 < \alpha_{11} < 1, \alpha_{21} > 0 , 0 \le \alpha_{22} < 1 \right\} \\
            \Lambda_4 = \left\{ \begin{pmatrix} \alpha_{11} & \alpha_{12} \\ 0 & \alpha_{22} \end{pmatrix}, 0 \le \alpha_{11} < 1, \alpha_{12} > 0, 0 < \alpha_{22} < 1 \right\},
        \end{array}
    \]
    along with \(\mathbf{f}_{\theta}^N \in \mathcal Q_{\cup_{j=1}^4 \Lambda_j}\) and \(\mathbf{f}_{\theta'}^N \in \mathcal Q_{\cup_{j=1}^4 \Lambda_j}\).
    Without loss of generality, we focus on two particular situations, the others being similar by symmetry of all indices.
    
    \subsection{Disjoint submodels}
        If \(\mathbf{f}_{\theta}^N \in \mathcal Q_{\Lambda_1}\) and \(\mathbf{f}_{\theta'}^N \in \mathcal Q_{\Lambda_2}\),
        or \(\mathbf{f}_{\theta}^N \in \mathcal Q_{\Lambda_1}\) and \(\mathbf{f}_{\theta'}^N \in \mathcal Q_{\Lambda_4}\),
        or \(\mathbf{f}_{\theta}^N \in \mathcal Q_{\Lambda_3}\) and \(\mathbf{f}_{\theta'}^N \in \mathcal Q_{\Lambda_4}\),
        proving identifiability boils down to equating rational functions with non-trivially different degrees, which can be shown to be impossible.
        As a consequence, those situations cannot happen.
        
    \subsection{Intersecting submodels}
        Consider the pair \(\Lambda_1\)-\(\Lambda_3\).
        Then, for \(\mathbf{f}_{\theta}^N \in \mathcal Q_{\Lambda_1} \cup \mathcal Q_{\Lambda_3}\), denoting $\theta = (\mu, \alpha, \beta, \lambda_0)$,
        for all $\nu\in\RR$:
        \[\begin{dcases}
          {f_\theta^N}_{11}(\nu) = \dfrac{\mu_1}{1-\alpha_{11}} \dfrac{\beta_1^2 \alpha_{11}(2-\alpha_{11})}{\beta_1^2(1-\alpha_{11})^2 + 4 \pi^2\nu^2} + \dfrac{\mu_1}{1-\alpha_{11}} + \lambda_0\\
          {f_\theta^N}_{12}(\nu) = \left( {f_\theta^N}_{11}(\nu) - \lambda_0 \right) \dfrac{\alpha_{21} \beta_2}{\beta_2^2(1-\alpha_{22})^2 + 4\pi^2\nu^2}\left(\beta_2(1-\alpha_{22}) + 2 \pi \iu \nu\right)\\
          {f_\theta^N}_{22}(\nu) = \left( \dfrac{\mu_1 \alpha_{21}}{(1-\alpha_{22})(1-\alpha_{11})} + \dfrac{\mu_2}{1-\alpha_{22}} \right)\dfrac{\beta_2^2 + 4\pi^2\nu^2}{\beta_2^2(1-\alpha_{22})^2 + 4\pi^2\nu^2} \,+ \\
          \phantom{{f_\theta^N}_{22}(\nu) = } \dfrac{\mu_1 \alpha_{21}^2 \beta_2^2 }{1-\alpha_{11}} \dfrac{\beta_1^2 + 4\pi^2\nu^2}{(\beta_2^2(1-\alpha_{22})^2 + 4\pi^2\nu^2)(\beta_1^2(1-\alpha_{11})^2 + 4\pi^2\nu^2)} + \lambda_0\,,
        \end{dcases}
        \]
        and ${f_\theta^N}_{21}(\nu) = \overline{{f_\theta^N}_{12}(\nu)}$.
        
        Let \(\mathbf{f}_{\theta}^N \in \mathcal Q_{\Lambda_1}\), \(\mathbf{f}_{\theta'}^N \in \mathcal Q_{\Lambda_3}\) and assume that $\mathbf{f}_{\theta}^N = \mathbf{f}_{\theta'}^N$.
        Then, noting $\theta = (\mu, \alpha, \beta, \lambda_0)$ and $\theta' = (\mu', \alpha', \beta', \lambda_0')$,
        from $\mathfrak{Re} \left( {f_\theta^N}_{12}(1) \right) = \mathfrak{Re} \left({f_{\theta'}^N}_{12}(1) \right)$ and $\mathfrak{Im} \left( {f_\theta^N}_{12}(1) \right) = \mathfrak{Im} \left( {f_{\theta'}^N}_{12}(1) \right)$,
        the following system of equations is established:
        \begin{align}\label{eq:sys_re_im2}
          \begin{cases}
            \dfrac{({f_\theta^N}_{11}(1) - \lambda_0)\alpha_{21} \beta_2}{\beta_2^2 + 4\pi^2} \beta_2 = \dfrac{({f_{\theta'}^N}_{11}(1) - \lambda_0')\alpha_{21}' \beta_2'}{{\beta_2'}^2(1-\alpha_{22}')^2 + 4\pi^2} \beta_2'(1-\alpha_{22}') \\
            \dfrac{({f_\theta^N}_{11}(1) - \lambda_0)\alpha_{21} \beta_2}{\beta_2^2 + 4\pi^2} = \dfrac{({f_{\theta'}^N}_{11}(1) - \lambda_0')\alpha_{21}' \beta_2'}{ {\beta_2'}^2(1-\alpha_{22}')^2 + 4\pi^2}\,.
          \end{cases}
        \end{align}
        Since
        \[
          {f_\theta^N}_{11}(1) - \lambda_0 = \frac{\mu_1}{1-\alpha_{11}}\frac{\beta_1^2 + 4 \pi^2}{\beta_1^2(1-\alpha_{11})^2 + 4\pi^2} \neq 0 \,,
        \]
        and $\alpha_{21} \beta_2 > 0$,
        Equations~\eqref{eq:sys_re_im2} imply that
        \begin{align}\label{eq:E_id_12}
          \beta_2 = \beta_2'(1-\alpha_{22}') \,.
        \end{align}
        Now, from ${f_{\theta'}^N}_{12}(0) = \left( {f_{\theta'}^N}_{11}(0) - \lambda_0' \right) \alpha_{21}' / (1-\alpha_{22}')$
        and ${f_{\theta'}^N}_{11}(0) - \lambda_0' = \mu_1' / (1-\alpha_{11}')^3 \neq 0$,
        it comes $\alpha_{21}' = {f_{\theta'}^N}_{12}(0) (1-\alpha_{22}') / \left( {f_{\theta'}^N}_{11}(0) - \lambda_0' \right)$ (similarly \(\alpha_{21} = {f_{\theta}^N}_{12}(0) / \left( {f_{\theta}^N}_{11}(0) - \lambda_0 \right)\)) and the second equality of Equations~\eqref{eq:sys_re_im2} implies that
        \begin{align}\label{eq:E_id_22}
          \frac{{f_\theta^N}_{11}(1) - \lambda_0}{{f_\theta^N}_{11}(0) - \lambda_0} \frac{{f_\theta^N}_{12}(0)\beta_2}{\beta_2^2 + 4\pi^2} = \frac{{f_{\theta'}^N}_{11}(1) - \lambda_0'}{{f_{\theta'}^N}_{11}(0) - \lambda_0'} \frac{{f_{\theta'}^N}_{12}(0){\beta_2'}(1-\alpha_{22}')}{{\beta_2'}^2(1-\alpha_{22}')^2 + 4\pi^2}\,.
        \end{align}
        By Equation~\eqref{eq:E_id_12}
        \[
          \frac{{f_\theta^N}_{12}(0)\beta_2}{\beta_2^2 + 4\pi^2} = \frac{{f_{\theta'}^N}_{12}(0){\beta_2'}(1-\alpha_{22}')}{{\beta_2'}^2(1-\alpha_{22}')^2 + 4\pi^2}\,,
        \]
        and since $\alpha_{21} > 0$, ${f_\theta^N}_{12}(0) \neq 0$ (similarly for \({f_{\theta'}^N}\)) and Equation~\eqref{eq:E_id_22} leads to
        \[
          \frac{{f_\theta^N}_{11}(1) - \lambda_0}{{f_\theta^N}_{11}(0) - \lambda_0} = \frac{{f_{\theta'}^N}_{11}(1) - \lambda_0'}{{f_{\theta'}^N}_{11}(0) - \lambda_0'}\,.
        \]
        
        Since ${f_\theta^N}_{11}$ is strictly decreasing, ${f_\theta^N}_{11}(1)-{f_\theta^N}_{11}(0) = {f_{\theta'}^N}_{11}(1)-{f_{\theta'}^N}_{11}(0) \neq 0$, and it comes $\lambda_0 = \lambda_0'$.
        
        Now, the expression of ${f_\theta^N}_{11}$ is similar to that of the univariate spectral density $f_\theta^N$ in Proposition~3, except that \(\alpha_{11}\) may be null.
        Either \(\alpha_{11} = 0\), then \({f_{\theta'}^N}_{11} = {f_\theta^N}_{11} = \mu_1 + \lambda_0\) is constant, which is impossible since \(\mu_1' {\beta_1'}^2 \alpha_{11}'(2-\alpha_{11}') \neq 0\).
        Or \(\alpha_{11} \neq 0\), then following the proof in Appendix~B, from ${f_\theta^N}_{11} = {f_{\theta'}^N}_{11}$ and $\lambda_0 = \lambda_0'$
        (since $\alpha_{11} > 0$ and $\alpha_{11}'> 0$)
        it comes:
        \[
        \begin{dcases}
          \mu_1 = \mu_1' \nonumber\\
          \alpha_{11} = \alpha_{11}'\nonumber\\
          \beta_1 = \beta_1'\nonumber\,.
        \end{dcases}
        \]
        
        Next, from ${f_\theta^N}_{12}(0) = {f_{\theta'}^N}_{12}(0)$, it comes:
        \begin{equation*}
            \left( {f_\theta^N}_{11}(0) - \lambda_0 \right) \alpha_{21}
            =
            \left( {f_{\theta'}^N}_{11}(0) - \lambda_0' \right) \frac{\alpha_{21}'}{1-\alpha_{22}'} \,,
        \end{equation*}
        which implies, with ${f_\theta^N}_{11}(0) - \lambda_0 = {f_{\theta'}^N}_{11}(0) - \lambda_0' \neq 0$:
        \begin{equation}\label{eq:bi_last_system2}
          \alpha_{21} = \frac{\alpha_{21}'}{1-\alpha_{22}'}\,.
        \end{equation}
        
        Now, from $\lim_{\nu' \to \infty} {f_\theta^N}_{22}(\nu') = \lim_{\nu' \to \infty} {f_{\theta'}^N}_{22}(\nu')$,
        \begin{align*}
        	\frac{\mu_1 \alpha_{21}}{1-\alpha_{11}} + \mu_2 &= \frac{\mu_1 \alpha_{21}'}{(1-\alpha_{11})(1-\alpha_{22}')} + \frac{\mu_2'}{1-\alpha_{22}'} \,,
        \end{align*}
        and leveraging Equation~\eqref{eq:bi_last_system2}, it comes:
        \begin{align}\label{eq:bi_last_equation2}
          \mu_2 &= \frac{\mu_2'}{1-\alpha_{22}'}\,.
        \end{align}
        
        Moreover,
        \begin{align*}
        	{f_{\theta'}^N}_{22}(0)
        	&= \left(
        	  \frac{\mu_1'}{1-\alpha_{11}'} \frac{\alpha_{21}'}{1-\alpha_{22}'} +
        	\frac{\mu_2'}{1-\alpha_{22}'}
        	\right) \frac{1}{(1-\alpha_{22}')^2} +
        	\frac{\mu_1'}{(1-\alpha_{11}')^3} \frac{{\alpha_{21}'}^2}{(1-\alpha_{22}')^2} + \lambda_0' \,.
        \end{align*}
        Thus, equation ${f_\theta^N}_{22}(0) = {f_{\theta'}^N}_{22}(0)$ leads to:
        \[
            \left(
        	  \frac{\mu_1}{1-\alpha_{11}} \alpha_{21} +
        	\mu_2
        	\right) +
        	\frac{\mu_1}{(1-\alpha_{11})^3} \alpha_{21}^2
        	= 
        	\left(
        	  \frac{\mu_1'}{1-\alpha_{11}'} \frac{\alpha_{21}'}{1-\alpha_{22}'} +
        	\frac{\mu_2'}{1-\alpha_{22}'}
        	\right) \frac{1}{(1-\alpha_{22}')^2} +
        	\frac{\mu_1'}{(1-\alpha_{11}')^3} \frac{{\alpha_{21}'}^2}{(1-\alpha_{22}')^2} \,,
        \]
        which in turn implies:
        \[
          1 = \frac{1}{(1-\alpha_{22}')^2} \,,
        \]
        meaning that \(\alpha_{22}' = 0\).
        Thus, it comes that \((\alpha_{11}, \alpha_{21}, \alpha_{22}) = (\alpha_{11}', \alpha_{21}', \alpha_{22}')\),
        \((\beta_1, \beta_2) = (\beta_1', \beta_2')\) and \((\mu_1, \mu_2) = (\mu_1', \mu_2')\),
        \ie \(\theta = \theta'\).

\section{Additional experiment}
\label{appendix:additional_experiment}

			    In this section, we present an additional numerical experiment in order to challenge the ability of our proposed procedure for detecting the support of interactions for any set of true parameters. It consists in generating a random set of true parameters for a bivariate noisy Hawkes process,
			    the interaction matrix having different sparsity structures,
			    from which we generate a single realisation with 15000 event times. We partition the corresponding simulation window $[0,T]$ in 10 equally-sized subwindows and estimate the parameters in each subsample assuming the full model. The true parameters are generated as follows:
			    \begin{itemize}
			        \item 
			        $\beta_1, \beta_2$ and $\lambda_0$ are simulated independently such that:
			        \begin{align*}
                        \beta_1 &\sim \text{Exp}(1/2) + 0.5\\
                        \beta_2 &\sim \text{Exp}(1/2) + 0.5\\
                        \lambda_0 &\sim \text{Exp}(1/2) + 0.1\,.
			        \end{align*}
			        The minimal values (\(0.5\) for \(\beta_1\) and \(\beta_2\), and \(0.1\) for \(\lambda_0\)) are chosen such that the effects from the past in the Hawkes process can be correctly detected and that there is a minimal effect of noise in the observations.
			        \item $\mu_1$ and \(\mu_2\) are simulated independently such that:
			        \begin{align*}
                        \mu_1 &\sim \text{Exp}(1/2)\\
                        \mu_2 &\sim \text{Exp}(1/2)\,,\\
			        \end{align*}
			        and we only keep realisations such that $0.5 < \mu_1/\mu_2 < 2.0$,
			        in such a way that the subprocesses have comparable numbers of child points.
			        \item the matrix $\alpha$ is simulated along a spike-and-slab procedure:
			        for every $1 \le i,j \le 2$,
			        \begin{equation*}
                        \alpha_{ij} = X_{ij} U_{ij}\,,\quad \text{where $X_{ij}\sim \text{Exp}(1/2)$ and $U_{ij}\sim \text{Binomial}(1, 2/3)$
                        independently from the rest.}\\
			        \end{equation*}
			        This corresponds to setting each interaction parameter $\alpha_{ij}$ equal to 0 with probability $1/3$, otherwise it is simulated according to an exponential distribution.
			        Additionally, we set to 0 any interaction such that $\alpha_{ij} < 0.1$, in order to only consider effects that are fairly detectable.
			        This allows us to explore a wide variety of interaction networks in our estimations.
			    \end{itemize}

			    We carried out this experiment 100 times and used the subsampling scheme to estimate the support for each generated dataset. For each realisation, we consider that $\hat \alpha_{ij}$ is a null estimation if there is at least $30\%$ of the subsampled estimations estimated equal to $0$ (similarly to what we propose in Remark~\ref{rem:subsampling} in the paper, and in the above Table~\ref{tab:partition_quantile}). 
			    
			   Figure~\ref{fig:accuracy_noise} represents the accuracy of support estimation if we consider the simulations with a noise $\lambda_0 < \lambda_{max}$ for different levels of maximal noise $\lambda_{max}$. For example, if we consider only the simulations with a noise level $\lambda_0 < 2.8$, we reach a level of accuracy of approximately $77\%$, and this is obtained with a single subsampled simulation. Among all simulations, the support of the true parameter $\alpha$ is estimated correctly $64\%$ of the time and we can see from Figure~\ref{fig:accuracy_noise} that the accuracy of support estimation is particularly good for lower levels of noise. This additional experiment allows us to demonstrate the performance of our method for a wider range of scenarios than those considered in Section \ref{sec:bivariate_numerical_results}.

		    \begin{figure}
			    \centering
	    \includegraphics[width=.8\linewidth]{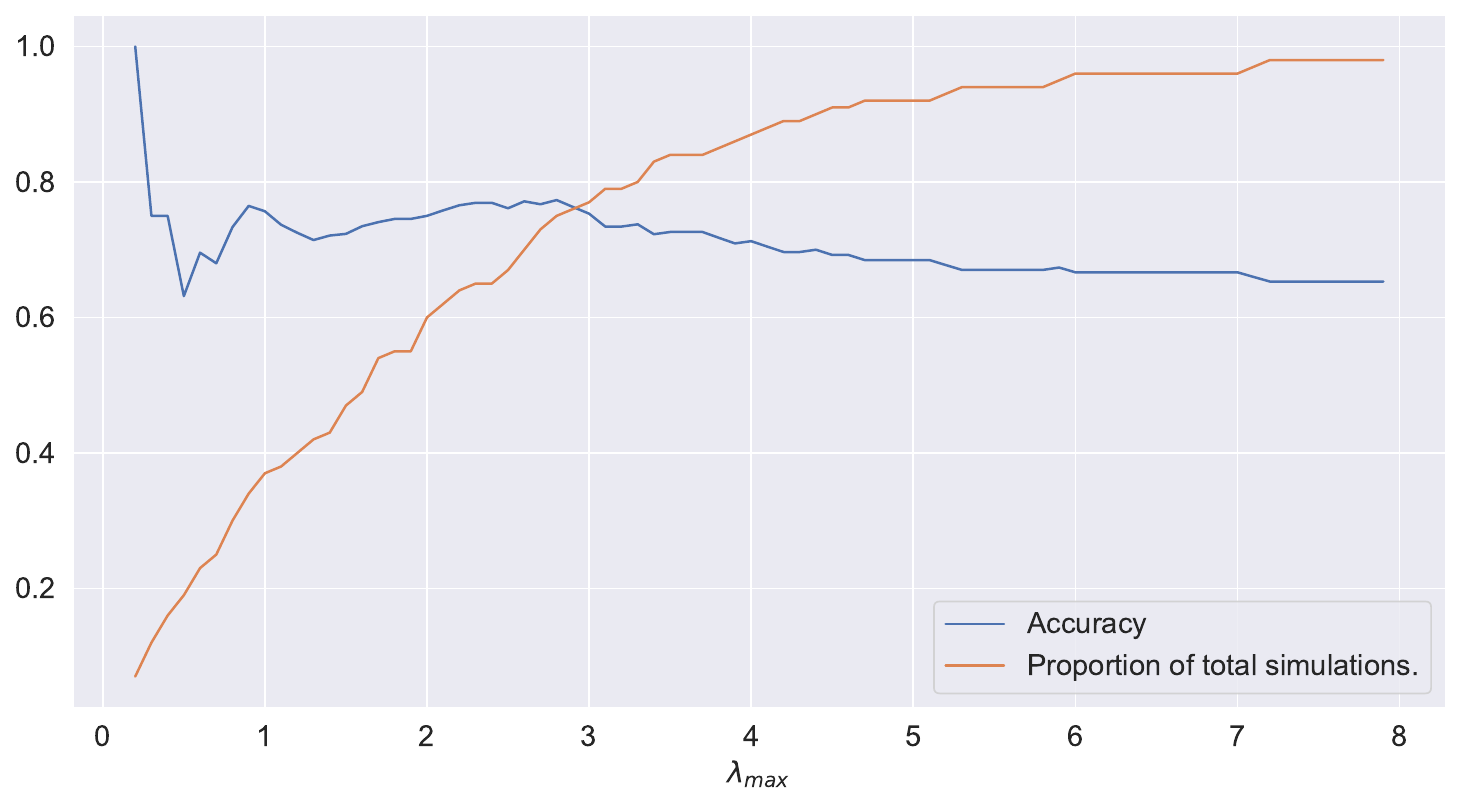}
			    \caption{In blue, accuracy of support estimation for simulations with a level of noise $\lambda_0$ less than $\lambda_{max}$. In orange, the corresponding proportion of simulations considered among the 100 experiments.}
			    \label{fig:accuracy_noise}
			    \end{figure}

\end{document}